\documentclass[12pt]{article}
\usepackage{amsmath}
\usepackage{graphicx,psfrag,epsf}
\usepackage{enumerate}
\usepackage{natbib}
\usepackage{url} 

\usepackage{amssymb,amsthm}
\usepackage{multirow, array}
\usepackage[colorlinks,citecolor=blue,urlcolor=blue]{hyperref}

\newcommand{\blind}{0}

\newtheorem{definition}{Definition}
\newtheorem{corollary}{Corollary}
\newtheorem{theorem}{Theorem}

\newcommand{\R}{\mathbb{R}}
\newcommand{\x}{\mathbf{x}}
\newcommand{\y}{\mathbf{y}}
\newcommand{\z}{\mathbf{z}}
\newcommand{\e}{\mathbf{e}}
\newcommand{\g}{\mathbf{g}}
\newcommand{\mean}{\mu}
\newcommand{\covkern}{\mathfrak{K}}
\newcommand{\Psup}{P^{\sup}}
\newcommand{\Pinf}{P^{\inf}}
\newcommand{\E}{\mathbb{E}}
\newcommand{\doe}{X}

\newcommand{\Tide}{T}
\newcommand{\Surge}{S}
\newcommand{\PhiVar}{t_o}
\newcommand{\tPlus}{t_+}
\newcommand{\tMinus}{t_-}
\newcommand{\threshold}{\tau}
\newcommand{\npostGPsim}{s}
\newcommand{\meters}{\text{ m}}

\begin{document}

	\def\spacingset#1{\renewcommand{\baselinestretch}%
		{#1}\small\normalsize} \spacingset{1}

		
	\if0\blind
	{
		\title{\bf Profile extrema for visualizing and quantifying uncertainties on excursion regions. Application to coastal flooding}
		
		\author{Dario Azzimonti\footnotemark[1] \footnotemark[2]\hspace{.2cm}\\
			Istituto Dalle Molle di studi sull'Intelligenza Artificiale (IDSIA) \\
			Galleria 2, Via Cantonale 2c, 6928 Manno, Switzerland,\\
			and \\
			David Ginsbourger\footnotemark[1] \\
			 UQOD, Idiap Research Institute \\
			 Centre du Parc, Rue Marconi 19, 1920 Martigny, Switzerland; and\\
			 IMSV, Department of Mathematics and Statistics, University of Bern \\
			 Alpeneggstrasse 22, 3012 Bern, Switzerland,\\
			 and \\
			J\'er\'emy Rohmer\footnotemark[1] \footnotemark[3]\hspace{.2cm}\\
			BRGM \\
			3 Av. C. Guillemin, BP 36009, 45060 Orl\'eans Cedex 2, France,\\
			and \\
			D\'eborah Idier\footnotemark[1] \footnotemark[3]\hspace{.2cm}\\
			BRGM \\
			3 Av. C. Guillemin, BP 36009, 45060 Orl\'eans Cedex 2, France.}
		
			\renewcommand{\thefootnote}{\fnsymbol{footnote}}
			\footnotetext[1]{All authors thank F. Paris  and R. Pedreros (BRGM) for providing the Boucholeurs flooding model.}
			\footnotetext[2]{The author acknowledges the financial support of the \textit{Hasler Foundation, grant number 16065} and the \textit{SNSF, grant number 167199}.} 
			\footnotetext[3]{The authors would like to thank the OQUAIDO chair in applied mathematics and acknowledge financial support from BRGM via the PROSUB project}			
			\renewcommand{\thefootnote}{\arabic{footnote}}
		\date{}
		\maketitle
	} \fi
	
	\if1\blind
	{
		\bigskip
		\bigskip
		\bigskip
		\begin{center}
			{\LARGE\bf Profile extrema for visualizing and quantifying uncertainties on excursion regions. Application to coastal flooding}
		\end{center}
		\medskip
	} \fi

	\bigskip
	\newpage
	\begin{abstract}
		We consider the problem of describing excursion sets of a real-valued function $f$,~i.e. the set of inputs where $f$ is above a fixed threshold. Such regions are hard to visualize if the input space dimension, $d$, is higher than 2. For a given projection matrix from the input space to a lower dimensional (usually $1,2$) subspace, we introduce profile sup (inf) functions that associate to each point in the projection's image the sup (inf) of the function constrained over the pre-image of this point by the considered projection. Plots of profile extrema functions convey a simple, although intrinsically partial, visualization of the set. We consider expensive to evaluate functions where only a very limited number of evaluations, $n$, is available, e.g.~$n<100d$, and we surrogate $f$ with a posterior quantity of a Gaussian process (GP) model. We first compute profile extrema functions for the posterior mean given $n$ evaluations of $f$. We quantify the uncertainty on such estimates by studying the distribution of GP profile extrema with posterior quasi-realizations obtained from an approximating process. We control such approximation with a bound inherited from the Borell-TIS inequality. The technique is applied to analytical functions ($d=2,3$) and to a $5$-dimensional coastal flooding test case for a site located on the Atlantic French coast. Here $f$ is a numerical model returning the area of flooded surface in the coastal region given some offshore conditions. Profile extrema functions allowed us to better understand which offshore conditions impact large flooding events.  
		
	\end{abstract}

	
	\noindent%
	\vfill
	
	\newpage
	\spacingset{1} 
	\section{Introduction}
	\label{sec:intro}

\subsection{Problem statement}
In many domains of application (structural reliability, \citet{dubourg2013metamodel}, nuclear safety, \citet{Chevalier.etal2014},  environmental problems, \citet{BolinLindgren2014},  natural risk, \citet{RohmerIdier2012,bayarri.etal2009using})  
the region in the parameter space leading to response values above or below a given threshold is of crucial interest. 
In particular here, we assume that a phenomenon can be modelled by a function $f: D \subset \R^d \rightarrow \R$, where $D$ is a compact set and we focus on excursion sets of the form 
\begin{equation*}
\varGamma = \{ \x \in D : f(\x) \geq \threshold \},
\end{equation*}
where $\threshold\in \R$ is a prescribed threshold depending on the application at hand. This work presents a method to gain visual insights on $\varGamma$ when the input space dimension $d \geq 3$. 

Visualization of point clouds and manifolds in high dimensions is a very active field of research. Principal components analysis and its kernel version \citep{Scholkopf1997} play a fundamental role in dimensionality reduction for high dimensional data. Parallel coordinates \citep{Inselberg1985,Inselberg2009} are widely used for exploratory data analysis. In order to visualize more complex structure, techniques relying on the topological concept of Morse-Smale complex \citep{Edelsbrunner2008,Gerber.etal2010} provide powerful, but hard to interpret, tools. Projection based techniques, in particular tours~\citep[see, e.g., ][]{Asimov1985,Cook_etal1995,Cook2008},  are also a prominent method to visualize data in high-dimensional Euclidean spaces. 
All these techniques, however, are based on a finite set of data points while, in our case, we provide a representation that exploits the structure of $\varGamma$.

The contribution of this paper is two-fold: first we show how to use profile extrema functions for visualizing excursion sets of known deterministic functions; second, we implement these tools for expensive computer experiments with Gaussian process (GP) emulation. 

For a given projection matrix $\Psi \in \R^{d\times p}$, a profile sup (inf respectively) is a function $\Psup_\Psi f \ $ $(\Pinf_\Psi f)$ that associates to each point $\eta$ in the image of $\Psi^T$, the $\sup$ ($\inf$ respectively) of $f$ over all points with projection equal to $\eta$, i.e. $\Psup_\Psi f (\eta) = \sup_{\Psi^T \x = \eta} f(\x)$.  
For example, if $\Psi^T = [1, 0, \ldots, 0] \in \R^{1 \times d}$ is the canonical projection over the first coordinate, then $\Psup_\Psi f(\eta)$ is the $\sup$ of $f$ over all points in $D$ with first coordinate equal to $\eta$. In this case we can plot the values of $\Psup_\Psi f$ and $\Pinf_\Psi f$ over their $1$-dimensional ($p$ dimensional, in general) domain and study such functions. In particular, all $\tilde{\eta}$ such that $\Psup_\Psi f(\tilde{\eta}) <\threshold$,~i.e. all inputs $\x \in D$ such that $\Psi^T\x = \tilde{\eta}$, identify regions of non-excursion. Analogously $\Pinf_\Psi f(\eta)>\threshold$ selects regions of excursion. We compute profile extrema for known, fast-to-evaluate functions $f$ with numerical optimization, as described in Section~\ref{sec:profiles}.

 In this paper, however, we focus on the case where the function $f$ is very expensive to evaluate and thus the number of evaluations, $n$, is limited,~e.g. $n < 100d$. Following the classical computer experiments literature \citep[see, e.g.,][]{Sacks.etal1989,Santner2013design}, we define a design of experiments $\doe_n = (\x_1, ,\ldots, \x_n ) \in D^n$ and denote with $\y_n= (f(\x_1), \ldots, f(\x_n))^T$ the corresponding vector of evaluations. We assume that $f$ is a realization of a GP $(Z_\x)_{\x \in D} \sim GP(\mean,\covkern)$ with prior mean function $\mean: D \subset \R^d \rightarrow \R$ and covariance kernel $\covkern: D \times D \rightarrow \R$. 
The set $\varGamma$ can be estimated from the posterior GP distribution \citep[see,~e.g.][]{Bingham.etal2014,BolinLindgren2014,Azzimonti2016}, in particular, here we consider the plug-in estimate
\begin{equation*}
\hat{\varGamma}_n = \{\x \in D : \gamma_n(\x) \geq \threshold\},
\end{equation*}
where $\gamma_n$ is a posterior quantity related to $Z$. Examples of interesting quantities $\gamma_n$ are the posterior mean, the $1-\alpha$ posterior quantiles or a posterior realization of $Z$. 
If $d=1,2$, a visualization of $\hat{\varGamma}_n$ only requires evaluations of $\gamma_n$ on a grid; however, when $d\geq 3$, this approach is problematic. We propose a technique to compute and plot profile extrema functions for $\gamma_n$ and we detail how to use profile extrema to identify regions of interest. 

A GP model not only provides an estimate for $\varGamma$, but it also allows for a quantification of the associated uncertainty. In particular, here we study the distribution of the stochastic object $\Psup_\Psi Z$ and we provide point-wise confidence statements on profile extrema functions. The quantities $\Psup_\Psi Z,\ \Pinf_\Psi Z$ are strongly related to extreme values of Gaussian processes indexed by compact sets which have been widely studied in probability theory \citep[see, e.g., ][]{Adler.Taylor2007,Azais.Wschebor2009} and have been the subject of recent interest in the computer experiments literature, see, e.g.~\citet[][Chapter~6]{Chevalier2013} and \citet{Ginsbourger_etal2014}. Here we obtain point-wise confidence intervals for profile extrema functions by using posterior quasi-realizations of an approximating process obtained as in~\citet{Azzimonti.etal2016}. Moreover we introduce a probabilistic bound for the quantiles of the posterior profile function based on the sample quantiles of the approximating process. This procedure enabled the identification of regions of interests for the coastal flooding test case presented in Section~\ref{sec:motivating}. 

\subsection{Outline of the paper}
In Section~\ref{sec:profiles} we introduce profile extrema functions for a generic function $\gamma$ and we apply them to a synthetic test case. In this section, we also discuss implementation details. Section~\ref{sec:UQ} describes the GP model and the quantification of uncertainty on the profiles with approximate posterior GP simulations. Moreover we provide a probabilistic bound (Theorem~\ref{theo:bound} and Corollary~\ref{cor:bounds}, proofs in Appendix~\ref{sec:proofs}) for this approximation. In Section~\ref{sec:motivating} we apply the procedure on a coastal flooding problem where $f$ is evaluated with an expensive-to-evaluate computer experiment. Further details are in Appendix~\ref{sec:full5dRes}. Finally, in Section~\ref{sec:discussion}, we discuss advantages and drawbacks of the current method and propose possible extensions.

\section{Profile extrema functions}
\label{sec:profiles}

\subsection{Definitions}
\label{subsec:defProf}
In this section, we introduce the concept of profile extrema for a continuously differentiable function $\gamma: D \rightarrow \R$, $D \subset \R^ d$ compact. While $\gamma$ can be any function, in this paper we mainly consider $\gamma$ as the posterior mean, a quantile or a posterior realization of a GP. The first example of profile extrema are coordinate profiles defined below,~i.e. $2 d$ univariate functions describing the extremal behavior of $\gamma$ for each coordinate. 
\begin{definition}[Coordinate profile extrema]
	For each $i \in \{1, \ldots, d \}$, let us denote with $\e_i$ the $i$-th canonical coordinate vector, then the \emph{i-th coordinate profile extrema} functions are 
	\begin{align}
	\Psup_{\e_i} \gamma = \Psup_i \gamma &: \ \eta \in E_{\e_i} \longrightarrow \  \Psup_i \gamma(\eta) := \sup_{x_i = \eta}\gamma(\x) \\
	\Pinf_{\e_i} \gamma = \Pinf_i \gamma &: \ \eta \in E_{\e_i} \longrightarrow \ \Pinf_i \gamma(\eta) := \inf_{x_i = \eta} \gamma(\x).
	\end{align}
	where $E_{\e_i} = \{ \eta \in \R : \e_i^T \x = \eta, \ \x \in D \}$ and $\x = (x_1, \ldots, x_i, \ldots, x_d)$.
\end{definition}

By studying these functions we can split the input space $D$ in simple regions providing a marginal visualization of $\varGamma$. Consider the first coordinate profile sup function $\Psup_1 \gamma$. If there exists $\eta^*$ such that $\Psup_1 \gamma(\eta^*) < \threshold$, the function $\gamma$ is below $\threshold$ for all $\x \in D_{i,\eta^ *}= \{ \x \in D : x_i = \eta^* \}$. In a symmetric way, if $\eta^*$ is such that $\Pinf_1 \gamma(\eta^*) \geq \threshold$, the function $\gamma$ is above $\threshold$ over the whole $D_{i,\eta^*}$.  Since $\Psup_1$ and $\Pinf_1$ are univariate functions we can evaluate them over a discretization of $E_{\e_1}$ and with a simple plot we can locate those regions. Profile functions could also lead to more undetermined situations: for example, for $\tilde{\eta}$ such that $\Psup_1 \gamma(\tilde{\eta}) \geq \threshold$, there exists at least one point $\bar{\x}^{\tilde{\eta}} = (\tilde{\eta}, \bar{x}_2, \ldots, \bar{x}_d) \in D_{1,\tilde{\eta}}$ such that $\gamma(\bar{\x}^{\tilde{\eta}}) \geq \threshold$ and we cannot flag the region $D_{1,\tilde{\eta}}$ neither as excursion nor as non-excursion. 
Coordinate profile extrema are the results of constrained optimization problems over particular projections on the coordinate axes. Profile extrema functions further generalize this concept. 

\begin{definition}[Profile extrema]
	Consider a full column rank matrix $\Psi \in \R^{d \times p}$ with $p < d$. Let $E_\Psi := \{\Psi^T \x : \x \in D \}$ and $D_{\Psi,\eta} := \{ \x \in D : \Psi^T \x = \eta \}$, then the functions 
	\begin{align}
	\Psup_\Psi \gamma &: \eta \in E_\Psi \longrightarrow \  \Psup_\Psi \gamma(\eta) := \sup_{\x \in D_{\Psi, \eta}}\gamma(\x) \\
	\Pinf_\Psi \gamma &: \eta \in E_\Psi \longrightarrow \ \Pinf_\Psi \gamma(\eta) := \inf_{\x \in D_{\Psi, \eta}}\gamma(\x), 
	\end{align}
	are called \emph{Profile sup} and \emph{Profile inf} functions of $\gamma$ respectively. The set $E_\Psi$ is the image of $D$ under $\Psi^T$ and $D_{\Psi,\eta}$ is the pre-image of $\eta$ under $\Psi^T$ restricted to $D$.
	\label{def:profExtrema}
\end{definition}

Since we are interested in visualization here we only consider $p \in \{1,2\}$ and $\Psi$ as orthogonal projections because they have a more direct interpretation.  

\subsection{Analytical example}
\label{subsec:analExample}

Let us review the concepts just introduced on an analytical test case. Consider
\begin{equation}
\gamma(\mathbf{x}) = \sin\left( a \mathbf{v_1}^T \mathbf{x} +b\right) + \cos\left(c \mathbf{v_2}^T \mathbf{x} +d \right) \qquad \mathbf{x} \in [0,1]^2, \ \mathbf{v_1}, \mathbf{v_2}  \in \mathbb{R}^2
\label{eq:analExample}
\end{equation}
where $a,b,c,d \in \R$, $\mathbf{v_1} = [\cos(\theta),\sin(\theta)]^T$, $\mathbf{v_2} = [\cos(\theta+\pi/2),\sin(\theta+\pi/2)]^T$ and $\theta= \pi/6$. 

%
%
%
%
%

Figure~\ref{fig:analyticalFun} shows contour lines of $\gamma$ with parameters chosen as $a=1,b=0,c=10,d=0$. We fix a threshold $\threshold=0$ and we compute the coordinate profile extrema, plotted in Figure~\ref{fig:coordAnalytic}. In this case, $\Psup_1 \gamma(\eta) <\threshold$ for $\eta \in [0,0.13]$ and $\Psup_2 \gamma(\eta) <\threshold$ for $\eta \in [0,0.25]\cup [0.70,0.80]$. The second coordinate seems more informative on the excursion: $\Psup_1 \gamma$ flags as non-excursion a region with volume $0.13$ while $\Psup_2 \gamma$ flags a region with volume $0.35$. The true set has volume $0.127$, i.e. the region of non excursion has volume $0.873$. This example shows that canonical directions might not be the most appropriate ones to visualize this function. Figure~\ref{fig:obliqueAnalytic} shows profile extrema computed along the (oblique) generating directions $\mathbf{v_1},\mathbf{v_2}$, i.e. profile extrema in Definition~\ref{def:profExtrema} with $\Psi= \mathbf{v_1}$ or $\mathbf{v_2}$. By using oblique directions we can say that for $\x$ such that $\mathbf{v_1}^T\x \in [0,0.52] \cup [1.22,1.37]$ (area excluded $0.33$) and for $\x$ such that $\mathbf{v_2}^T \x \in [-0.5,-0.1] \cup [0.11,0.54] \cup [0.71,0.87]$ (area excluded $0.68$) there is no excursion. As shown in Figure~\ref{fig:anObliqueFun} profile extrema along this direction allow us to exclude a much larger portion (total area excluded $0.79$ versus $0.44$ with coordinate profiles) of the input space. In Appendix~\ref{sec:bivEx3d} we combine coordinate, oblique and bivariate profiles extrema for a finer visualization of the excursion set on a $3$-dimensional analytical function inspired by~\eqref{eq:analExample}. 

%
%
%
%
%
%
%
%

\begin{figure}[t]
	\begin{minipage}{0.49\textwidth}
		\centering
		\includegraphics[width=\linewidth]{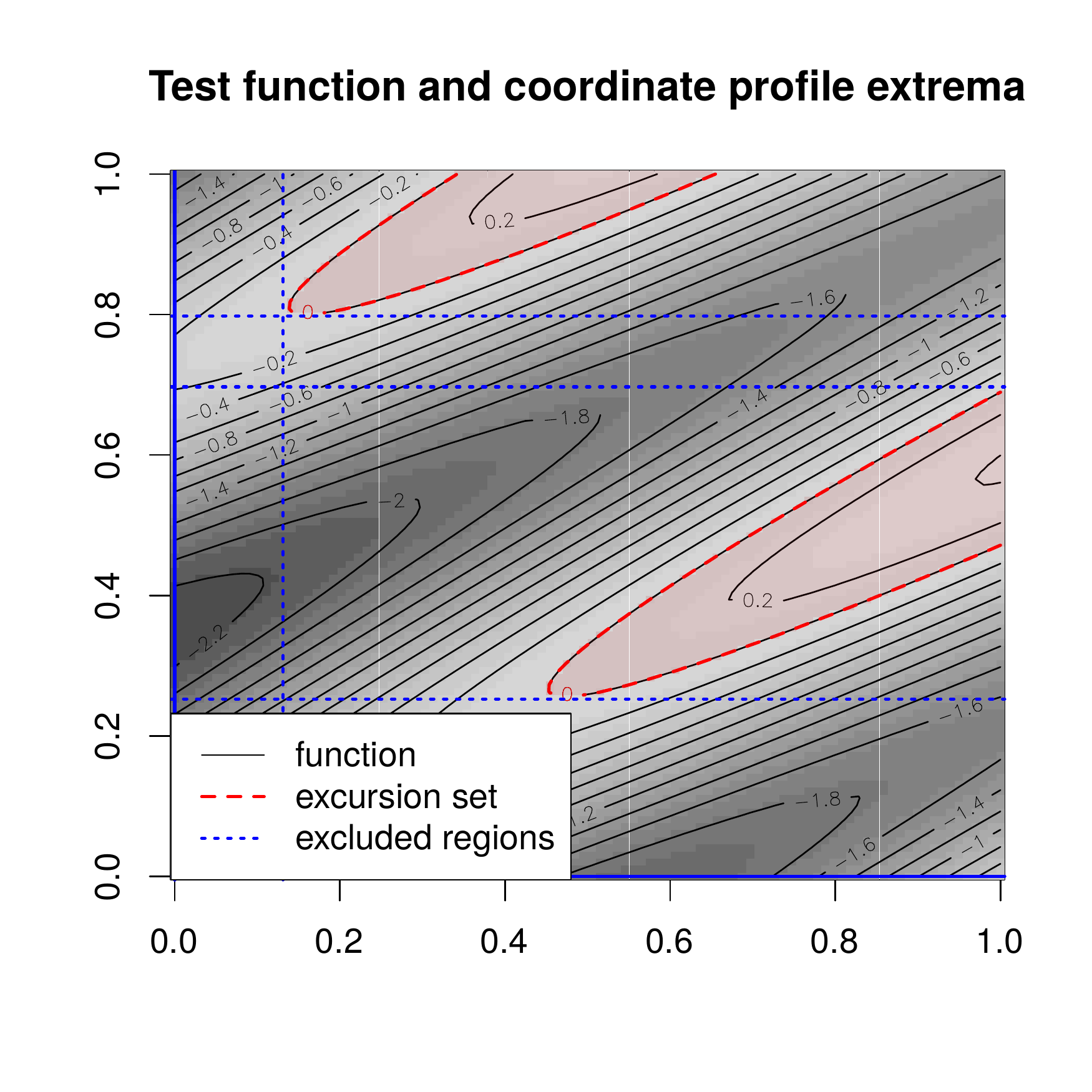}
		\caption{Example of $2$-dimensional test case in equation~\eqref{eq:analExample}. Threshold $\threshold$ in red dashed. The regions excluded with coordinate profiles are delimited by dotted blue lines.}
		\label{fig:analyticalFun}	
	\end{minipage} \hfill\hspace{0.05cm}
	\begin{minipage}{0.49\textwidth}
		\centering
		\includegraphics[width=\linewidth]{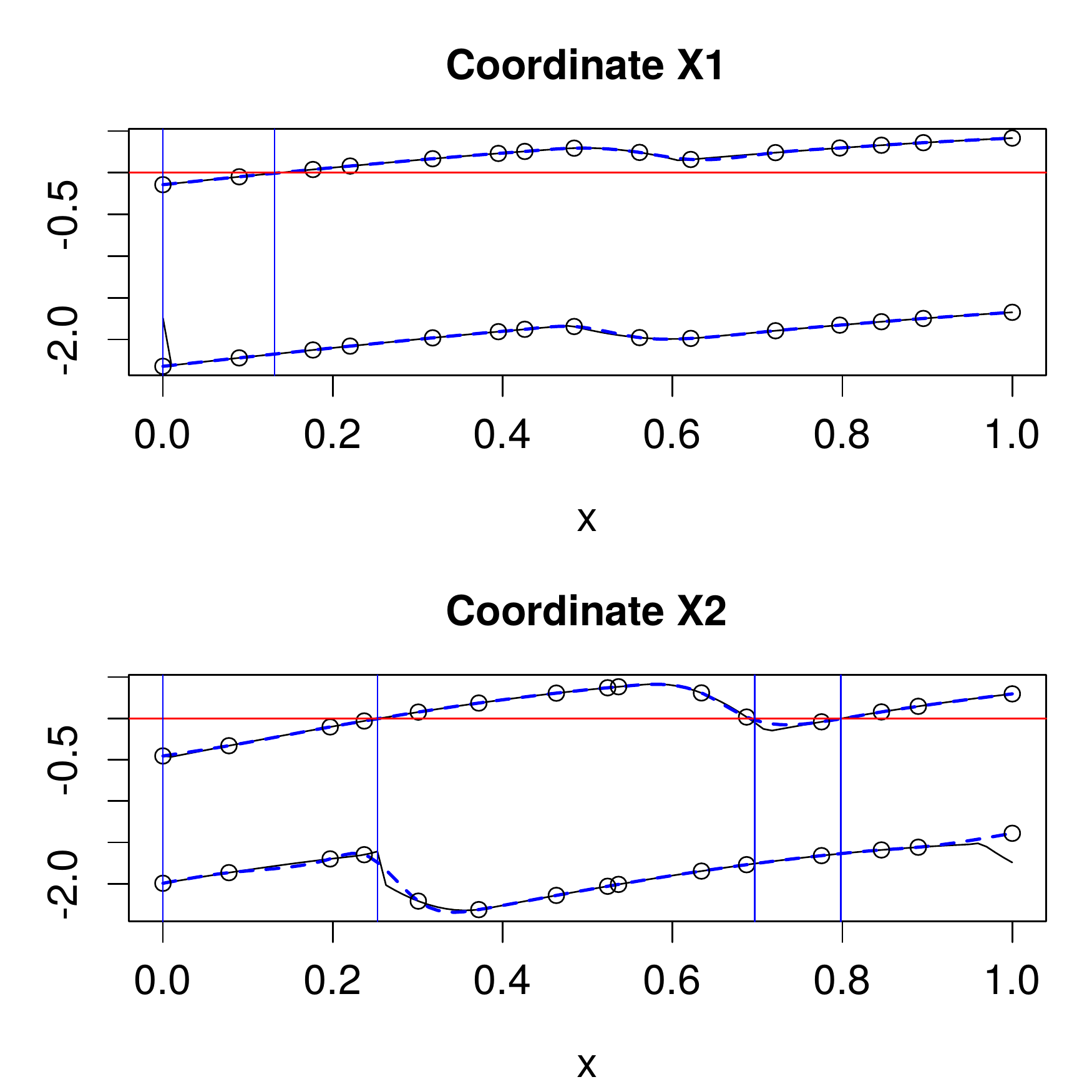}
		\caption{Profile extrema functions for $\gamma$ (black solid) and their approximation (blue dashed) from $k=15$ points (circles). Threshold $\threshold=0$, horizontal solid line.}
		\label{fig:coordAnalytic}
	\end{minipage}	
\end{figure}

\begin{figure}[h!]
	\begin{minipage}{0.49\textwidth}
		\centering
		\includegraphics[width=\linewidth]{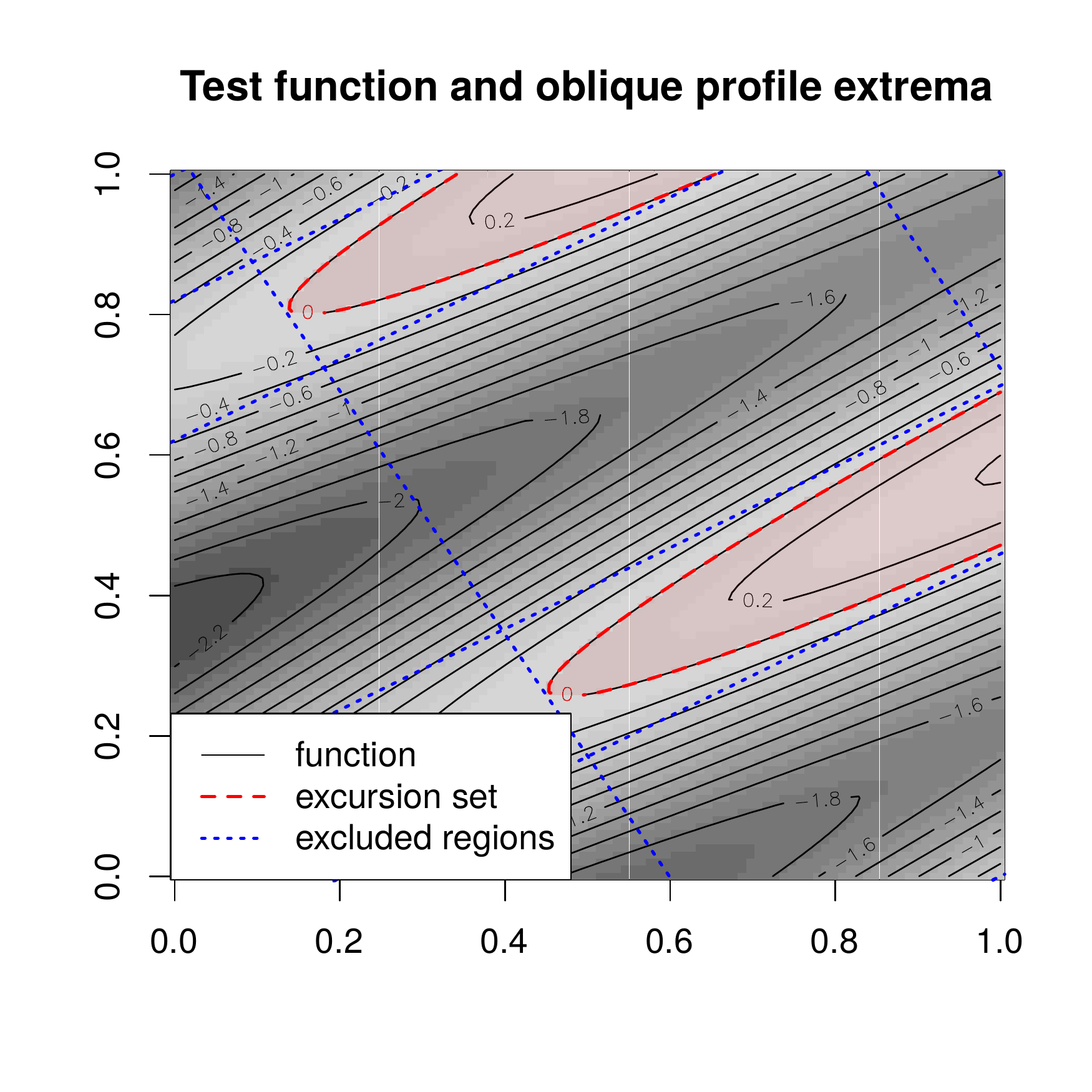}
		\caption{Example of $2$-dimensional test case in equation~\eqref{eq:analExample}. Regions excluded by oblique profile extrema delimited by dotted lines.}
		\label{fig:anObliqueFun}	
	\end{minipage} \hfill\hspace{0.05cm}
	\begin{minipage}{0.49\textwidth}
		\centering
		\includegraphics[width=\linewidth]{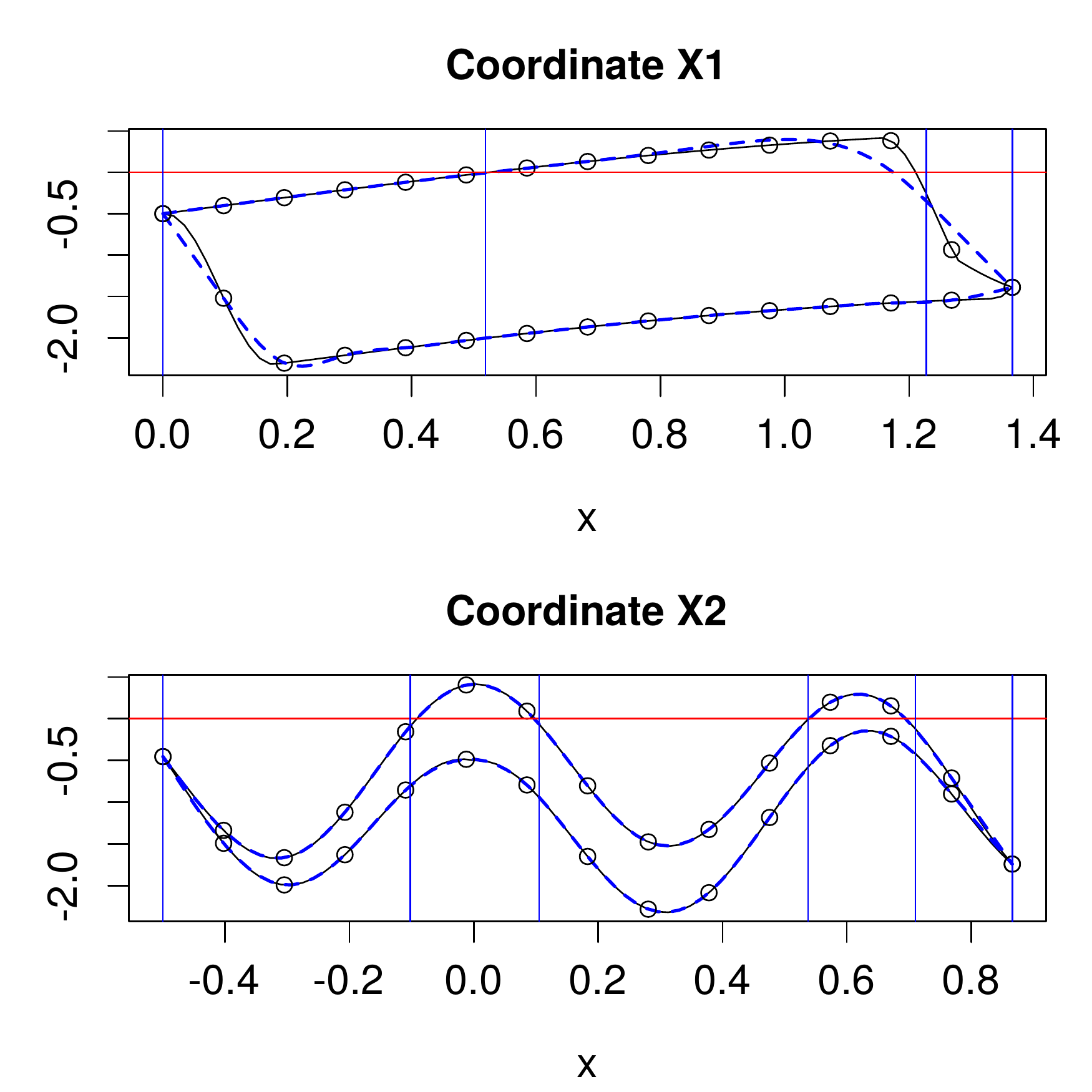}
		\caption{Oblique profile extrema functions for $\gamma$ (black solid lines) and their approximation (blue dashed lines) from $k=15$ points.}
		\label{fig:obliqueAnalytic}
	\end{minipage}	
\end{figure}

\subsection{Implementation}
\label{subsec:implementation}
The value of the profile sup $\Psup_{\Psi} \gamma(\eta)$, for a fixed matrix $\Psi \in \R^{d\times p}$, is the solution of a constrained optimization problem. In general, $\gamma$ is not convex, therefore we are not guaranteed to achieve a global optimum. We implement a local method to obtain profile extrema functions for $\Psi =\e_i \in \R^d, i \in \{1,\ldots,d\}$, and for a matrix $\Psi \in \R^{d\times p}$ of rank $p$.

In the first case we obtain the value for $\Psup_i \gamma$, for a coordinate $i$ with the following procedure. For each $\eta$, define a new function $\tilde{\gamma}_{i,\eta}: \R^{d-1} \rightarrow \R$ such that $\tilde{\gamma}_i(\z) := \gamma(z_1, \ldots,z_{i-1},\eta,z_{i}, \ldots, z_{d-1})$ for any $\z = (z_1, \ldots, z_{d-1}) \in D_{-i}$, where $D_{-i}$ is the $d-1$ dimensional restriction of $D$ without the $i$th coordinate. We can estimate the value $\Psup_i \gamma(\eta)$ with a L-BFGS-B algorithm. If the gradient of $\gamma$ is available it is used in the optimization, otherwise numerical derivatives are used. This method depends on the starting point chosen for the optimization. In general we are interested in the evaluation of the functions $\Psup_i \gamma(\eta)$ for all $\eta$ and for all $i \in \{1, \ldots, d\}$. By exploiting the smoothness of $\gamma$ we re-use the previously obtained points of optimum as starting point for the subsequent optimizations. 

In the second case, we assume that the matrix $\Psi$ has full column rank and that $D$ is a hyper-cube. The value $\Psup_{\Psi} \gamma(\eta)$ is the solution of the constrained optimization problem
\begin{equation}
\text{maximize }  \gamma(\x) \qquad \text{ subject to }  \Psi^T\x = \eta, \ \x \in D.
\label{eq:optFormulation}
\end{equation}
We transform~\eqref{eq:optFormulation} in an unconstrained optimization problem first by computing an arbitrary $\xi \in \R^d$ such that $\Psi^T\xi = \eta$ and then by observing that if $\x$ satisfies the equality constraint then $\x = \xi + \operatorname{Null}(\Psi^T)\z$, where $\operatorname{Null}(\Psi^T) \in \R^{d \times (d-p)}$ is the matrix representing the null space of $\Psi^T$ and $\z\in \R^{d-p}$. We can then maximize $\gamma$ over the $d-p$ dimensional space using a (hard) barrier function for the hyper-cube inequality constraints. Both $\xi$ and the starting point(s) for the interior point method are computed by solving auxiliary linear programs.  

In both cases, a multi-start optimization procedure provides a more reliable estimate of the global optimum at an increased computational cost. 
The two methods are implemented in the R programming language \citep{Rcore} in the package \verb|profExtrema|, available on CRAN. The first method uses the L-BFGS-B \citep{Byrd.etal1995} implementation in the function \verb|optim|, base package; the second method solves the auxiliary linear programs with \verb|lpcdd| from the package \verb|rcdd| (a R interface for \verb|cddlib|, \citet{Fukuda_cddlib}) and optimizes the barrier function with the BFGS implementation in \verb|optim|.

\subsection{Approximation of profile extrema functions}
\label{subsec:approxProf}

The cost of evaluating the profile extrema functions at each $\eta$ increases as $d - p$ grows larger. Plots of profile extrema functions with $\Psi\in \R^{d\times p}$ with $p=1$ or $2$ require evaluations over $E_\Psi \subset \R^p$. If the input dimension $d$ is high, each evaluation is a high dimensional optimization and visualizing the function at sufficient resolution becomes expensive. We mitigate this issue by approximating $\Psup_\Psi \gamma (\eta)$ at any $\eta$ with an interpolation of profile function evaluations at few space filling points $\eta_1, \ldots, \eta_k \in E_\Psi$. 

Coordinate profile extrema functions can be seen as a special type of optimal value functions. First order characterizations of such objects have been studied in optimization and perturbation theory, see, e.g., \citet{Danskin1967, BonnansShapiro1998}. In particular, consider the $i$-th coordinate profile extrema function $\Psup_i \gamma$. We can write $\gamma: D_{\mid -i} \times D_{\mid i} \rightarrow \R$ and $\Psup_i \gamma(\eta) = \sup_{y\in D_{\mid -i}} \gamma(y,\eta)$. If $\gamma(y,\cdot)$ is differentiable  with $\nabla_\eta \gamma(y,\cdot)$ continuous for $y \in D_{\mid -i}$ and there exists a unique $y^* = \arg \max \gamma(y,\eta)$, then $\nabla_\eta \Psup_i \gamma(\eta) = \nabla_\eta \gamma(\cdot,\cdot)_{\arrowvert_{(y^*,\eta)}}$. 
We can then evaluate $\Psup_{i} \gamma(\eta)$ and $\nabla_\eta\eta \Psup_{i}\gamma(\eta)$ at few space filling points $\eta_1, \ldots, \eta_k$, for a small $k>0$, and interpolate the value $\Psup_{i} \gamma(\eta)$ at any $\eta \in E_\Psi$ using a first order approximation. 
For arbitrary $\Psi$, we approximate $\Psup_\Psi \gamma (\eta)$ at any $\eta$ with cubic splines or kriging interpolators as they do not require $\nabla_\eta \Psup_\Psi$.

Below we present an example for the case $p=1$ with a cubic spline approximation. 
%
%
%
The blue dashed lines in Figures~\ref{fig:coordAnalytic},~\ref{fig:obliqueAnalytic} show the approximation of coordinate and oblique profile extrema functions, the exact functions are also plotted in black solid lines. We select $k=15$ points $\eta_1, \ldots, \eta_k \in E_{\e_i}$ with Latin hyper-cube sampling, the approximate profiles at any point are computed with a cubic spline regression from $k$ exact calculations. Exact coordinate profiles $\Psup_i,\Pinf_i$, for $i=1,2$ over a grid of equispaced $100$ points on $[0,1]$ require $0.25$ seconds and oblique profiles $\Psup_{\mathbf{v_i}},\Pinf_{\mathbf{v_i}}$, $i=1,2$ require $0.64$ seconds on a laptop with 2.6 GHz Intel Core i5-7300U CPU. On the same laptop, cubic spline approximations plotted on the same grid require $0.04$ seconds for coordinate profiles ($0.13$ oblique). On this example, the maximum absolute median errors achieved with the approximation of $\Pinf \gamma$, $\Psup \gamma$ are $0.07 \%$ and $0.54 \%$ for coordinate profiles and $0.19 \%$ and $5.08 \%$ for oblique profiles.  
%
%
%
%
The number of approximation points, $k$, is an important parameter for the approximation, the choice of $k$ boils down to a trade-off between computational time and precision. For 1-dimensional profile extrema ($p=1$) the default choice implemented is $k=10\sqrt{d}$. When the profile functions are smooth this heuristics often leads to good results, however it is generally better to compare different values for $k$. Cubic splines are implemented here only for $p=1$, for $p=2$ we use kriging. 

\section{Uncertainty quantification for profile extrema}
\label{sec:UQ}

\subsection{Profile extrema for expensive-to-evaluate functions}
\label{subsec:introGP}

In the previous section we introduced profile extrema functions to visualize excursion sets of arbitrary functions. This method, however, requires performing potentially very high numbers of function evaluations which can be prohibitive in many computer experiments, as, for example, in the motivating application of Section~\ref{sec:motivating}. We therefore rely on a  Gaussian process emulator based on few evaluations of the expensive function. We consider a design of experiments (DoE) with $n$ points $\doe_n \in D^n$ and the corresponding function evaluations $\mathbf{y}_n$ as described in Section~\ref{sec:intro}. We select a prior GP $(Z_{\x})_{\x \in D} \sim GP(\mean,\covkern)$ with mean $\mean(\x) = \sum_{j=1}^m c_j h_j(\x)$, where $\mathbf{c} = (c_1, \ldots, c_m)$ is a vector of unknown coefficients and $h_j$ are known basis functions. The covariance kernel is typically selected from a stationary parametric family, e.g.~Mat\'ern family, and its hyper-parameters are estimated with maximum likelihood. Here we use the R package \verb|DiceKriging| \citep{Roustant.etal2012} for fitting and prediction. Given the data and hyper-parameters the posterior mean and covariance kernel can be computed with the following universal kriging (UK) equations.
\begin{align}
\label{eq:postMean}
\mean_n(\x) &= \mathbf{h}(\x)^T\widehat{\mathbf{c}} + \covkern(\x, \doe_n)\covkern(\doe_n,\doe_n)^{-1} \left(\y_n - {H}\hat{\mathbf{c}} \right) \\ \label{eq:postCov}
\covkern_n(\x,\x^\prime) &= \covkern(\x,\x^\prime) - \covkern(\x, \doe_n)\covkern(\doe_n, \doe_n)^{-1}\covkern(\x^\prime, \doe_n)^T + \\ \nonumber
&\boldsymbol{\lambda}(\x,\doe_n)({H}^T\covkern(\doe_n, \doe_n)^{-1}{H})^{-1}\boldsymbol{\lambda}(\x^\prime,\doe_n)^T, \ \text{for } \x,\x^\prime \in D \qquad \text{with} \\ \nonumber
\boldsymbol{\lambda}(\x,\doe_n) &= (\mathbf{h}(\x)^T - \covkern(\x, \doe_n)\covkern(\doe_n, \doe_n)^{-1}{H}) \qquad \text{ and} \\ \nonumber
\widehat{\mathbf{c}} &= ({H}^T\covkern(\doe_n, \doe_n)^{-1}{H})^{-1}{H}^T\covkern(\doe_n, \doe_n)^{-1}\y_n
\end{align}
where $\mathbf{h}(\x) = (h_1(\x), \ldots, h_m(\x))^T$, ${H}= [h_j(\x_i)]_{i=1,\ldots,n, j=1,\ldots, m}$,  $\covkern(\x,\doe_n) = (\covkern(\x,\x_1), \ldots, \covkern(\x,\x_n) )$, for $\x \in D$ and $\covkern(\doe_n,\doe_n) = (\covkern(\x_i,\x_j))_{i,j=1, \ldots, n} \in \R^{n \times n}$. If $m=1$ and $h_1 \equiv 1$ we have ordinary kriging (OK) equations.

Once we have fitted the GP, profile extrema functions can be directly computed with the posterior mean, leading to a visualization of a plug-in estimate $\hat{\varGamma} = \{ \x \in D : \mean_n(x) \geq \threshold \}$ for $\varGamma$. The GP assumption allows us to also quantify the uncertainty on profile extrema functions, at fixed hyper-parameters. In particular, in the next section, we develop a Monte Carlo method that exploits approximate posterior realizations to obtain point-wise confidence intervals for $\Psup_{\e_i} f$ and $\Pinf_{\e_i} f$, for $i=1, \ldots, d$.

\subsection{Approximation for posterior field realizations}
\label{subsec:approxReals}

Consider the GP $(Z_\x)_{\x \in D}$, we denote with $Z_\x(\omega)$ a realization of the process at $\x$, for $\omega \in \Omega$, the underlying probability space. We are interested in posterior realizations of $Z_\x$ given $Z_{\doe_n} = \mathbf{y}_n$. A brute force approach would involve simulating $\npostGPsim$ times the field over a space filling design on $D$ and then evaluating the profile extrema function by taking the discrete extrema. This procedure however is very costly as it requires $\npostGPsim$ exact simulations of the field over many points and its results strongly depend on the chosen discretization. Following~\citet{Azzimonti.etal2016} we use an approximating process in order to reduce the computational cost. We consider the sequence of points $G= (\g_1, \ldots, \g_\ell) \in D^\ell$ and the \textit{approximating process of $Z$ based on $G$}
\begin{equation}
\widetilde{Z}_\x = a(\x) + \mathbf{b}^T(\x) Z_G, \qquad \x \in D, 
\label{eq:tildeZ}
\end{equation}
where $a: D \rightarrow \R$ is a trend function and $\mathbf{b}: D \rightarrow \R^\ell$ is a continuous vector-valued function of deterministic weights, and $Z_G = (Z_{\g_1}, \ldots, Z_{\g_\ell})$ is the $\ell$-dimensional random vector given by the values of the original process $Z$ at $G$. Here we select $a(\x) = \mean(\x)$, $\mathbf{b}(\x) = \covkern(G,G)^{-1}\covkern(\x,G)^T$ where $\covkern(\x,G) = (\covkern(\x,\g_1), \ldots, \covkern(\x,\g_l)) \in \R^{1\times \ell}$ and $\covkern(G,G) = [\covkern(\g_i,\g_j)]_{i,j =1, \ldots, \ell} \in \R^{\ell \times \ell}$. In what follows, the points $G$ are called \emph{pilot points}, borrowing the term from geostatistics~\citep[see,e.g.][Chapter~4.2]{Scheidt2006}, and here they are selected with Algorihtm~B from~\citet{Azzimonti.etal2016}. The number of pilot points $\ell$ can be empirically chosen by stopping when the optimum of Algorithm~B's objective function stabilizes around a value. In moderate dimensions, usually $\ell$ between $50$ and $150$ leads to distributions close to the true one. The bound introduced in section~\ref{subsec:bound} and, in particular, the related tightness indicator, equation~\eqref{eq:integratedVarDiff}, can also be used to choose the number of pilot points. For a fixed $\omega \in \Omega$,  $\widetilde{Z}_\x(\omega)$ is a function of $\x$ that has an analytical expression  and only requires posterior simulations of the original field $Z$ at $G$ and linear operations with pre-calculated kriging weights. Moreover the gradient of $\widetilde{Z}_\x(\omega)$ with respect to $\x$ is known, see Appendix~\ref{sec:gradTildeZ}, therefore we compute profile extrema for each realization, $\Psup_\Psi \widetilde{Z}_\x(\omega)$ and $\Pinf_\Psi \widetilde{Z}_\x(\omega)$, with the algorithms introduced in Section~\ref{subsec:implementation}.

\subsection{Analytical example under uncertainty}
\label{subsec:analExUQ}


Let us consider the analytical example illustrated in Section~\ref{subsec:analExample}. We evaluate the function in Equation~\eqref{eq:analExample} only at $n$ points and we approximate it on a dense grid $100\times 100$ with a GP. We consider two DoE with $n=20$ and $n=90$, both obtained with (randomized) maximin Latin hypercube sampling (R package \verb|DiceDesign|, \citet{Dupuy_dDesign_2015}) and we choose a prior GP with constant mean and Mat\`ern covariance kernel with $\nu=5/2$. We estimate the covariance hyper-parameters with maximum likelihood and the constant mean value with the ordinary kriging formulae. The model with $n=20$ evaluations provides a rough approximation for $f$, while the model with $n=90$ is an accurate reconstruction. The $Q^2$ criterion\footnote{$Q^2 = 1 - \sum_{j=1}^{N_{test}} (\mean_{n}(x_j) - y^{test}_j)^2 / (\sum_{j=1}^{N_{test}} (y^{test}_j - \overline{y^{test}} )^2)$ with $\overline{y^{test} }= 1/{N_{test}}\sum_{j=1}^{N_{test}} y^{test}_j$} computed over a dense grid of $N_{test}=10000$ test data is equal to $0.86$ for the posterior mean with $n=20$ and $0.99$ for $n=90$.  	
For both models, we consider $\ell=80$ pilot points and we compute the point-wise confidence intervals for the profile extrema on the mean with Monte Carlo simulations. 
Figure~\ref{fig:ex2dmean_20} shows the posterior GP mean obtained from $n=20$ evaluations of the function in Equation~\eqref{eq:analExample}, the regions delimited by the oblique profile extrema on the mean along the directions $\mathbf{v_1},\mathbf{v_2}$ and its $90\%$ point-wise confidence intervals for the threshold $\threshold= 0$. Figure~\ref{fig:ex2dprof_20} shows the profile extrema for the mean, the profile extrema for $\npostGPsim=150$ approximate realizations of the posterior process and the empirical point-wise $90\%$ confidence intervals (dark shaded tube, red). Figure~\ref{fig:ex2dprof_90} shows profile extrema for the mean on the model with $n=90$ observations. In this case the confidence intervals (dark shaded tube, red) are much tighter indicating a smaller uncertainty.

\begin{figure}
	\begin{minipage}{0.5\textwidth}
		\centering
		\includegraphics[width=\linewidth]{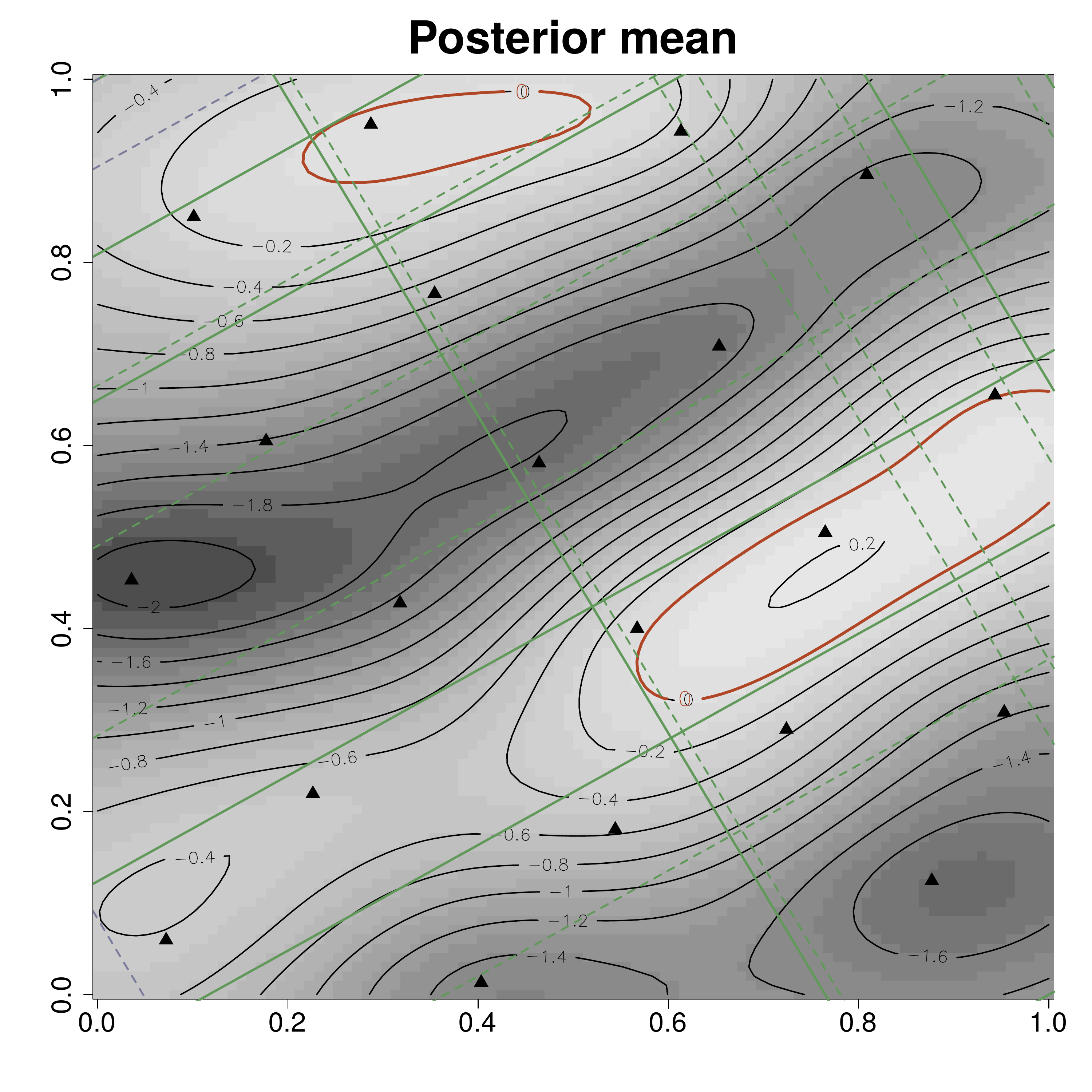}
		\caption{Posterior mean of GP trained on $n=20$ evaluations of function in equation~\eqref{eq:analExample}. Excursion region ($\threshold=0$, shaded red), regions selected by the profile extrema on mean (solid) and on $95\% - 5\%$ quantiles (dashed).}
		\label{fig:ex2dmean_20}	
	\end{minipage} \hfill\hspace{0.1cm}
	\begin{minipage}{0.5\textwidth}
		\centering
		\includegraphics[width=\linewidth]{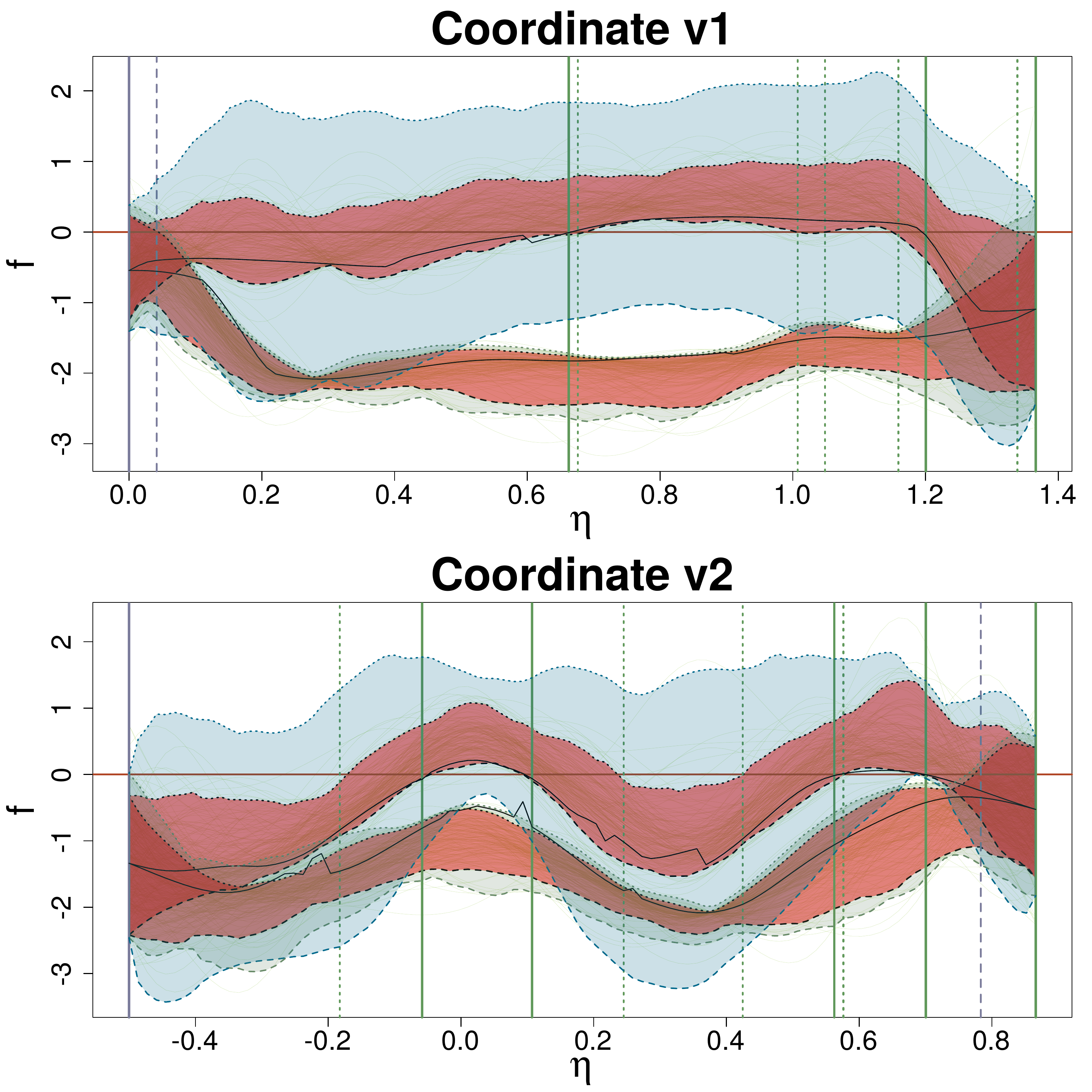}
		\caption{{Profile extrema functions ($n=20$, $\threshold=0$): mean (black, solid), breakpoints (vertical green; mean: solid, quantiles: dashed), quantiles ($90 \%$, black dashed, CI dark red tube), upper-lower bound (light blue tube).}}
		\label{fig:ex2dprof_20}
	\end{minipage}	
\end{figure}

\begin{figure}[h!]
	\begin{minipage}{0.5\textwidth}
		\centering
		\includegraphics[width=\linewidth]{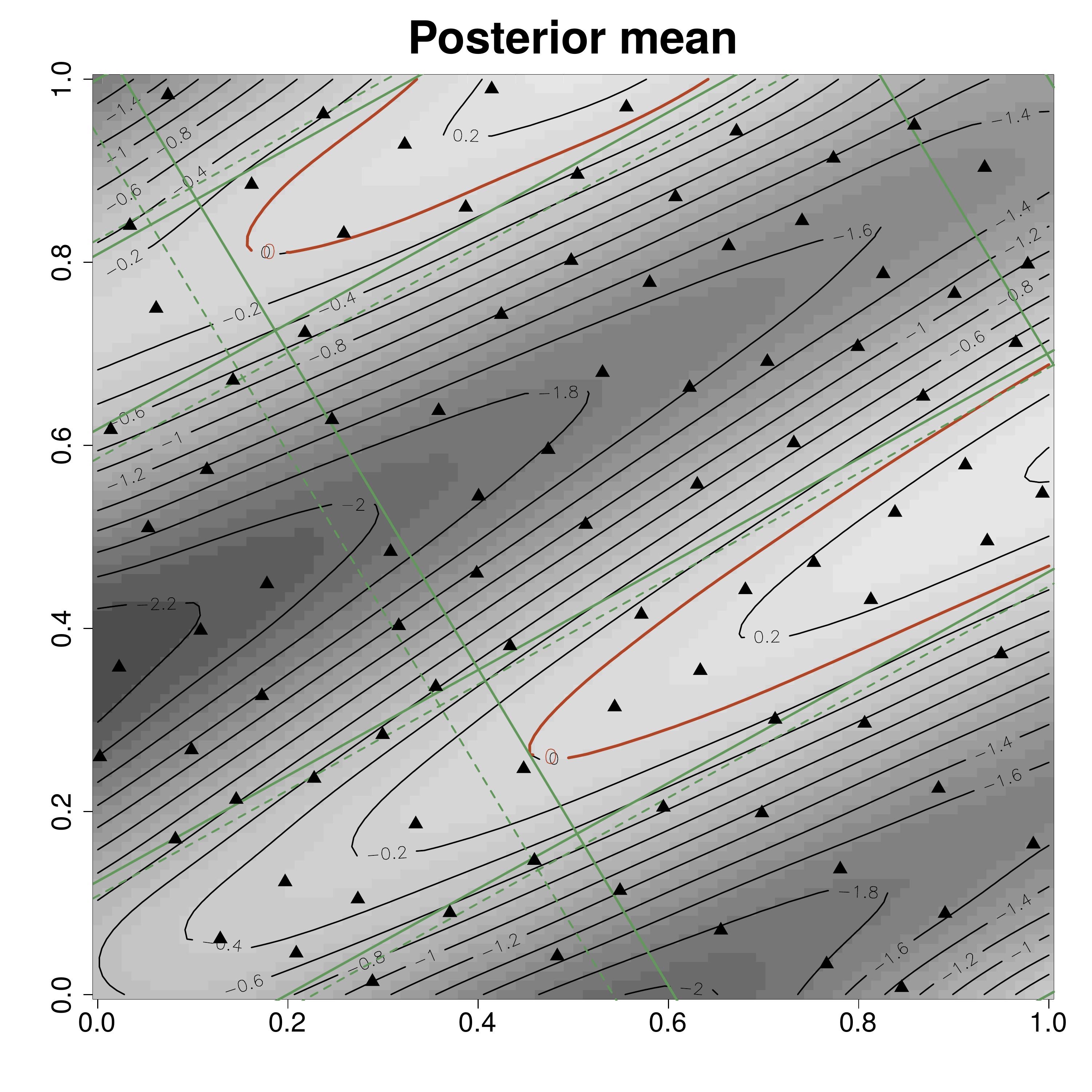}
		\caption{Posterior mean of GP trained on $n=90$ evaluations of function in equation~\eqref{eq:analExample}. Excursion region ($\threshold=0$, shaded red), regions selected by the profile extrema on mean (solid) and on $95\% - 5\%$ quantiles (dashed).}
		\label{fig:ex2dmean_90}	
		\end{minipage} \hfill\hspace{0.1cm}
		\begin{minipage}{0.5\textwidth}
			\centering
			\includegraphics[width=\linewidth]{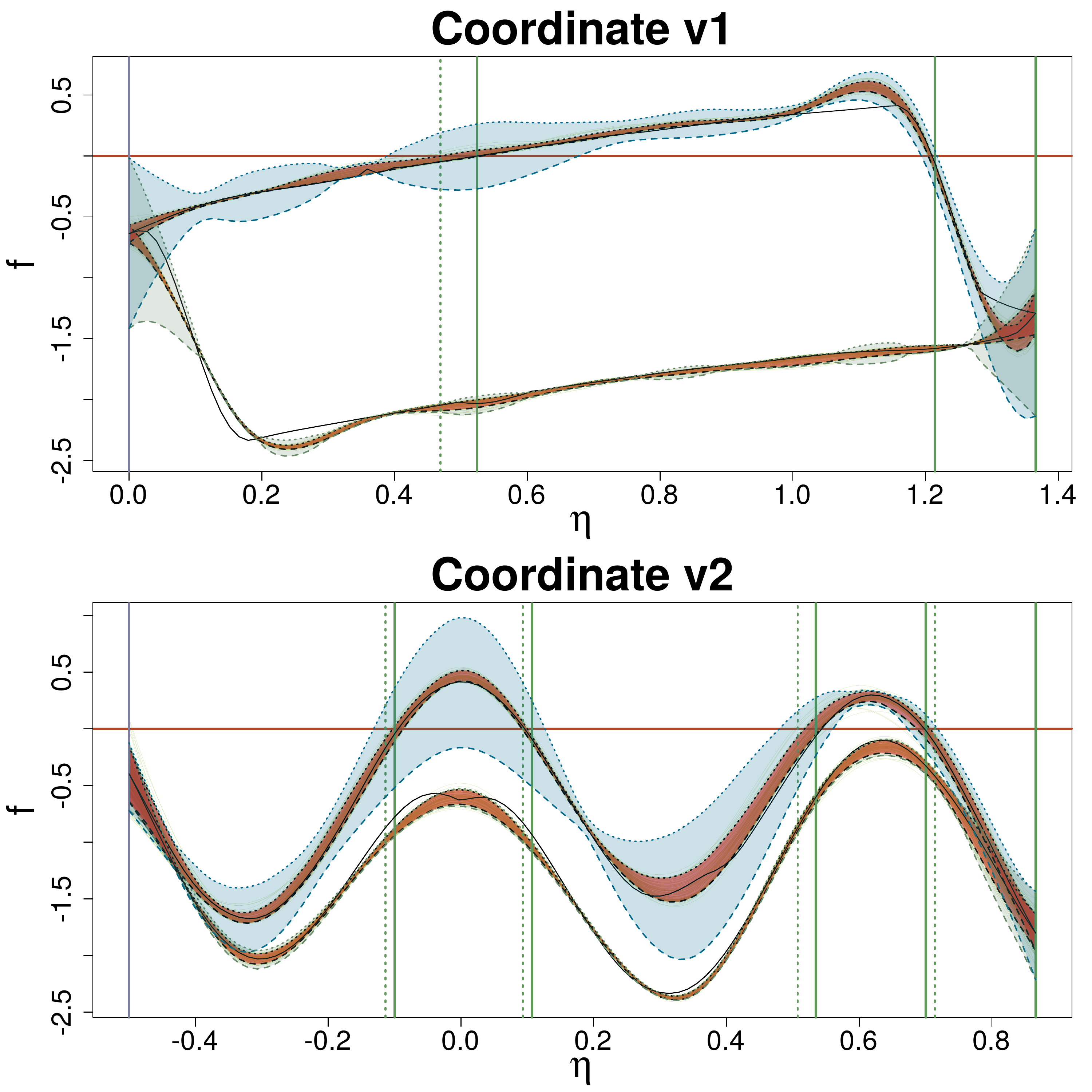}
			\caption{Profile extrema functions ($n=90$, $\threshold=0$): mean (black, solid), breakpoints (vertical green; mean: solid, quantiles: dashed), quantiles ($90 \%$, black dashed, CI dark red tube), upper-lower bound (light blue tube).}
			\label{fig:ex2dprof_90}
			\end{minipage}	
\end{figure}

\subsection{Bounds for the approximation}
\label{subsec:bound}

The approximating process $\widetilde{Z}$ does not provide proper probabilistic statements for $\Psup_\Psi Z, \Pinf_\Psi Z$ as it is based on the $\ell$-dimensional random vector $Z_G$. However since the distribution of the error process $Z-\widetilde{Z}$ can be expressed in closed form, we can control the approximating error with the following probabilistic bound. Proofs are in Appendix~\ref{sec:proofs}.

\begin{theorem}
	Consider a Gaussian process $(Z_\x)_{\x \in D} \sim GP(\mean,\covkern)$, the approximating process $\widetilde{Z}$ of $Z$ based on the points $G$ (Equation~\eqref{eq:tildeZ}) and $T \subset D$, then for any $u>\mean^{\widetilde{\Delta}}_T$
	\begin{equation}
	Pr\left( \left\lvert \sup_{\x \in T} Z_\x - \sup_{\x \in T} \widetilde{Z}_\x \right\rvert > u \right) \leq 2 \exp \left( -\frac{(u- \mean^{\widetilde{\Delta}}_T )^2 }{2 (\sigma^{\widetilde{\Delta}}_T)^2} \right),
	\label{eq:supBound}
	\end{equation}
	where 
	\begin{align}
	\label{eq:mmvvDelta}
	\mean^{\widetilde{\Delta}}_T &=\sup_{\x \in T} \lvert \mean^{\widetilde{\Delta}}(\x) \rvert \text{ and } (\sigma^{\widetilde{\Delta}}_T)^2 = \sup_{\x \in T} \covkern^{\widetilde{\Delta}}(\x,\x) \text{ with } \\ \nonumber
	\mean^{\widetilde{\Delta}}(\x) &= \E[Z_\x - \widetilde{Z}_\x] = \mean(\x) - a(\x) - \mathbf{b}^T(\x)\mean(G) \qquad \x \in T \\ \nonumber
	\covkern^{\widetilde{\Delta}}(\x,\y) &=
	\covkern(\x,\y) - \covkern(\y,G)\mathbf{b}(\x) - \covkern(\x,G)\mathbf{b}(\y) + \mathbf{b}^T(\x)\covkern(G,G)\mathbf{b}(\y), \qquad \x,\y \in T 
	\end{align}
	\label{theo:bound}
\end{theorem}

If the approximating process $\widetilde{Z}$ is unbiased then $\mean^{\widetilde{\Delta}}(\x) \equiv 0$ and the inequality in~\eqref{eq:supBound} is valid for any $u>0$.

\begin{corollary}
	\label{cor:bounds}
	Consider the profile sup random functions defined as 
	\begin{equation*}
	\Psup_{\Psi} \widetilde{Z} (\eta) := \sup_{ \{\x \in D : \Psi^T\x = \eta \}} \widetilde{Z}_\x \ \text{ and } \ \Psup_{\Psi} Z (\eta) := \sup_{ \{\x \in D : \Psi^T\x = \eta \}} Z_\x, \qquad \eta \in E_\Psi,
	\end{equation*}
	where $\widetilde{Z}$ is an unbiased approximate process in the form of Equation~\eqref{eq:tildeZ}. For any $\eta^* \in E_\Psi$ and $\alpha \in (0,1)$
	\begin{gather}
	\nonumber
	Pr\big(\Psup_{\Psi} Z (\eta^*) \leq u^+_{\alpha} \big) \geq 1- \alpha, \qquad \text{ with } \ u^+_{\alpha} = \tilde{u}^+_\beta+\sqrt{2 \left(\sigma^{\widetilde{\Delta}}_T\right)^2 \log\left(\frac{2}{\alpha -\beta}\right)} \text{ and } \\ 
	Pr\big(\Psup_{\Psi} Z (\eta^*) \geq u^-_{\alpha} \big) \geq 1- \alpha, \qquad \text{ with } \ u^-_{\alpha} = \tilde{u}^-_\beta -\sqrt{2 \left(\sigma^{\widetilde{\Delta}}_T\right)^2 \log\left(\frac{2}{\alpha -2\beta}\right)}
	\label{eq:profBoundCor}
	\end{gather}
	with $(\sigma^{\widetilde{\Delta}}_T)^2$ as in~\eqref{eq:mmvvDelta}, $T= \{ \x \in D: \Psi^T\x = \eta \}$, $\alpha > 2\beta$ and $\tilde{u}^+_\beta,\tilde{u}^-_\beta$ are the $1-\beta, \beta$ quantiles for $\Psup_{\Psi} \widetilde{Z} (\eta^*)$ respectively. Equations~\eqref{eq:profBoundCor} imply $Pr\big(\Psup_{\Psi} Z (\eta^*) \in [u^-_\alpha, u^+_\alpha] \big) \geq 1-2\alpha$.
\end{corollary}

In practice the quantiles $\tilde{u}^\pm_\beta$ in equation~\eqref{eq:profBoundCor} are estimated with sample quantiles $\widehat{\tilde{u}^\pm_\beta}$ from the realizations of $\widetilde{Z}$ and $u^\pm_\alpha$ are estimated  
by plugging-in $\widehat{\tilde{u}^\pm_\beta}$ in equation~\eqref{eq:profBoundCor}. 

Figures~\ref{fig:ex2dprof_20} and~\ref{fig:ex2dprof_90} show the conservative bound in equation~\eqref{eq:profBoundCor} on $\Psup_i Z$ (skyblue,  lightly shaded tube) and $\Pinf_i Z$ (seagreen, lightly shaded), $i=1,2$, on the example presented in Section~\ref{subsec:analExUQ}. The uncertainty is much smaller with $n=90$, however the bound is still very conservative. An indicator for the bound tightness is the quantity $(\sigma^{\widetilde{\Delta}}_T)^2(\eta,\ell)$, for each $\eta \in E_\Psi$. Here we explicit the dependency of the approximation on $\ell$. In particular, we study the integral of this map,~i.e.
\begin{equation}
I(\sigma^{\widetilde{\Delta}}_T)^2(\ell) := \int_{E_\Psi} (\sigma^{\widetilde{\Delta}}_T)^2 (\eta,\ell) d\eta
\label{eq:integratedVarDiff}
\end{equation}
 Appendix~\ref{sec:deltaTex2d} shows a comparison of $I(\sigma^{\widetilde{\Delta}}_T)^2$ from the example in figure~\ref{fig:ex2dprof_90}, as a function of $\ell$. Figure~\ref{fig:ex2dDeltaTcomparison}, in particular, shows exponential decrease in $I(\sigma^{\widetilde{\Delta}}_T)^2$ as $\ell$ grows to $80$.

\subsection{Discussion on method's parameters}
Profile extrema functions for an expensive-to-evaluate function require the user to set several parameters summarized in Table~\ref{tab:paramSummary}, Appendix~\ref{sec:deltaTex2d}. 
The DoE selected to train the GP model has the biggest influence on the profile extrema uncertainty. Figures~\ref{fig:ex2dprof_20} and~\ref{fig:ex2dprof_90} show that a GP model trained on more function evaluations, reduces the uncertainty on profiles. The bounds on profile extrema, equation~\eqref{eq:profBoundCor},  however, are also controlled by two parameters chosen by the user given a fixed DoE: the number of pilot points, $\ell$, and of  GP realizations, $\npostGPsim$. The number of posterior GP realizations, $\npostGPsim$, controls the accuracy of the empirical quantiles $\tilde{u}^\pm_\beta$. More importantly, the number of pilot points $\ell$ controls directly the tightness of the bounds in equation~\eqref{eq:profBoundCor}. This parameter plays a more prominent role in uncertainty quantification as more pilot points lead to a smaller $\sigma^{\widetilde{\Delta}}_T$ and a tighter bound, see,~e.g., Figure~\ref{fig:ex2dDeltaTcomparison}, Appendix~\ref{sec:deltaTex2d}.  

\section{Motivating application: coastal flooding}
\label{sec:motivating}

\subsection{Motivation and study case}
Coastal flooding models experienced recent progresses opening new research and applications perspectives. However, their computational cost ($>$ hours) hinders their use when a large number of simulations is required  for estimating the excursion region corresponding to the critical forcing conditions leading to inundation or when forecast is needed \citep{RohmerIdier2012,idier13}. 

We focus here on coastal flooding induced by overflow and we consider the Boucholeurs area (French Atlantic coast, see Figure~\ref{fig:LOC}). This area is located close to La Rochelle and was flooded during the 2010 Xynthia storm event. This event was characterized by a high storm surge ($>1.5 \meters$ at La Rochelle tide gauge) in phase with a high spring tide \citep{bertin14}. Here, we focus on these primary drivers (tide and storm surge) and on how they affect the resulting flooded surface ($Y$, in square meters).

\begin{figure}
	\centering
	\includegraphics[width=0.85\linewidth]{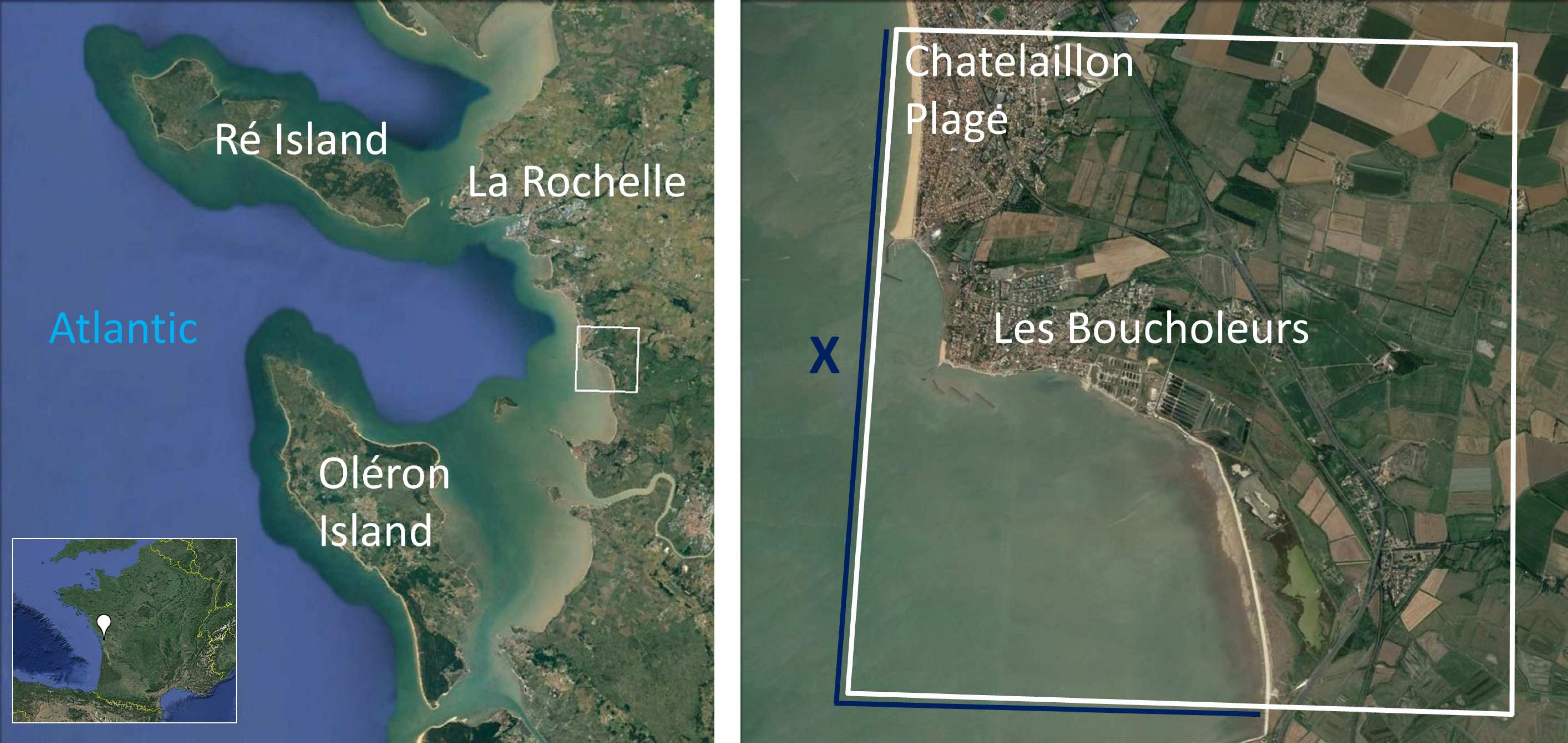}
	\caption{\label{fig:LOC}Study site location (left) and computational domain limits (right, in white) and location of the forcing conditions $\x$ (right, in blue).}
\end{figure}

\subsubsection{The forcing conditions}

The offshore forcing conditions correspond to the tide and storm surge temporal evolution  (see Figure~\ref{fig:X}a) are denoted with $\x = (T,S,t_0,t_+,t_-)$. They are parametrized as follows:
\begin{itemize}
	\item the tide is simplified by a sinusoidal signal parametrised by its high tide level $T \in [0.95 \meters,3.70 \meters]$, (see Figure~\ref{fig:X}a);
	\item the surge signal is assumed to be described by a triangular model (see Figure~\ref{fig:X}a) using four parameters: the peak amplitude $S \in [0.65 \meters,2.50 \meters]$, the phase difference $t_0\in [-6,6]$ hours, between the surge peak and the high tide, the time duration of the raising part $t_- \in [-12.0,-0.5]$ hours, and the falling part $t_+ \in [0.5,12.0]$ hours. 
\end{itemize}

\begin{figure}
	\centering
	\includegraphics[width=0.85\linewidth]{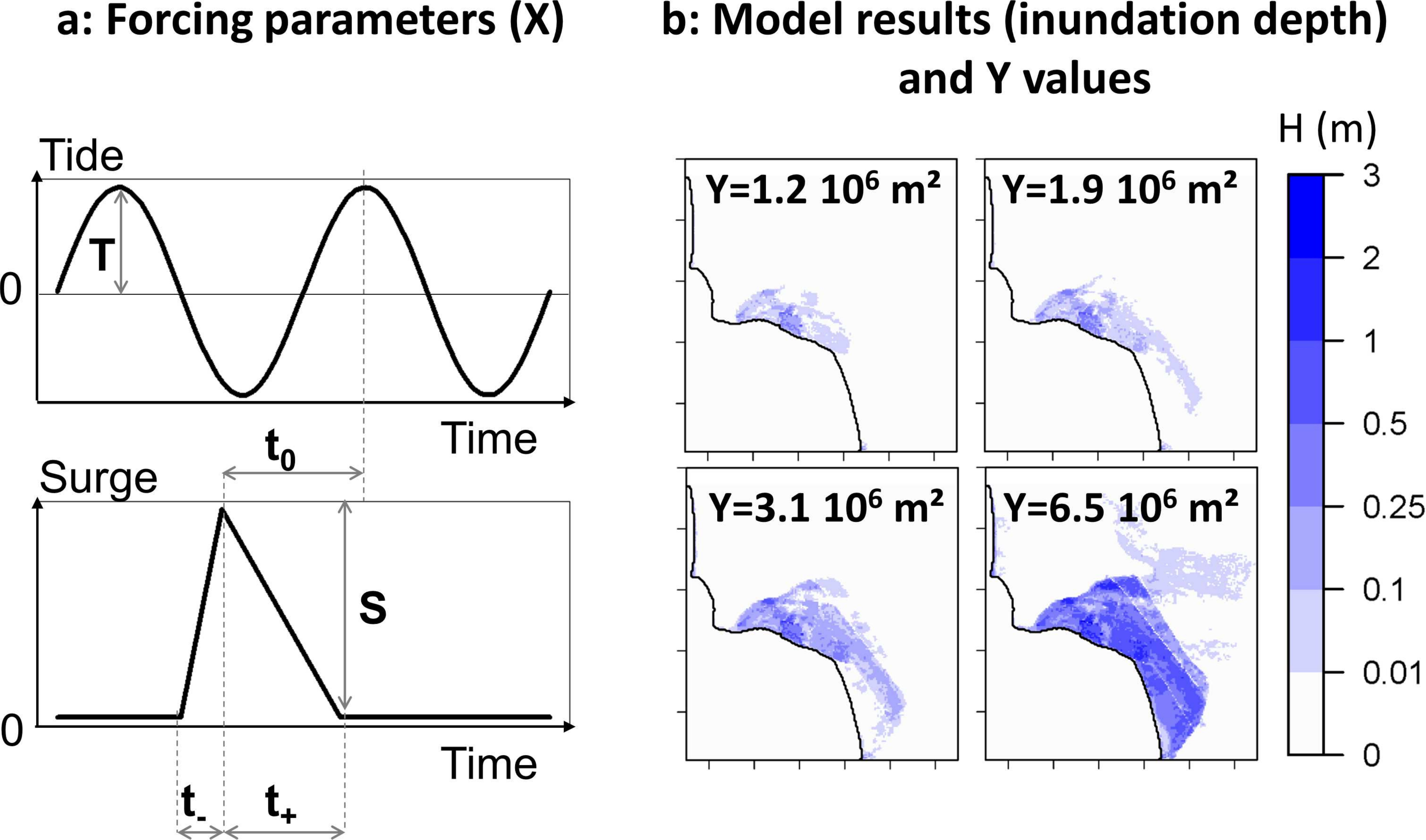}
	\caption{\label{fig:X}(a) Schematic representation of the tide and surge temporal signals and the different parameters describing them. (b) Maps of inland water height for given values of the $\x$ parameters, and deduced $Y$ value of flood surface (in square meter).}
\end{figure}

\subsubsection{The numerical model}
The numerical modelling of the coastal flood relies on the MARS model \citep{lazure08}. This finite-difference model solves the shallow-water equations and was originally designed to compute regional coastal hydrodynamics, e.g., tide and storm induced water level and currents. 
The MARS model has here been adapted to account for the specificities of local coastal flooding processes: hydraulic processes around connections like culverts and weirs, coastal defence breaching. This model has been implemented on the study site (white box in Figure~\ref{fig:LOC}) with a spatial resolution of $25 \meters$ and a total number of mesh cells of $>$39,000. The land cover effect on the flood is taken into account by using a spatially varying friction coefficient. The different hydraulic connections (e.g., the hydraulic culverts below the roads, dike, railway,…) are taken into account in the modeling. The forcing conditions (time series deduced from the parameters $\x$, Figure~\ref{fig:X}a) are uniform over the open boundaries of the domain in blue on Figure~\ref{fig:LOC}. A single model run takes about 30-60 minutes of computation using a single CPU. For more details on the study site, the model set-up and validation, see~\citet{Rohmer.etal2018}. It should be noted that when wave overtopping is dominant in the flooding processes, other types of models should be used \citep{leroy15}, with computation times 2 orders of magnitude larger. An overflow case allows setting up statistical developments which will useful also for the more expensive models. Figure~\ref{fig:X}b provides examples of the inundation depth ($H$) computed in each cell for given forcing conditions $\x$, as well as the resulting flood surface value ($Y$). We consider the threshold values $\threshold^{(Y)}_1=1.2\times 10^6 \meters^2, \ldots, \threshold^{(Y)}_4=6.5\times 10^6 \meters^2$ introduced in Figure~\ref{fig:X}b.

\subsubsection{Gaussian process model}
\label{subsec:model}

We consider a rescaled input space $D= [0,1]^5$ and the function $\check{f}: D \subset \R^5 \rightarrow \R_+$ with $Y = \check{f}(\x)$ for each $\x = (x_1, \ldots, x_5) = (\Tide, \Surge, \PhiVar, \tPlus,\tMinus) \in D$. In the remainder we keep the notation $(\Tide, \Surge, \PhiVar, \tPlus,\tMinus)$ for the rescaled input. 
We are interested in estimating 
\begin{equation}
\varGamma = \{ \x \in D : \check{f}(\x) \geq \threshold \}, \ \text{ where } \  \threshold= \threshold^{(Y)}_k,\ \text{ for } \ k=1, \ldots, 4.
\label{eq:Gamma5d}
\end{equation} 

We consider the square root transformed output data $f(\x) = \sqrt{\check{f}(\x)} = \sqrt{Y}$ for $\x \in D$ and the square root thresholds $\threshold_k = \sqrt{\threshold^{(Y)}_k}$. This transformation was chosen, after fitting the model on different scales for $Y$, because it provided the best cross-validation metrics. 

We fix an initial DoE $\doe_n \in D^n$, with $n=200$ points obtained by evaluating the function $f$ on the first $500$ points of the $5$-d Sobol' sequence and by selecting the first $n=200$ points leading to a flood of any magnitude. The evaluations $f(\doe_n)$ at $\doe_n$ are denoted with $\mathbf{y}_n \in \R^n$. 
We consider a GP model with a tensor product prior Mat\'ern covariance kernel $\nu=3/2$ and prior mean of the form
\begin{equation}
\mu(\x) = c_0+c_1 x_1 + c_2 x_2 + c_3 x_3^2 + c_4 x_4 + c_5 x_5 = c_0 + \sum_{j=1}^5c_j h_j(\x).
\label{eq:priorMean}
\end{equation}
The covariance kernel hyper-parameters are estimated with maximum likelihood from $(\doe_n, \mathbf{y}_n)$ and the posterior mean and covariance kernel are obtained with Equations~\eqref{eq:postMean}, \eqref{eq:postCov}.
The GP and the basis functions $h_j$ were selected using expert-based information achieving a 
$Q^2 = 0.958$. A comparison of different model fits is shown in Appendix~\ref{sec:full5dRes}.

We estimate $\varGamma$ from the posterior mean with 
\begin{equation*}
\hat{\varGamma}_{n,\threshold_k} = \{ \x \in D : \mean_n(\x) \geq \threshold_k\}, \ \text{ for }k=1, \ldots, 4.
\end{equation*}

\subsection{Procedure overview}
\begin{figure}[h!]
	\centering
	\includegraphics[width=0.85\linewidth]{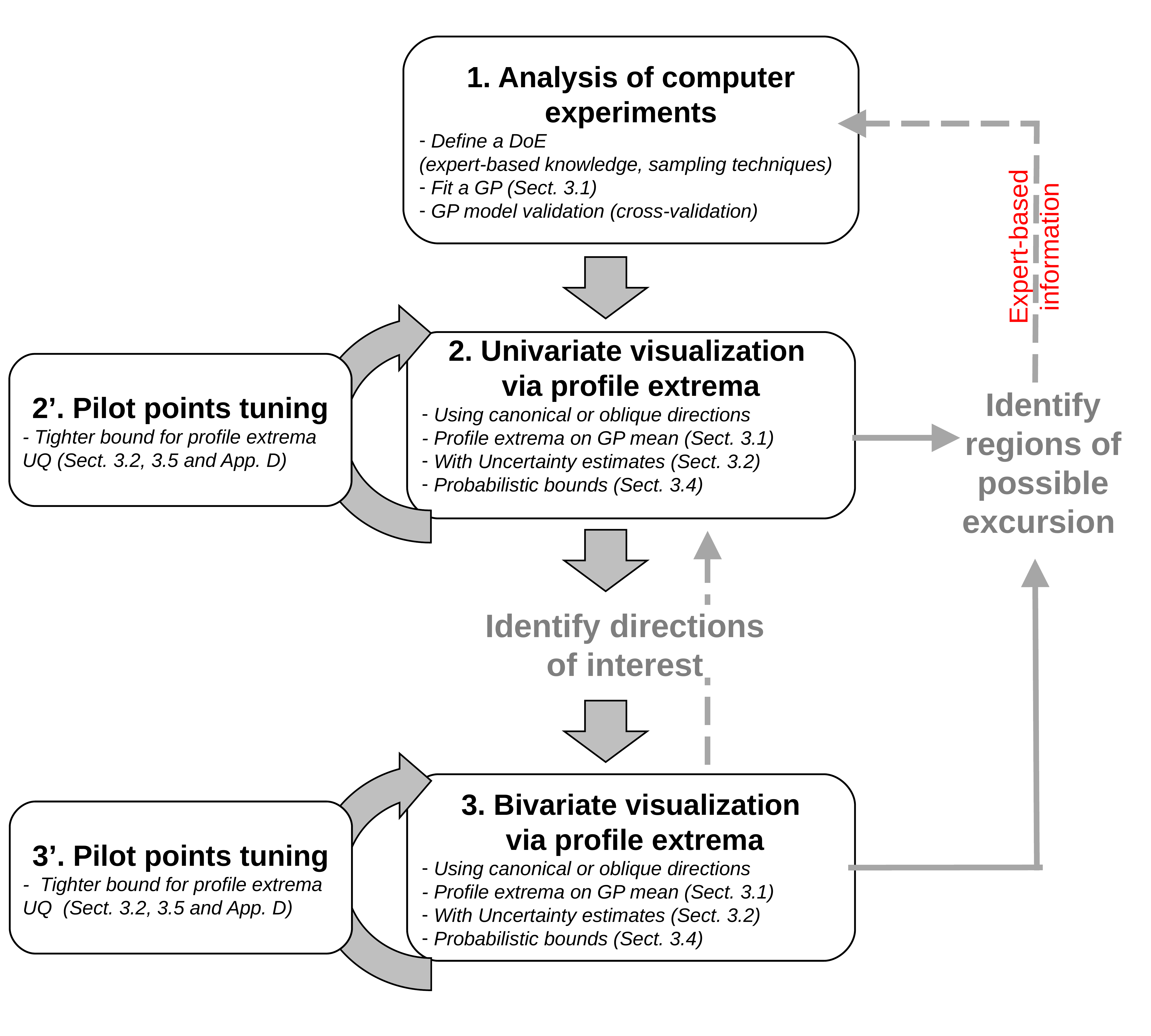}
	\caption{Flow chart for the procedure used in the test case.}
	\label{fig:flowChart}	
\end{figure}

In the following sections profile extrema are used to explore visually estimates for $\varGamma$ and to quantify their uncertainty. The proposed procedure is summarized in Figure~\ref{fig:flowChart}:
\begin{description}
	\item[step 1: design of experiment and emulation.] Select a DoE $\doe_n$ ($n=200$) and run the MARS model (numerical simulator) to compute the flooded area $Y_n$; Fit a GP (\ref{subsec:model}) on $\doe_n,Y_n$, evaluate the emulation quality with leave-one-out-cross validation.
	\item[step 2: univariate profile extrema.] Compute coordinate profile extrema (i.e.~along the canonical directions) on the GP mean and give a first visual indication on the excursion set. Uncertainty quantification on profile extrema (Section~\ref{sec:UQ}) guides expert knowledge in identifying possible regions of excursion. A comparison of several profile extrema plots with a different number of pilot points $\ell$ increases the understanding of profile functions uncertainty without additional numerical simulator runs. The conclusions drawn from 1d profile extrema  can be used to refine the DoE in step 1 and ultimately decrease the uncertainty on the excursion set. This first analysis, for example, shows that some offshore conditions do not influence the excursion. Oblique profile extrema could be used in this phase, see, e.g., section~\ref{subsec:analExample}, if more informative directions are known in advance.
	\item[step 3: bivariate profile extrema.] Explore combinations of input variables that lead to excursion with bivariate profile extrema functions. Orthogonal projections can potentially be used as shown in Section~\ref{subsec:analExample} on the analytical example. Similarly as for step 3, refinements of the DoE in the regions of interest can be performed.
\end{description}

\subsection{Results on coastal flooding test case}
\label{subsec:results}

\subsubsection{Univariate profile extrema}

We start the analysis of $\hat{\varGamma}_{n,\threshold_k}$, for the thresholds $\threshold_1, \ldots, \threshold_4$, with coordinate profile extrema on the posterior GP mean. The uncertainty is visualized with posterior quantiles of the profile extrema (Section~\ref{subsec:approxReals}) and with the upper and lower bound of equation~\eqref{eq:profBoundCor}. 

Figure~\ref{fig:ex5dUQfun_approx} shows the coordinate profile extrema for the posterior mean of the process based on the design with $n=200$ points described in Section~\ref{subsec:model}, with the universal kriging prior mean defined in Equation~\eqref{eq:priorMean} and the lowest and highest threshold $\threshold_1,\threshold_4$.

Let us consider $\threshold_4=2549.5 = \sqrt{6.5e6}$, the highest threshold in dark red. Coordinate profile extrema on the posterior mean tell us that if $\Tide<0.57$, there is no excursion independently of the other coordinates. The $90\%$ point-wise confidence interval is based on $\npostGPsim=600$ approximate posterior realizations generated with $\ell=300$ pilot points. The point-wise confidence intervals identify a possible ($90\%$) region of non-excursion $\{ \x \in D : \Tide \in [0, \lambda] \text{ with }\lambda \in [0.52,0.6] \}$. If we consider the variable $\Surge$, a possible region of non-excursion (above $\tau_4$) is $\{ \x \in D : \Surge \in [0,\lambda], \ \lambda \in [0.28,0.43] \}$. This region does not exist for low threshold values ($<\tau_2$), thus indicating that small values of surge peak only play a role in moderate flooding events, i.e.~$<\tau_2$. Similar assessments are available for the other coordinates and the other thresholds, see Table~\ref{tab:5dSummary} in Appendix~\ref{sec:full5dRes} for a summary. Note that the variables $\tPlus,\tMinus$ do not bring information on the excursion regions as $\Pinf_{\e_{4,5}}\gamma_n(\eta)$ and $\Psup_{\e_{4,5}}\gamma_n(\eta)$ are consistently below and above the thresholds $\threshold_1, \threshold_4$ respectively. The bounds on the approximating process from equation~\eqref{eq:profBoundCor}, plotted as wide light blue tube, show that there is still uncertainty on this assessment. In fact, by accounting also for the approximation uncertainty, the possible region of non-excursion becomes $\{ \x \in D: \Tide \in [0,0.19	] \}$.  The tightness of the bound is evaluated by looking the integrated variance of the difference, i.e. equation~\eqref{eq:integratedVarDiff}. We chose $\ell=300$ as the resulting integrated variance is small enough and at the same time the method is not computationally too expensive. For example the average integrated variance over all dimensions is $41.6 \%$ and $18.3 \%$ smaller than with $\ell=37$ for $\Pinf$ and $\Psup$ respectively. On the other hand, the computational time for $(\sigma_T^{\widetilde{\Delta}})$ grew from $302$ seconds ($\ell=37$) to $1595$ seconds for $\ell=300$. More details in Appendix~\ref{sec:full5dRes}.

\begin{figure}[h!]
		\centering
		\includegraphics[width=0.85\linewidth]{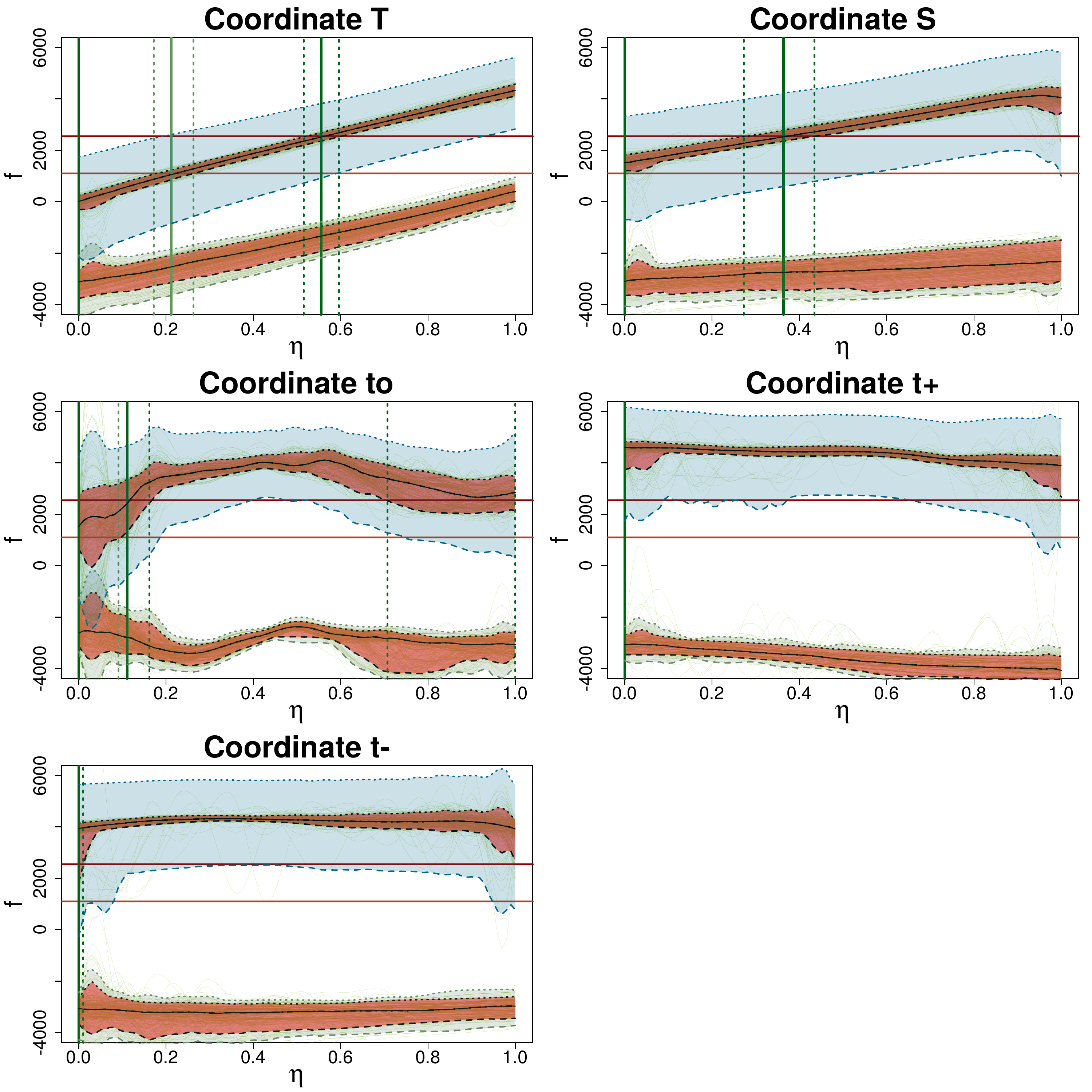}
		\caption{Coordinate profiles for $5$-dimensional test case. Excursion thresholds (red, light $\threshold_1$, dark $\threshold_4$), breakpoints $\Psup_i/\Pinf_i$ on mean (vertical green lines) and on $95\% - 5\%$ quantiles (vertical green lines, dashed), CI dark red tube, bound range light blue tube.}
		\label{fig:ex5dUQfun_approx}	
\end{figure}

\subsubsection{Bivariate profile extrema}

We now focus on the combination of variables $\Tide, \Surge$ and $\Tide, \PhiVar$, whose more prominent role on the excursion was outlined by the coordinate profile extrema, and we consider $\threshold=\tau_1$. We explore which values of these combinations lead to excursion with bivariate profile extrema functions. For example, the profile extrema for $\Tide, \Surge$ is obtained with 
\begin{equation*}
\Psi = \begin{bmatrix}
1 & 0 & 0 & 0 & 0 \\
0 & 1 & 0 & 0 & 0
\end{bmatrix}^T.
\end{equation*}
We compute the empirical quantile maps and the bound with $\ell=300$ and $\npostGPsim=420$ approximate posterior realizations, see Section~\ref{subsec:approxReals}. Figures~\ref{fig:ex5dbivProf} (a,b,c) show the contour lines of $\Psup_\Psi$ and $\Pinf_\Psi$ for the posterior GP mean based on the DoE described in Section~\ref{subsec:model}. The background colors indicate different heuristic measures of uncertainty: (a) the weighted inter-quantile range (i.e.~the empirical inter-quantile range for profile extrema maps if the threshold is between the upper and lower quantile or zero, otherwise); (b) the upper-lower bound range (i.e.~the difference between the upper and lower bound in~\eqref{eq:profBoundCor} if $\threshold$ is between them, zero otherwise); (c) the standard deviation of the difference $(\sigma^{\widetilde{\Delta}}_T) (\eta)$, $\eta \in E_\Psi$.

\begin{figure}[!]
	\begin{minipage}{0.02\textwidth}
		\centering 
		(a)
	\end{minipage}
	\begin{minipage}{\textwidth}
		\centering
		\includegraphics[width=0.85\linewidth]{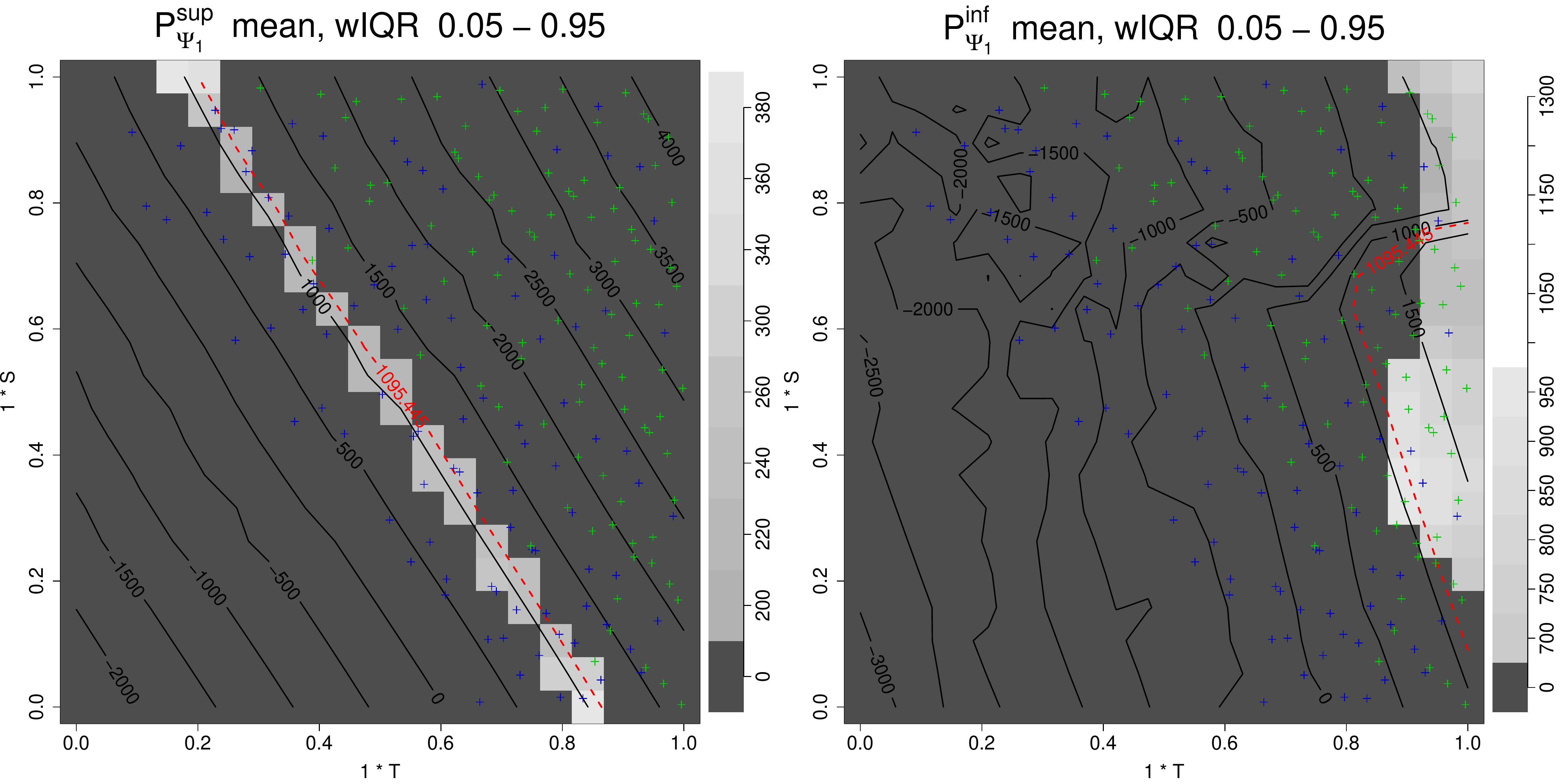}
	\end{minipage} \\
	\begin{minipage}{0.02\textwidth}
		\centering 
		(b)
	\end{minipage}
	\begin{minipage}{\textwidth}
		\centering
		\includegraphics[width=0.85\linewidth]{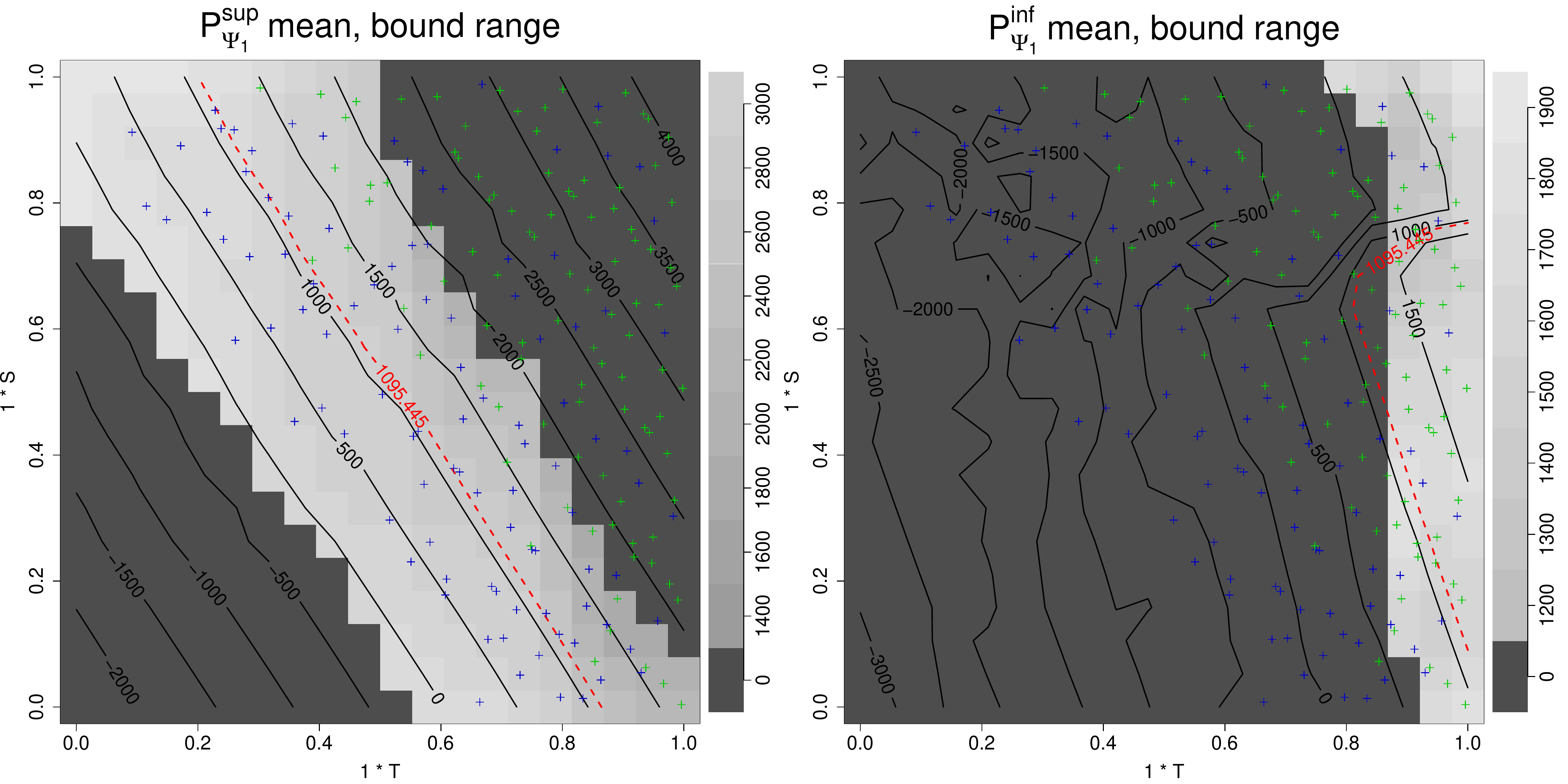}
	\end{minipage} \\
	\begin{minipage}{0.02\textwidth}
		\centering 
		(c)
	\end{minipage}
	\begin{minipage}{\textwidth}
		\centering 
		\includegraphics[width=0.85\linewidth]{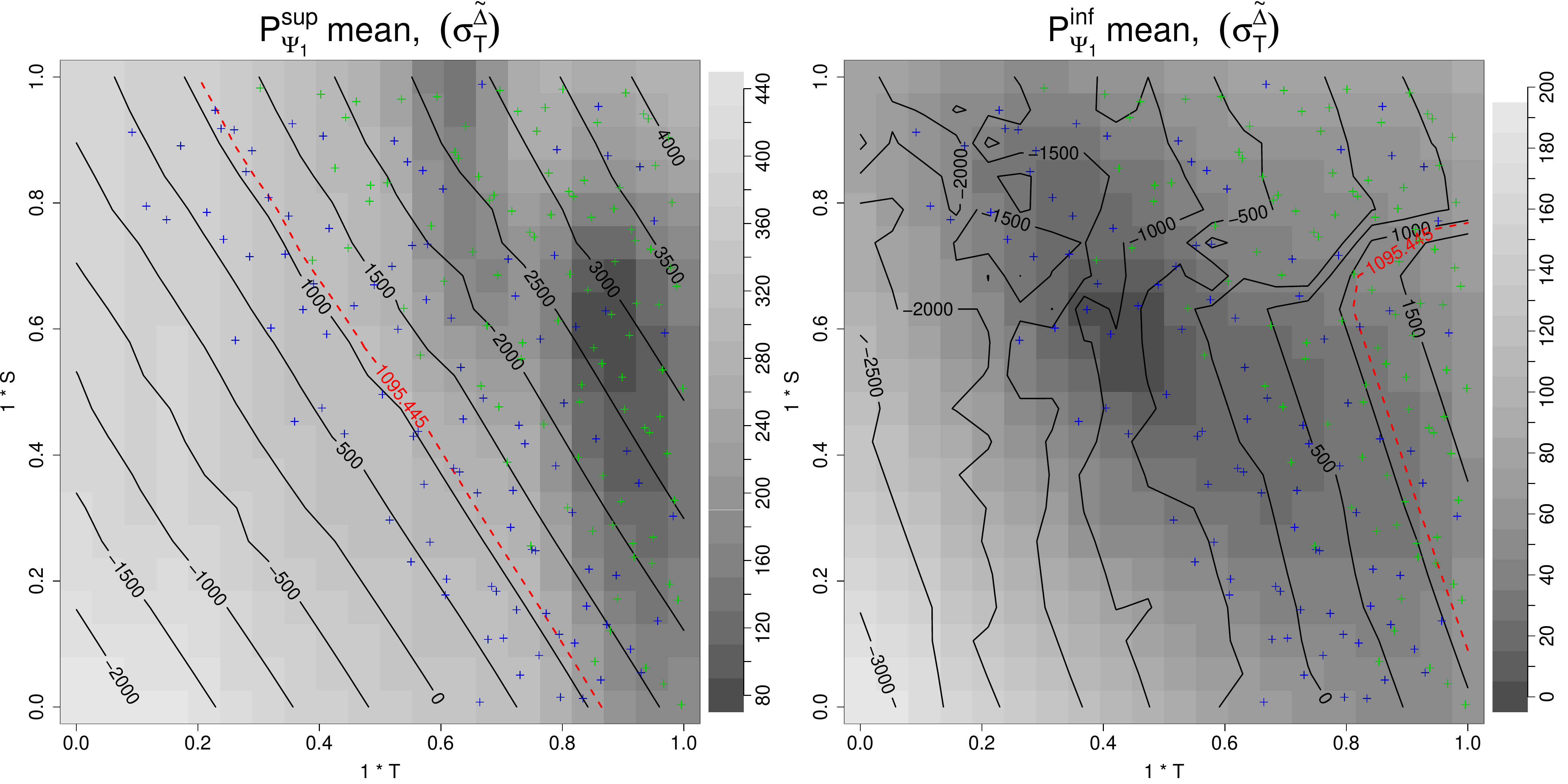}
	\end{minipage}
	\caption{Bivariate ($\Tide,\Surge$) profile extrema mean (contour lines). Threshold $\threshold_1$ dashed red line, $\doe_n$ locations (crosses, green above $\threshold_1$). Background: (a) weighted inter-quantile range ($95\% - 5\%$), (b) bound range (eq.~\eqref{eq:profBoundCor}), (c) approximation error standard deviation.}
	\label{fig:ex5dbivProf}	
\end{figure}

Let us start by analyzing $\Psup_\Psi,\Pinf_\Psi$ on the posterior mean, shown in the contour lines in Figures~\ref{fig:ex5dbivProf}.  The map on the left suggests that the region below the critical dashed red contour can be excluded, i.e.~the region is outside the excursion. In a symmetric way, $\Pinf_\Psi$ (Figure~\ref{fig:ex5dbivProf}, right) suggests that high values for $\Tide$ ($>0.8$) along with values for $\Surge$ in the range $[0.05,0.7]$ should lead to an excursion (i.e. flood). The bivariate profile extrema for $\Tide,\PhiVar$ in Figure~\ref{fig:ex5dbivProfTP}, instead, show (left) a non-flood region on the left of the critical contour and suggest (right)  excursion in the region $\Tide >0.95$,  $\PhiVar\in[0.42,0.68]$. These conclusions are consistent with our physical knowledge, except for the excursion domain on $\Tide,\Surge$ in Figure~\ref{fig:ex5dbivProf} (right). Indeed, an excursion domain bounded by maximum surge does not have a physical explanation because $Y$ should increase with $\Surge$, with no maximum bound of $\Surge$. 

The above analysis is based only on $\Psup$ and $\Pinf$ for the posterior GP mean and, since the DoE is small and non-adaptive, there could be high uncertainty in some parts of the input domain. The indicators plotted as background colors allow us to quantify this uncertainty. Both the weighted inter-quantile range and the bound range for $\Pinf_\Psi$ in Figure~\ref{fig:ex5dbivProf} (right) are high the region of large $\Tide$ and $\Surge$. Moreover, as shown in Figure~\ref{fig:ex5dbivProf}(c), the values of $(\sigma^{\widetilde{\Delta}}_T)$ are not high in that region indicating that the uncertainty due to the approximating process is not very high. Those insights collectively suggest that more function evaluations should be added in the $\Pinf_\Psi$ uncertain region.

\subsubsection{Summary of results}

Coordinate profile extrema functions on the coastal flooding test case enabled: (i) to highlight the major role of the high tide level, $\Tide$, whatever the considered thresholds, i.e., for small to large flooding events; (ii) to highlight the key role of the surge peak, $\Surge$, only for moderate thresholds i.e., moderate flooding events; (iii) to highlight the moderate role of the phase difference $\PhiVar$ alone; (iv) to exclude a strong influence of $\tMinus$ and $\tPlus$ alone for the excursion of the response whatever the considered thresholds. Moreover by studying bivariate profile extrema we could: (i) exclude regions where $\Tide$ and $\Surge$ are simultaneously small (e.g. $\Tide<0.4$ and $\Surge <0.6$); (ii) highlight the role of phase difference $\PhiVar$ and tide $\Tide$ combined with a possible excursion in the region $\Tide > 0.95$, $\Surge \in [0.42,0.68]$. 
Furthermore, the comparison of approximation error indicators (Figure~\ref{fig:ex5dbivProf}(c)) with uncertainty measures on profile extrema (Figure~\ref{fig:ex5dbivProf}(b)) enabled us to nuance the results and to track the main uncertainty source. In the coastal flooding test case, the uncertainty unlikely stems from the approximating process, but rather from the lack of function evaluations.

\section{Discussion}
\label{sec:discussion}

In this work we introduced a visualization technique, based on profile extrema functions, for excursion regions of expensive-to-evaluate functions. The main idea is to study the constrained optima of the functions on lower dimensional subspaces resulting from projections on lines or planes. By plotting profile extrema functions we can select regions of interest.  If the function is expensive to evaluate, a GP model is used to emulate the function and we showed how profile extrema can be computed on the GP model. In this case the conclusions strongly depend on the chosen GP model and, as we show in Appendix~\ref{sec:full5dRes}, the profile extrema uncertainty is affected by the modeling choices.  

In the coastal flooding test case, as sketched in Figure~\ref{fig:flowChart}, we selected directions of interest from coordinate profiles. For example, the profile sup along the direction of coordinate~$\Tide$ is below the threshold of interest for some values. This indicates that~$\Tide$ is a direction of interest for the excursion phenomena. This procedure was repeated for each coordinate. In test cases where canonical directions are not meaningful we could explore the excursion by looking at oblique directions as we did in the analytical example in Section~3.3.  

The bivariate profile extrema maps in Figure~\ref{fig:ex5dbivProf} also suggest a principled way to develop adaptive design of experiments that sequentially reduce uncertainties. For example, we could select the next evaluations as the minimizers of the integrated (or maximum of) $(\sigma^{\widetilde{\Delta}}_T)$. Such criteria should be analytically tractable and could lead to adaptive designs similar to classic IMSE (MSE) strategies,~\citet{Sacks.etal1989}. Alternative criteria could be obtained by minimizing bound-range uncertainties (Figure~\ref{fig:ex5dbivProf}(b)), as they provide direct information on the excursion, however their tractability is still unclear. 

Profile extrema functions require a continuously differentiable function on a compact domain. If the domain is not compact, then the optimization in the definition of profile extrema might not be well posed. In case of an open domain with a probability distribution on the inputs, profile extrema functions could be extended by using quantiles in place of the maximum \citep{RoyNotz2014}. Profile extrema functions could be extended to non-linear subsets, however it is technically not straightforward and it might result in visualizations that are much harder to interpret.
The overall approach developed here is a one-step procedure, and it could become part of an exploratory work flow. 
As shown in Appendix~\ref{sec:bivEx3d}, coordinate profile extrema, oblique and bivariate profiles can be combined to convey more information on the excursion set in simpler terms. A possible future extension could involve a treed procedure where the input space is restricted with constrained coordinate profile functions. Oblique coordinate profiles, i.e. profiles along non-canonical directions, require the user to choose which directions to explore. This choice could be driven by expert knowledge, such as in the motivating test case presented here. However, when such knowledge is not available, we could envisage a procedure following similar steps to projection pursuit~\citep{Cook_etal1995} where we obtain the most informative direction.


	
	


\bibliographystyle{apalike}
\bibliography{biblio}

\appendix

\section{Proofs}
\label{sec:proofs}

\begin{proof}[Theorem~\ref{theo:bound}]
Consider the GP regression set-up as described in Section~\ref{subsec:model}. For ease of notation let us denote the posterior process as $(Z_\x)_{\x \in T}\sim GP(\mean,\covkern)$, where we drop $n$, the number of observations as it is fixed in this section. Recall that the proposed approximate field $(\widetilde{Z}_\x)_{\x \in T}$ is defined as follows
\begin{equation*}
\widetilde{Z}_\x = a(\x) + \mathbf{b}^T(\x)Z_G \qquad \x \in T,
\end{equation*}  
where $a: T \rightarrow \R$ is a continuous trend function, $\mathbf{b}: T \rightarrow \R^\ell$ is a continuous vector-valued function of deterministic weights, $G = (\g_1, \ldots, \g_\ell ) \in T^\ell$ is a fixed sequence of points in $T$ and $Z_G = (Z_{\g_1}, \ldots, Z_{\g_\ell})^T$ is a $\ell$-dimensional random vector. Let us consider the difference process $(\widetilde{\Delta}_\x)_{\x \in T}$, defined as $\widetilde{\Delta}_\x := Z_\x - \widetilde{Z}_\x$, for each $\x \in T$. The mean function and covariance kernel of $\widetilde{\Delta}$ are
\begin{align*}
\mean^{\widetilde{\Delta}}(\x) &= \E[Z_\x - \widetilde{Z}_\x] = \mean(\x) - a(\x) -\mathbf{b}(\x)^T \mean(G) \qquad \x \in T \\
\covkern^{\widetilde{\Delta}}(\x,\y) &= \covkern(\x,\y) - \covkern(\y,G)\mathbf{b}(\x) - \covkern(\x,G)\mathbf{b}(\y) + \mathbf{b}^T(\x)\covkern(G,G) \mathbf{b}(\y)
\qquad \x,\y \in T
\end{align*}

First of all notice that
\begin{equation}
Pr\left(\left\lvert \sup_{\x \in T}  Z_\x - \sup_{\x \in T} \widetilde{Z}_\x \right\rvert >u \right) \leq Pr\left(\sup_{\x \in T} \lvert \widetilde{\Delta}_\x \rvert >u \right)
\label{eq:boundIneq1}
\end{equation}

Let us now consider the centred process $(\widetilde{\Delta}^C_\x)_{\x \in T} := (\widetilde{\Delta}_\x - \mean^{\widetilde{\Delta}}(\x))_{\x \in T}$. We have that
\begin{equation}
Pr\left(\sup_{\x \in T} \lvert \widetilde{\Delta}^C_\x \rvert > u\right) \leq 2Pr\left(\sup_{\x \in T} \widetilde{\Delta}^C_\x  > u\right) \leq 2 e^{-u^2 /2(\sigma^{\widetilde{\Delta}}_T)^2}, \qquad u>0
\label{eq:boundIneq2}
\end{equation}
where $(\sigma^{\widetilde{\Delta}}_T)^2 = \sup_{\x \in T} \E[(\widetilde{\Delta}^C_\x)^2] = \sup_{\x \in T} \covkern^{\widetilde{\Delta}}(\x,\x)$. The first inequality follows from the symmetric distribution of the centered field $\widetilde{\Delta}^C_\x$ and the second is the Borell-TIS inequality, see, e.g.,~\citet{Adler.Taylor2007}, Chapter~2 for more detail.

Since we have 
\begin{equation}
\sup_{\x \in T} \lvert \widetilde{\Delta}_\x \rvert \leq \sup_{\x \in T} \lvert \widetilde{\Delta}^C_\x \rvert +\sup_{\x \in T} \lvert \mean^{\widetilde{\Delta}}(\x) \rvert 
\label{eq:boundIneq3}
\end{equation}
then following Equations~\eqref{eq:boundIneq1},~\eqref{eq:boundIneq2},~\eqref{eq:boundIneq3}, if $u> \mean^{\widetilde{\Delta}}(\x)$
\begin{align*}
Pr\left(\left\lvert \sup_{\x \in T}  Z_\x - \sup_{\x \in T} \widetilde{Z}_\x \right\rvert >u \right) 
&\leq Pr\left( \sup_{\x \in T} \lvert \widetilde{\Delta}^C_\x \rvert +\sup_{\x \in T} \lvert \mean^{\widetilde{\Delta}}(\x) \rvert >u \right) \\
&\leq Pr\left( \sup_{\x \in T} \lvert \widetilde{\Delta}^C_\x \rvert > u - \sup_{\x \in T} \lvert \mean^{\widetilde{\Delta}}(\x) \rvert \right) \\
&\leq 2 \exp\left(-\frac{\left(u- \sup_{\x \in T} \lvert \mean^{\widetilde{\Delta}}(\x) \rvert\right)^2}{2(\sigma^{\widetilde{\Delta}}_T)^2} \right).
\end{align*}
If $\widetilde{Z}$ is an unbiased approximation for $Z$, then $\mean^{\widetilde{\Delta}}(\x) =0$ and the inequality is valid for any $u>0$. 
\end{proof}

	\begin{proof}[Proof of Corollary~\ref{cor:bounds}]
		Let us denote with $A = \Psup_{\Psi} Z(\eta^*)$ and $B = \Psup_{\Psi} \widetilde{Z}(\eta^*)$. By using Theorem~\ref{theo:bound} we have that for $u>0$ 
		\begin{equation*}
		Pr(\lvert A- B \rvert \geq u) \leq 2 e^{- \frac{(u - \mean_T^{\widetilde{\Delta}})^2}{2 (\sigma_T^{\widetilde{\Delta}})^2 }}
		\end{equation*}
		Moreover
		\begin{align*}
		Pr(A \leq u^+_\alpha ) &= 1- Pr(A \geq u^+_\alpha) = 1 - Pr(A-B+B \geq u^+_\alpha) \\
		&= 1- Pr(A-B+B \geq u^+_\alpha \mid B \geq \tilde{u}^+_\beta) Pr (B \geq \tilde{u}^+_\beta) \\
		&- Pr(A-B+B \geq u^+_\alpha, B \leq \tilde{u}^+_\beta) \\
		&\geq 1-\beta - Pr(A-B+\tilde{u}^+_\beta \geq u^+_\alpha, B \leq \tilde{u}^+_\beta) \\
		&\geq 1-\beta - Pr(\lvert A-B \rvert \geq u^+_\alpha -\tilde{u}^+_\beta, B \leq \tilde{u}^+_\beta) \\
		&\geq 1- \beta - 2e^{- \tfrac{(u^+_\alpha - \tilde{u}^+_\beta - \mean_T^{\widetilde{\Delta}})^2}{2 (\sigma_T^{\widetilde{\Delta}})^2 }}
		\end{align*}
		By solving for $u^+_\alpha$ in  $\beta + 2e^{- \tfrac{(u^+_\alpha - \tilde{u}^+_\beta - \mean_T^{\widetilde{\Delta}})^2}{2 (\sigma_T^{\widetilde{\Delta}})^2 }} = \alpha$ under the constraint $u^+_\alpha > \tilde{u}^+_\beta$, we obtain~\eqref{eq:profBoundCor}. For the other side notice that
		\begin{align*}
		Pr(A \geq u^-_\alpha ) &\geq Pr(B\geq \tilde{u}^-_\beta, \lvert A - B \rvert  \leq \tilde{u}^-_\beta - u^-_\alpha) \quad = 1- Pr(B > \tilde{u}^-_\beta, \lvert A-B \rvert \geq \tilde{u}^-_\beta - u^-_\alpha) \\
		&- Pr(B \leq \tilde{u}^-_\beta, \lvert A-B \rvert \geq \tilde{u}^-_\beta - u^-_\alpha) - Pr(B \leq \tilde{u}^-_\beta, \lvert A-B \rvert \leq \tilde{u}^-_\beta - u^-_\alpha) \\
		&\geq 1 - Pr(\lvert A-B \rvert \geq \tilde{u}^-_\beta - u^-_\alpha) - \beta - \beta \quad \geq 1- 2\beta - 2e^{- \tfrac{(u^-_\alpha - \tilde{u}^-_\beta)^2}{2 (\sigma_T^{\widetilde{\Delta}})^2 }}
		\end{align*}
		By solving for $u^-_\alpha$ in  $2\beta + 2e^{- \tfrac{(u^-_\alpha - \tilde{u}^-_\beta )^2}{2 (\sigma_T^{\widetilde{\Delta}})^2 }} = \alpha$ under the constraint $u^-_\alpha < \tilde{u}^-_\beta$, we obtain~\eqref{eq:profBoundCor}.
	\end{proof}

\section{Gradient of $\widetilde{Z}_\x$ with respect to $\x$}
\label{sec:gradTildeZ}

We are interested in the approximating process for the posterior distribution of $Z$ conditioned on $(\doe_n,\mathbf{y}_n)$. If $a,\mathbf{b}$ are chosen as the posterior mean and the kriging weights respectively, then we can write $\widetilde{Z}$ as
\begin{align*}
\widetilde{Z}_\x &= \Lambda(\x)^T \begin{bmatrix}
\mathbf{y}_n \\ Z_G 
\end{bmatrix} \qquad \text{with} \\
\Lambda(\x) &= K_{n+\ell}^{-1} \left( \covkern_{n+\ell}(\x) + {H}_{n+\ell} K_{H} \left( \mathbf{h}(\x) - {H}_{n+\ell}^T K_{n+\ell}^ {-1}\covkern_{n+\ell}(\x) \right) \right) \\
\text{where } \quad  K_{H} &= \left( {H}_{n+\ell}^T K_{n+\ell}^ {-1} {H}_{n+\ell}  \right)^{-1}, \quad  K_{n+\ell} = \covkern(A_{n+\ell},A_{n+\ell}), \qquad  \covkern_{n+\ell}(\x) = \covkern(x,A_{n +\ell}), \\ 
{H}_{n+\ell} &= [h_j(A_{n+\ell})]_{j=1, \ldots, m} \in \R^{(n+\ell) \times m} \quad 
\text{and } \quad A_{n+\ell} = [\doe_n^T, G^T]^T \in \R^{(n+\ell)\times d}
\end{align*}

Then $\nabla_x \widetilde{Z}_\x = \nabla_x \Lambda(\x)^T\begin{bmatrix}
\mathbf{y}_n \\ Z_G .
\end{bmatrix}$ and it suffices to compute the gradient of $\Lambda$ and 
\begin{gather*}
K_{n+\ell} \nabla_x\Lambda(\x) = \left( \nabla_x\covkern(\x,A_{n+\ell}) + {H}_{n+\ell}  K_{H} \left( \nabla_x\mathbf{h}(\x) - {H}_{n+\ell}^T K_{n+\ell}^ {-1}\nabla_x\covkern(\x,A_{n+\ell}) \right) \right) \\
\text{with } \nabla_x\covkern(\x,A_{n+\ell}) = [\nabla_x\covkern(\x, a_1), \ldots, \nabla_x\covkern(\x, a_{n+\ell})]^T \in \R^{(n+\ell) \times d} \\ \text{ and } \nabla_x\mathbf{h}(\x) = [\nabla_x h_1(\x), \ldots, h_m(\x)] \in \R^{m\times d}.
\end{gather*}

\section{Full results on flooding test case}
\label{sec:full5dRes}

In this section, we report more details on the GP model used in the flooding test case  in Section~\ref{sec:motivating} and we show the profile extrema functions for all thresholds $\threshold_1,\ldots, \threshold_4$. 

In table~\ref{tab:5dmodels} we compare different GP models according to two metrics: $Q^2$ on leave-one-out predictions and log likelihood ($logLik$). According to those metrics the prior mean function in equation~\ref{eq:priorMean} (universal kriging, UK) results in better fits than a constant prior mean (ordinary kriging, OK). On the other hand the difference between using a smoothness parameter $\nu=3/2$ or $\nu=5/2$ is very small. Here we chose the parameter $\nu=3/2$ as it leads to standardized model residuals that have a distribution closer to the normal one.    

\begin{table}[!htbp] \centering 
	\caption{Comparison of different GP models for flooding test case. OK: constant mean; UK: mean function as in equation~\ref{eq:priorMean}. Best values in bold.} 
	\label{tab:5dmodels} 
	\begin{tabular}{@{\extracolsep{5pt}} ccccc} 
		\\[-1.8ex]\hline 
		\hline 
		& Matern 3/2, OK & Matern 3/2, UK & Matern 5/2, OK & Matern 5/2, UK \\ 
		\hline 
		$Q^2$ & $0.95$ & $\mathbf{0.96}$ & $0.94$ & $\mathbf{0.96}$ \\ 
		$logLik$ & $$-$1,406.02$ & $$-$\textbf{1,363.73}$ & $$-$1,410.71$ & $$-$1,365.32$ \\ 
		\hline
	\end{tabular} 
\end{table}   

If the model fit is worse, then we obtain coordinate profiles with larger confidence bands and thus more uncertainty. We checked this assumption by computing coordinate profile extrema on each of the models in Table~\ref{tab:5dmodels}. Figures~\ref{fig:ex5dcoord_mat32OK} and~\ref{fig:ex5dcoord_mat52OK} show the first two coordinate profile plots for Mat\'ern $3/2$, OK and Mat\'ern $5/2$, OK. A quick glance already shows that the confidence bands are larger than in Figure~\ref{fig:ex5dUQfun_approx}. We further compared the integrated inter-quantile range for each model and we ranked them in increasing order. The average ranks over all coordinates for profile sup and profile inf are shown in Table~\ref{tab:5diqrModels}. Note how the chosen model (Mat\'ern $3/2$ model with UK prior mean) has an average rank $1.8$ for $\Psup f$ and $1.4$ for $\Pinf f$.

\begin{table}[!htbp] \centering 
	\caption{Average rank of integrated inter-quantile range for coordinate profile extrema.} 
	\label{tab:5diqrModels} 
	\begin{tabular}{@{\extracolsep{5pt}} ccccc} 
		\\[-1.8ex]\hline 
		\hline 
		& Matern 3/2, OK & Matern 3/2, UK & Matern 5/2, OK & Matern 5/2, UK \\ 
		\hline 
		$\Psup f$ & $2.4$ & $\mathbf{1.8}$ & $3$ & $2.8$ \\ 
		$\Pinf f$ & $3.8$ & $\textbf{1.4}$ & $2.8$ & $2$ \\ 
		\hline
	\end{tabular} 
\end{table}  

\begin{figure}
	\begin{minipage}{0.5\textwidth}
		\centering
		\includegraphics[width=\linewidth]{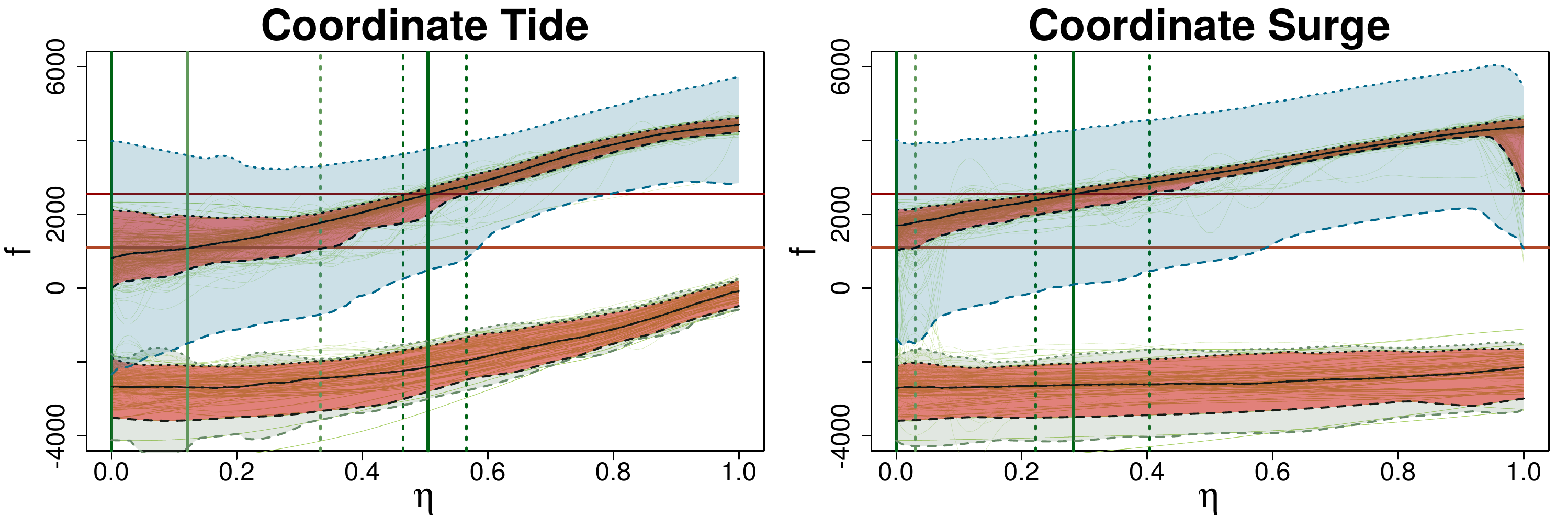}
		\caption{Coordinate profiles on GP model with Mat\'ern $\nu=3/2$ and constant mean.}
		\label{fig:ex5dcoord_mat32OK}	
	\end{minipage} \hfill\hspace{0.1cm}
	\begin{minipage}{0.5\textwidth}
		\centering
		\includegraphics[width=\linewidth]{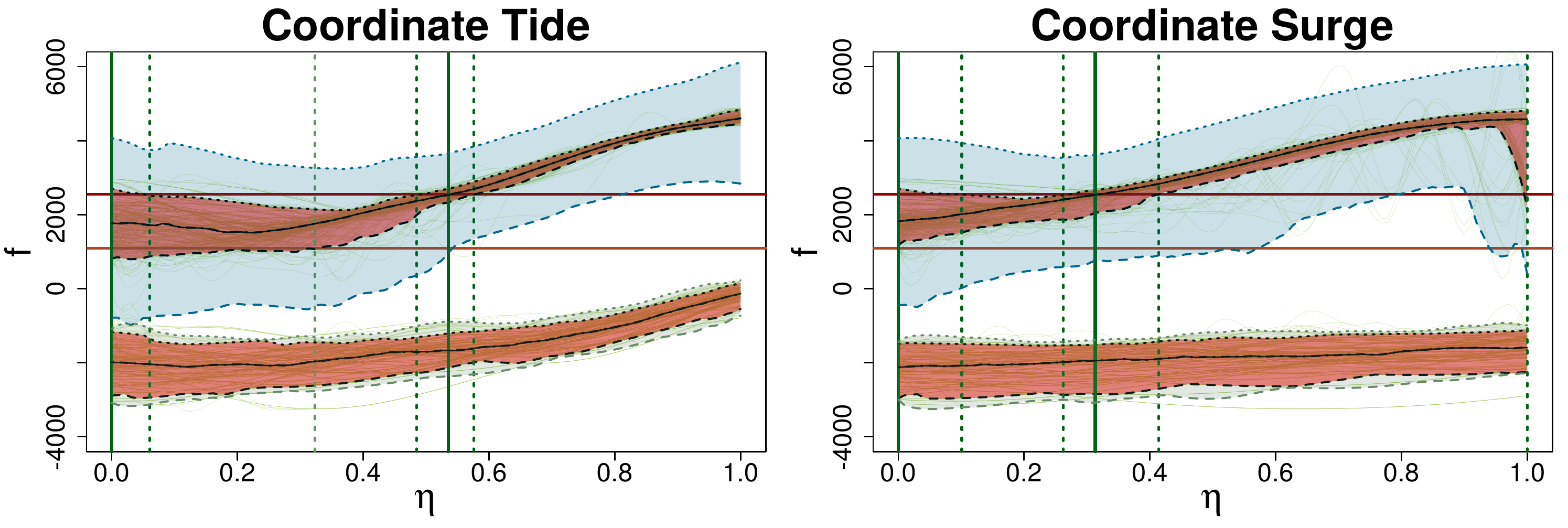}
		\caption{Coordinate profiles on GP model with Mat\'ern $\nu=5/2$ and constant mean.}
		\label{fig:ex5dcoord_mat52OK}
	\end{minipage}	
\end{figure}


Let us now focus on the chosen model: Mat\'ern $\nu=3/2$ with prior mean function as described in equation~\ref{eq:priorMean}. Table~\ref{tab:5dSummary} summarizes the intervals selected with the coordinate profile extrema functions. Figure~\ref{fig:ex5dUQfun_full} shows the coordinate profile extrema functions, the $5\%, 95\%$ quantiles, the boundary of the subsets selected (vertical lines) by posterior median and quantiles. The values of the bound with confidence level $\alpha=0.1$ are shown with the sky blue ($\Psup_i Z$) and sea green ($\Pinf_i Z$) shaded regions. The bound on the profile inf function is almost overlapping with the $90\%$ confidence intervals in red. The bound on the profile sup function instead  provides an higher quantile and identifies as region of possible non excursion $\Tide \in [0,0.19]$ above $\threshold_4$. For the other coordinates the bound does not provide information on regions of non excursion.

\begin{table}[h]
	\caption{\label{tab:5dSummary} Regions excluded with profile extrema functions on $5$-d test case. Interval defined in the top line, for each threshold the interval for the boundary computed from the $90\%$ approximate confidence intervals for the profiles is reported.}
	\begin{tabular}{ l  c  c c c c }
		\\[-1.8ex]\hline 
		\hline 
		& $\Tide \in [0,\lambda_{\Tide,k}]$	& $\Surge \in [0,\lambda_{\Surge,k}] $ & $\PhiVar \in [0,\lambda^a_{\PhiVar,k}] \cup [\lambda^b_{\PhiVar	,k},1]$ & $\tPlus$ & $\tMinus$   \\
		\hline 
		$\threshold_1$  & $\lambda_{\Tide,1} \in [0.17,0.26]$  & $-$ & $\lambda^a_{\PhiVar,1}=0, \lambda^b_{\PhiVar,1}=1$ & $-$ & $-$ \\
		$\threshold_2$ & $\lambda_{\Tide,2} \in [0.24,0.32]$  & $\lambda_{\Surge,2} \in [0,0.06]$ & $\lambda^a_{\PhiVar,2} \in [0,0.11], \lambda^b_{\PhiVar,2}=1$ & $-$ & $-$ \\
		$\threshold_3$ & $\lambda_{\Tide,3} \in [0.33,0.40]$  & $\lambda_{\Surge,2} \in [0,0.14]$ & $\lambda^a_{\PhiVar,3} \in [0,0.12], \lambda^b_{\PhiVar,3}=1$ & $-$ & $-$ \\
		$\threshold_4$ & $\lambda_{\Tide,4} \in [0.52,0.60]$  & $\lambda_{\Surge,2} \in [0.28,0.43]$ & $\lambda^a_{\PhiVar,4} \in [0,0.16],$ $\lambda^b_{\PhiVar,4} =1$  & $-$ & $-$ \\
		\hline
	\end{tabular}
\end{table}

\begin{figure}[h!]
	\centering
	\includegraphics[width=0.8\linewidth]{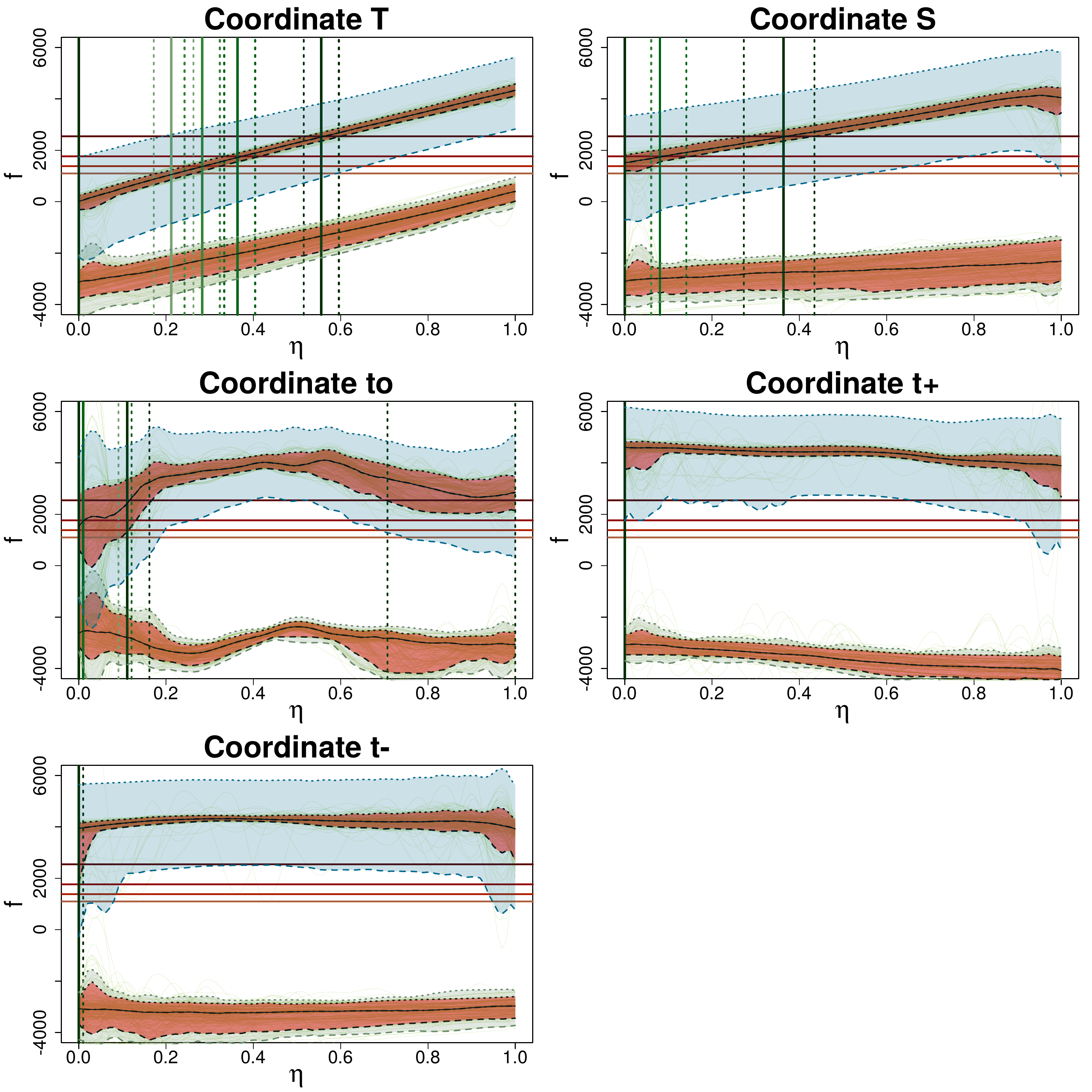}
	\caption{Coordinate profiles for $5$-dimensional test case. Thresholds (red, light to dark for $\threshold_1, \ldots, \threshold_4$), regions selected by profile extrema on mean (vertical green lines) and regions selected by profile extrema on $95\% - 5\%$ quantiles (vertical green lines, dashed). Bound on the confidence region ($\alpha=0.05$) for $\Psup_i Z$ (sky blue, shaded) and for $\Pinf_i Z$ (sea green, shaded).}
	\label{fig:ex5dUQfun_full}	
\end{figure}

Section~\ref{subsec:bound} introduced the integrated variance of the difference, equation~\eqref{eq:integratedVarDiff}, as an indicator for the bound tightness. Figure~\ref{fig:integVariance5d} shows the value $I(\sigma_T^{\widetilde{\Delta}})^2$ averaged over all dimensions for $\Psup$ and $\Pinf$ versus the number of pilot points $\ell$. Note the ``diminishing returns'' as $\ell$ increases. In particular, here we chose $\ell=300$ for computational reasons.

\begin{figure}[h!]
	\centering
	\includegraphics[width=0.5\linewidth]{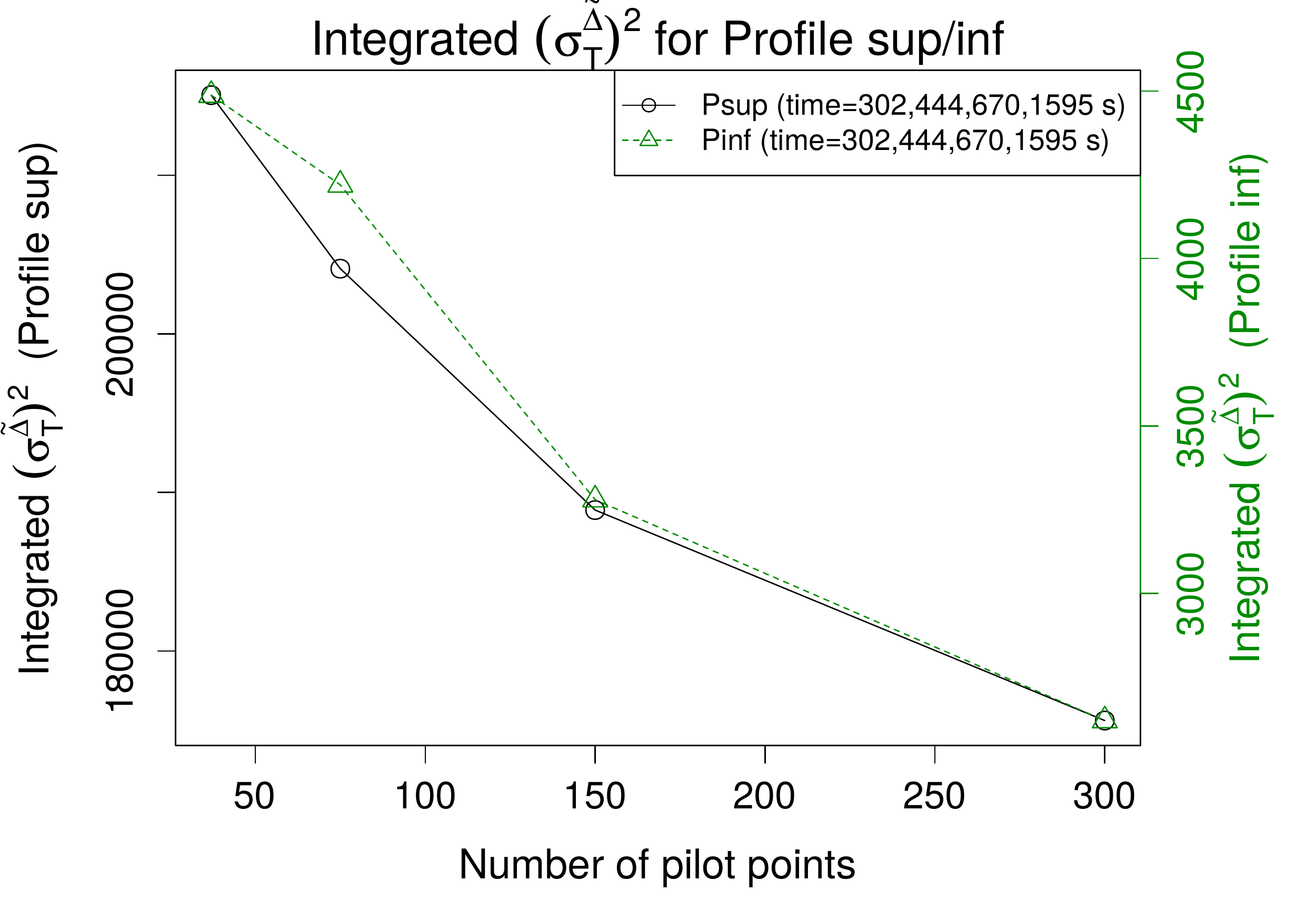}
	\caption{Average integrated variance of the difference, equation~\eqref{eq:integratedVarDiff}, for profile inf (dashed, triangles) and profile sup (solid, circles) versus $\ell$ $(37,75,150,300)$. Computational time for $(\sigma_T^{\tilde{\Delta}})^2$ in legend. Profile coordinates, 5-d test case.}
	\label{fig:integVariance5d}	
\end{figure}

\begin{figure}
	\begin{minipage}{0.32\textwidth}
		\centering
		\includegraphics[width=1.1\linewidth]{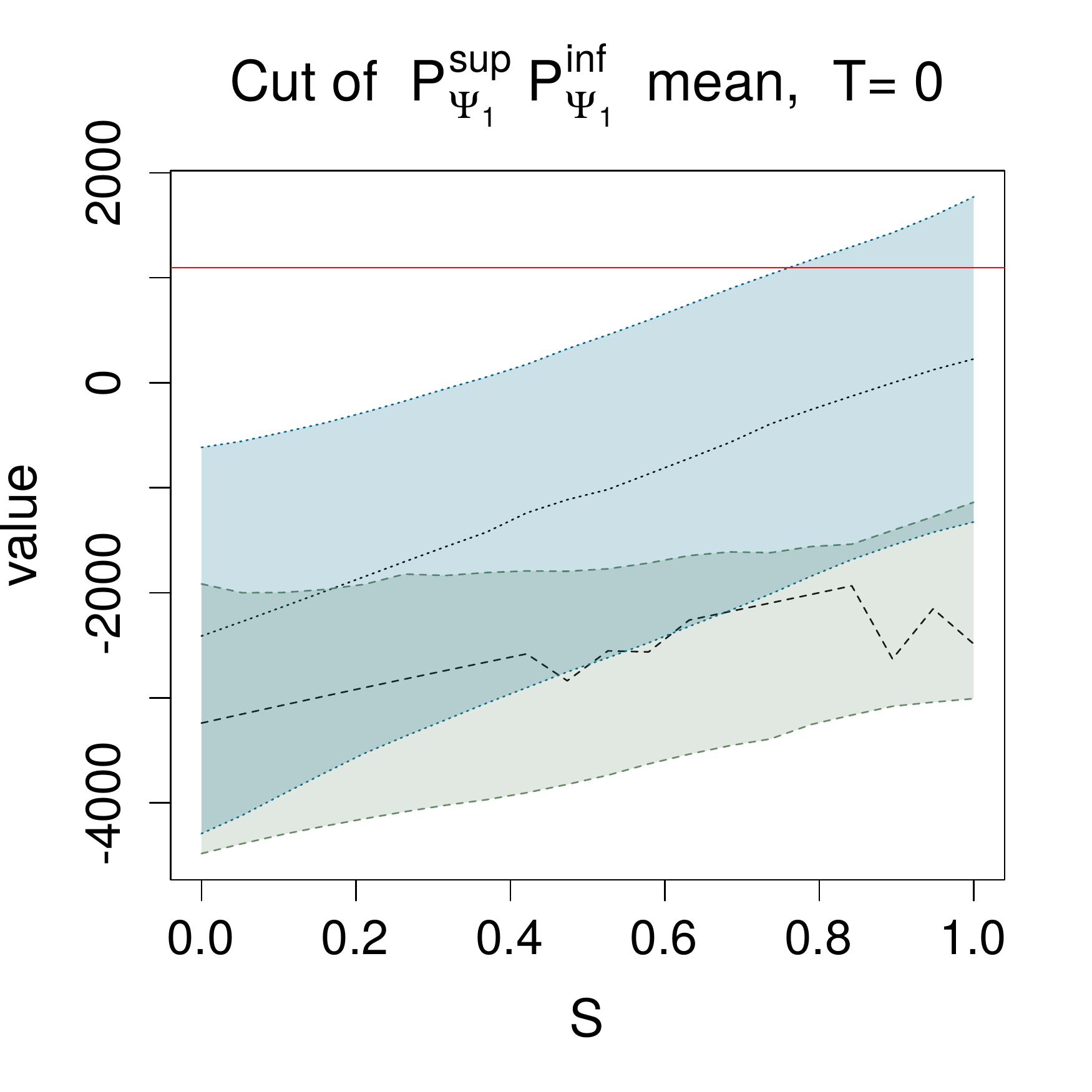}
	\end{minipage}
	\begin{minipage}{0.32\textwidth}
		\centering
		\includegraphics[width=1.1\linewidth]{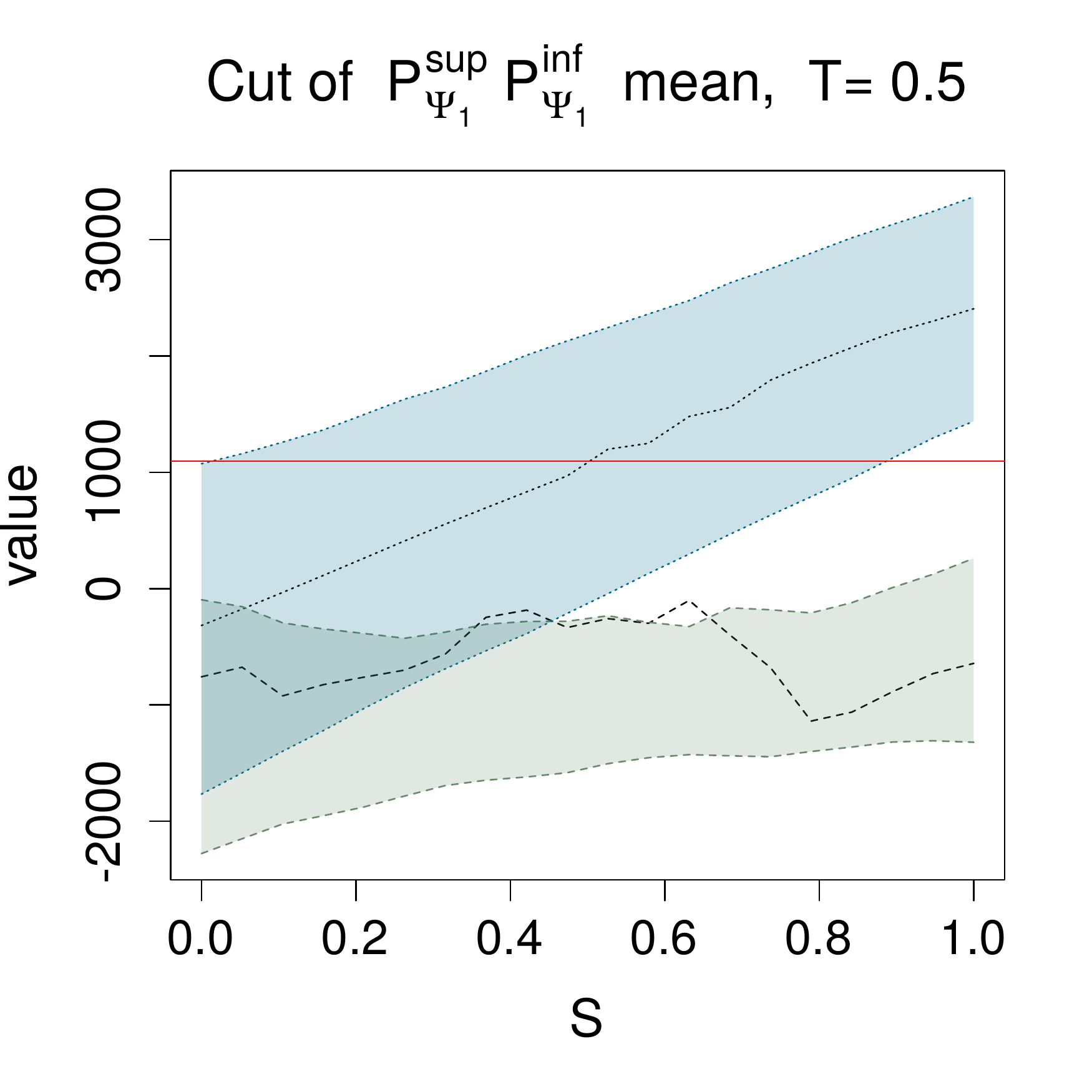}
	\end{minipage}
	\begin{minipage}{0.32\textwidth}
		\centering
		\includegraphics[width=1.1\linewidth]{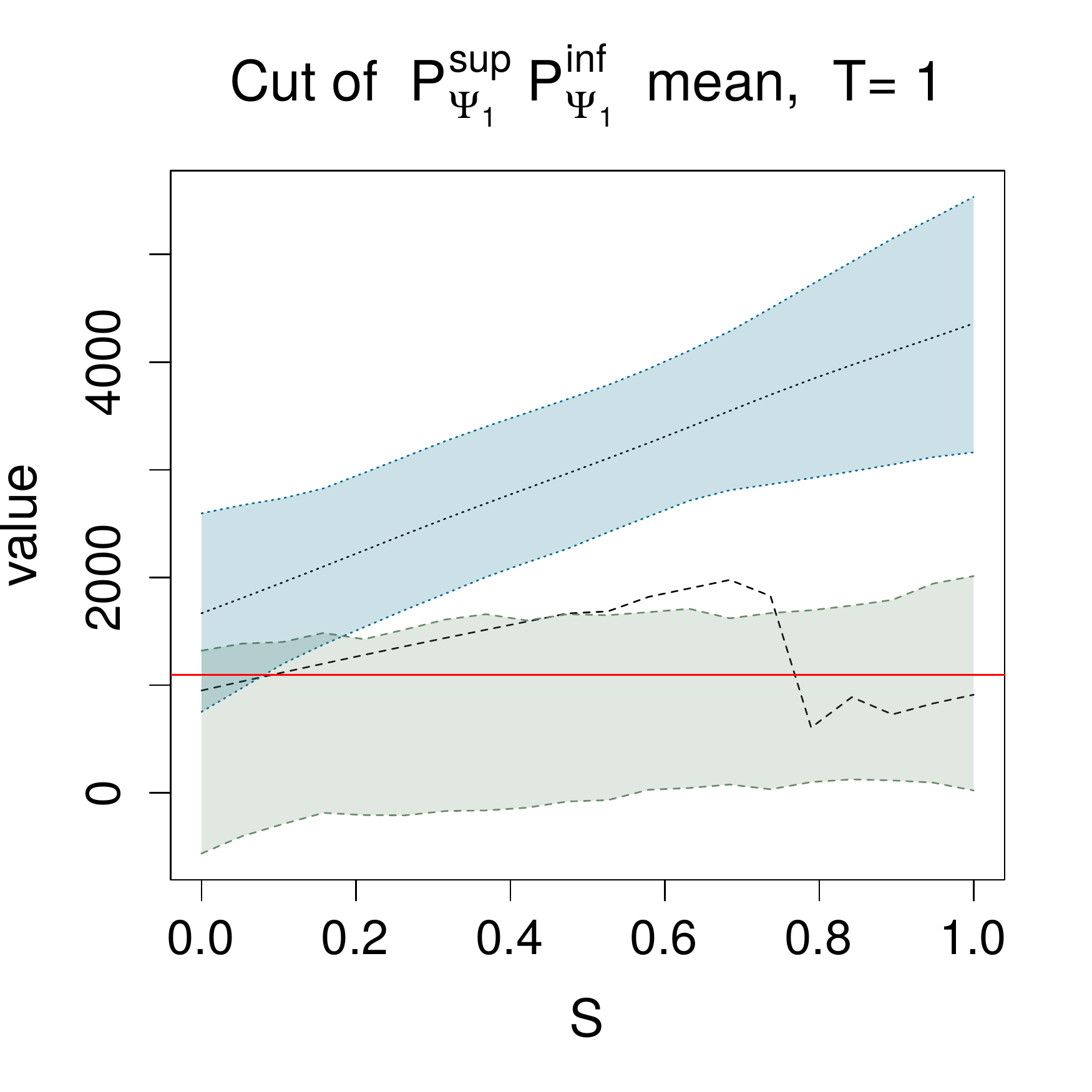}
	\end{minipage}
	\caption{Cut with $\Tide=0,0.5,1$ of $\Psup_{\Psi_1},\Pinf_{\Psi_1}$ on the profile mean, see also Figure~\ref{fig:ex5dbivProf}.}
	\label{fig:ex5dbivPE_TScuts}
\end{figure}

\begin{figure}[h!]
	\centering
	\includegraphics[width=0.85\linewidth]{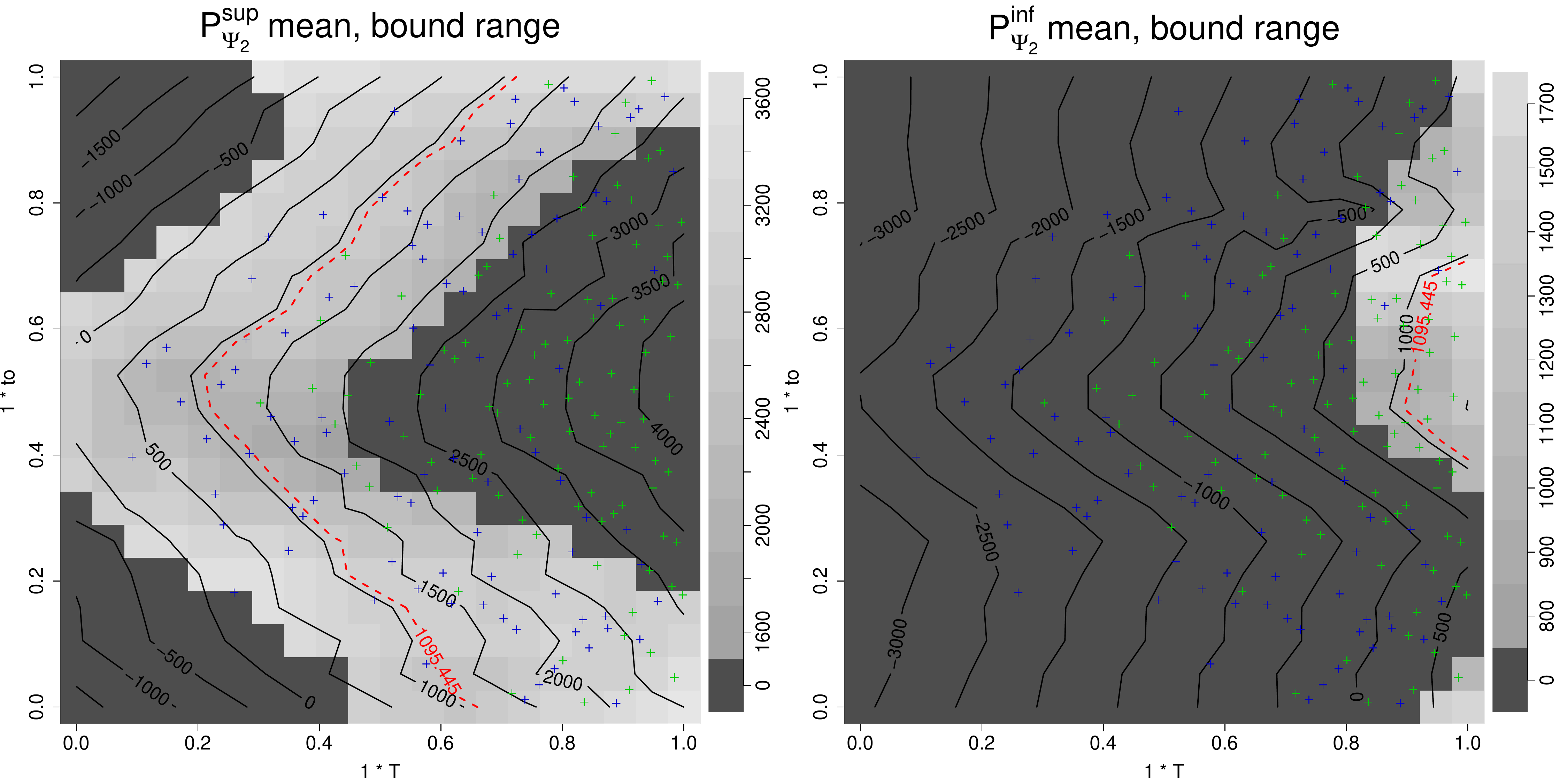}
	\caption{Bivariate profile extrema for $5$-dimensional test case with $\Tide,\PhiVar$. Excursion threshold $\threshold_1$ in red, bivariate profile mean values in contour lines, $\doe_n$ locations as crosses (green if above $\threshold_1$), background color denotes the upper-lower bound range (equation~\eqref{eq:profBoundCor}).}
	\label{fig:ex5dbivProfTP}	
\end{figure}

\begin{figure}
	\begin{minipage}{0.32\textwidth}
		\centering
		\includegraphics[width=1.1\linewidth]{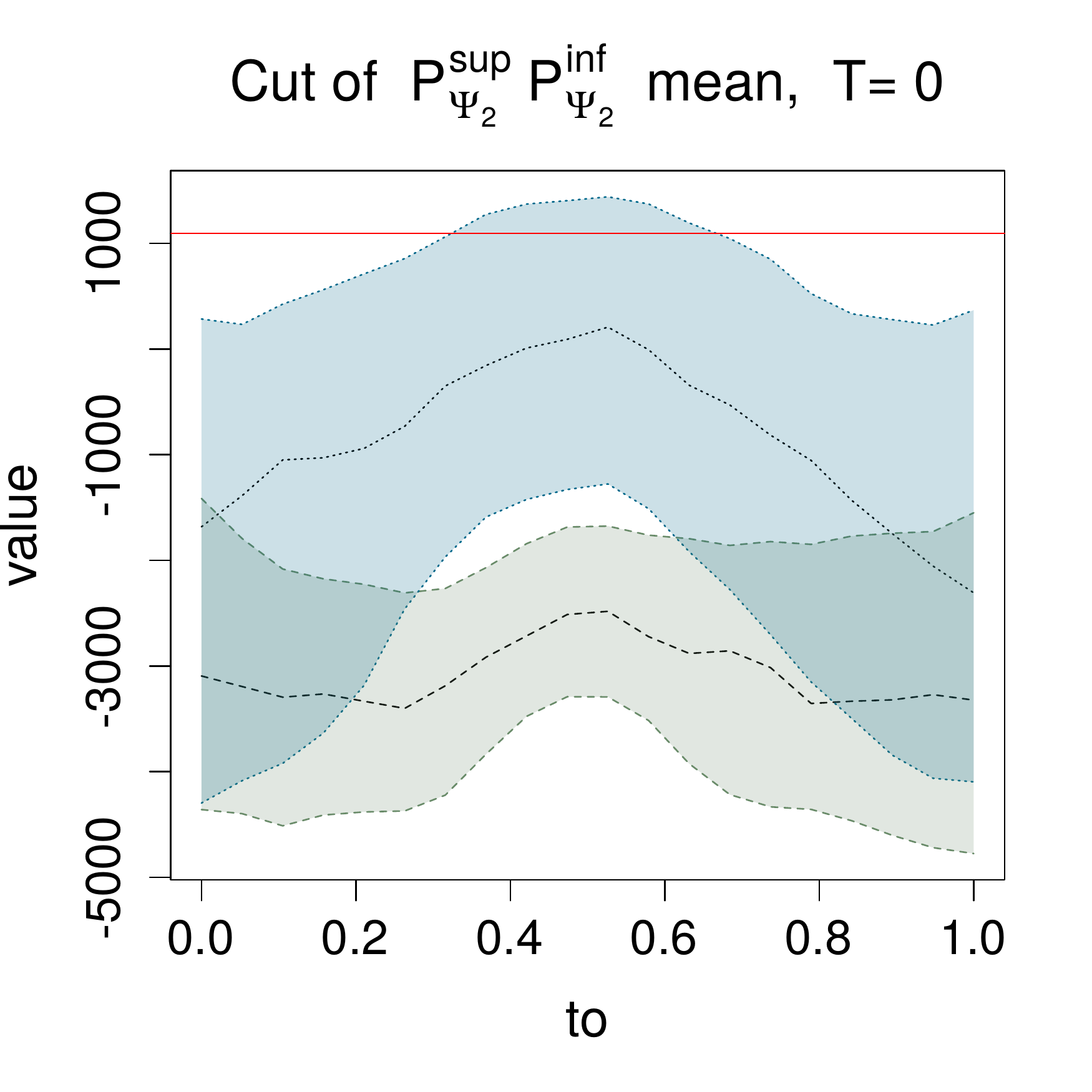}
	\end{minipage}
	\begin{minipage}{0.32\textwidth}
		\centering
		\includegraphics[width=1.1\linewidth]{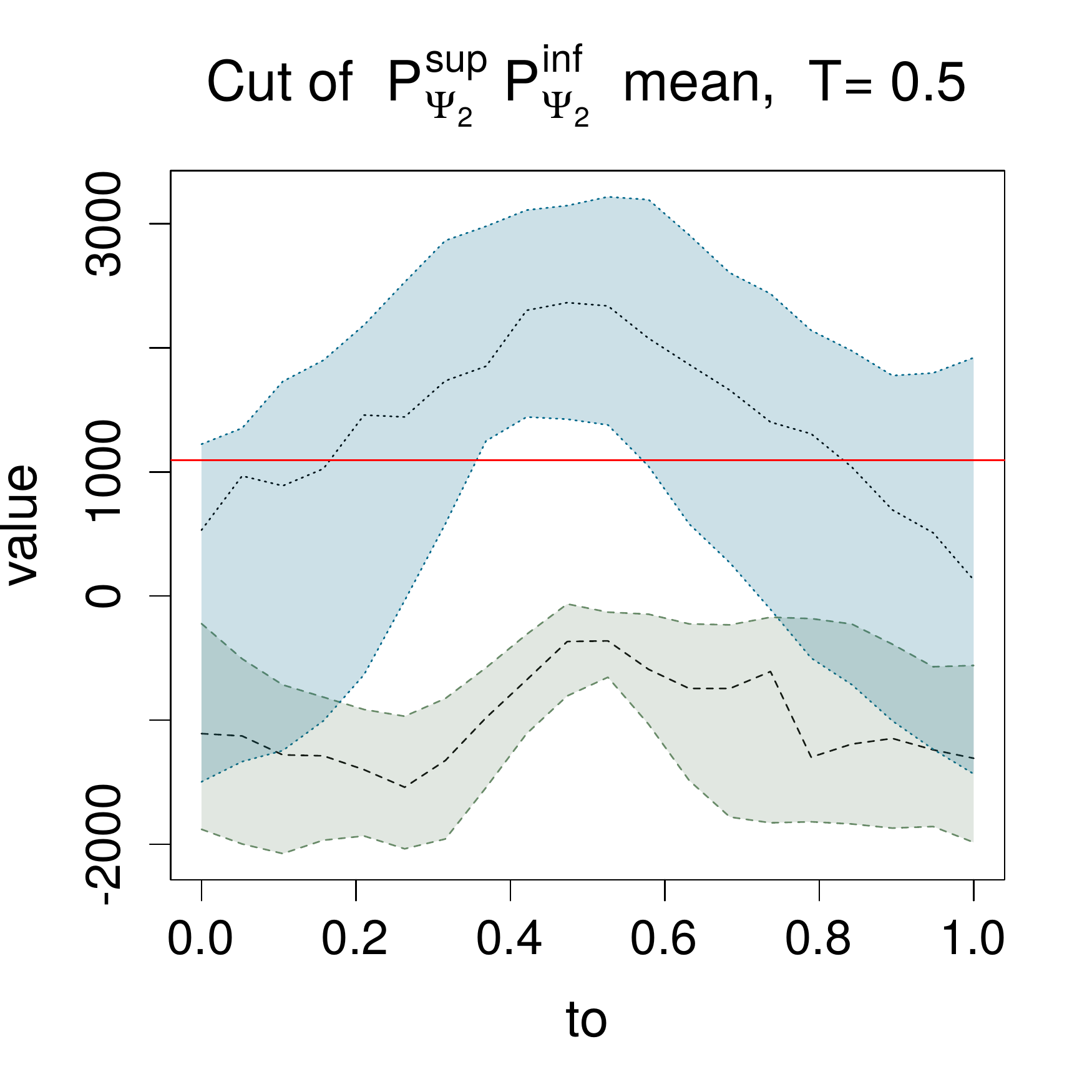}
	\end{minipage}
	\begin{minipage}{0.32\textwidth}
		\centering
		\includegraphics[width=1.1\linewidth]{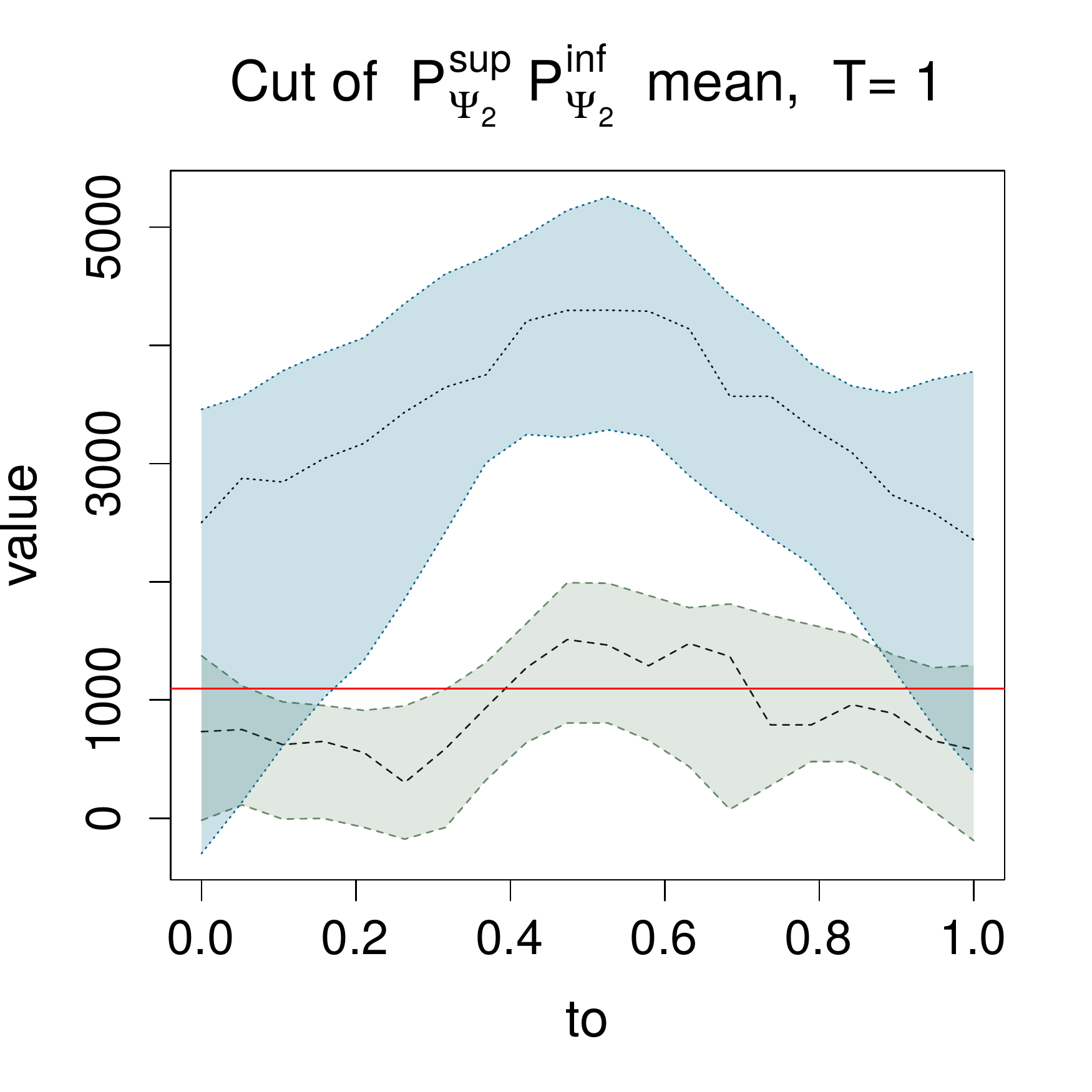}
	\end{minipage}
	\caption{Cut with $\Tide=0,0.5,1$ of $\Psup_{\Psi_2},\Pinf_{\Psi_2}$ on the profile mean, see also Figure~\ref{fig:ex5dbivProfTP}.}
	\label{fig:ex5dbivPE_TPcuts}
\end{figure}

Figure~\ref{fig:ex5dbivPE_TScuts} shows the cuts at $\Tide=0,0.5,1$ of the bivariate profile functions on the mean shown in Figure~\ref{fig:ex5dbivProf}(b) with their respective upper and lower bounds. Such plots can be used to zoom in particular regions of the bivariate profiles. They do not plot more information than bivariate profile maps, however they are simpler to read. For example, Figure~\ref{fig:ex5dbivPE_TScuts} shows that $\{\x \in D: \Tide=0 \text{ and } \Surge < 0.79 \}$ is a non-excursion region with high probability.

Figure~\ref{fig:ex5dbivProfTP} shows the bivariate profile extrema for the variables $\Tide, \PhiVar$. In particular the left-most map allows us to exclude all input regions on the left of the red dotted curve. In a symmetric way, the right-most plot tells us that values for $\Tide,\PhiVar$ in the region on the right-hand side of the red dotted line lead to an excursion. Also in this case the uncertainty quantification adds more insights on the result. The profile inf shows much uncertainty around the threshold, therefore while region of interest could be selected the uncertainty must be accounted for. The profile sup also tells us that close to the threshold the uncertainty is still high, however the top and bottom left triangles are not in the excursion with very little uncertainty. A very conservative assessment should select only this input region as non-excursion. Figure~\ref{fig:ex5dbivPE_TPcuts} shows the cuts at $\Tide=0,0.5,1$ of these bivariate profile functions. As in the other case such plots confirm information shown in bivariate maps. For example, we can clearly see that for $\Tide=0$ and $\PhiVar \in [0,0.37] \cup [0.63,1]$ there is no excursion with high probability. 

\section{An indicator of tightness for the bound}
\label{sec:deltaTex2d}

\begin{table}[h]
	\caption{\label{tab:paramSummary} Summary of profile extrema parameters.}
	\begin{tabular}{ l  l  c }
		\\[-1.8ex]\hline 
		\hline 
		Symbol & Meaning  & Section\\
		\hline 
		$n$ & Number of expensive computer experiment simulations & \ref{sec:intro} \\
		$d$ & Dimension of input space & \ref{sec:intro} \\
		$p$ & Dimension of profile extrema projection, usually $1$ or $2$ & \ref{sec:profiles} \\
		$k$ & Number of points used for deterministic approximation of $\Psup,\Pinf$ & \ref{subsec:approxProf} \\
		$\ell$ & Number of pilot points for GP approximation in~\eqref{eq:tildeZ} & \ref{subsec:approxReals} \\
		$s$ & Number of posterior approximate GP realizations & \ref{subsec:approxReals} \\
		\hline		
	\end{tabular}
\end{table}

\begin{figure}[h!]
	\begin{minipage}{0.5\textwidth}
		\centering
		\includegraphics[width=\linewidth]{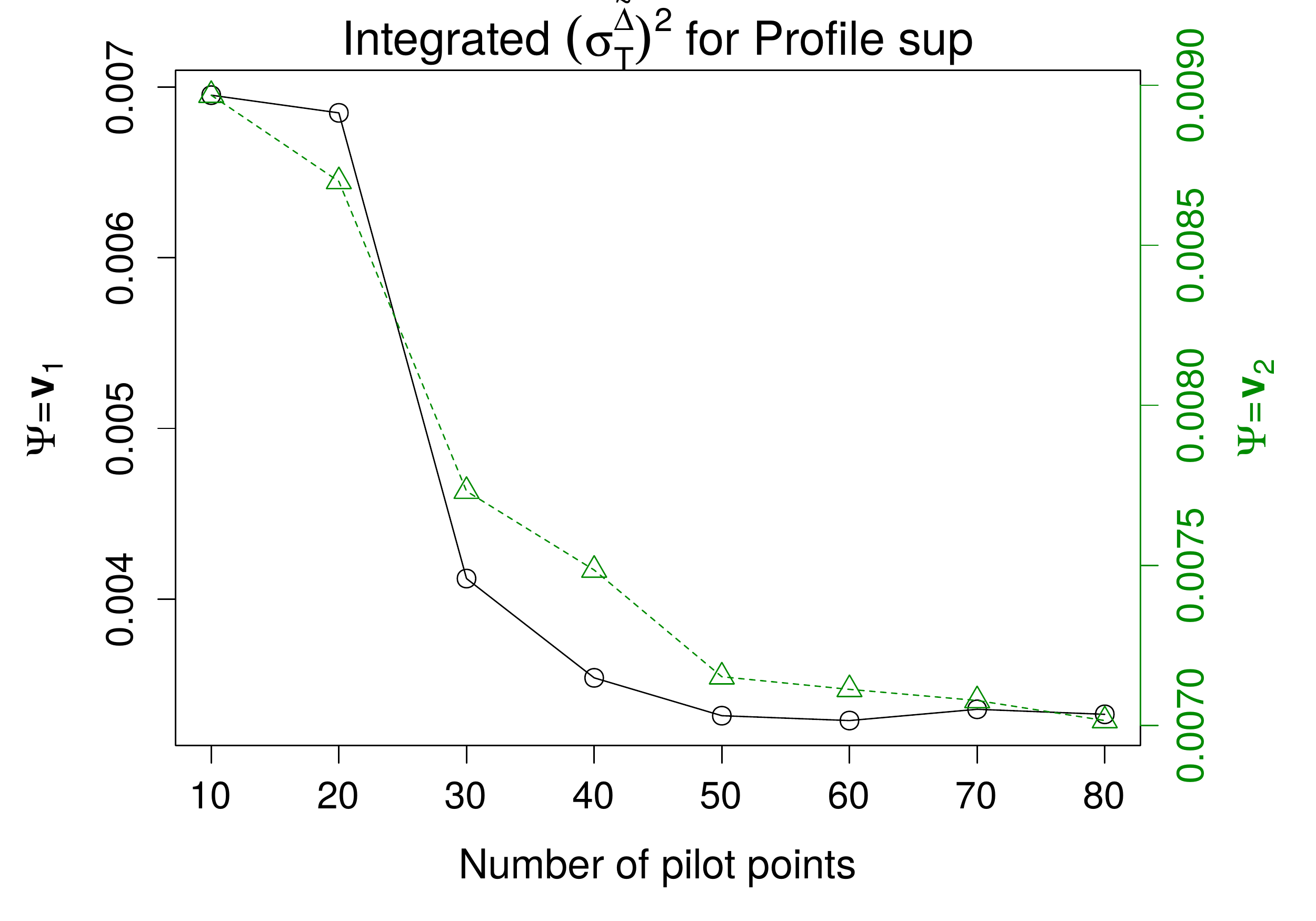}
	\end{minipage}
	\begin{minipage}{0.5\textwidth}
		\centering
		\includegraphics[width=\linewidth]{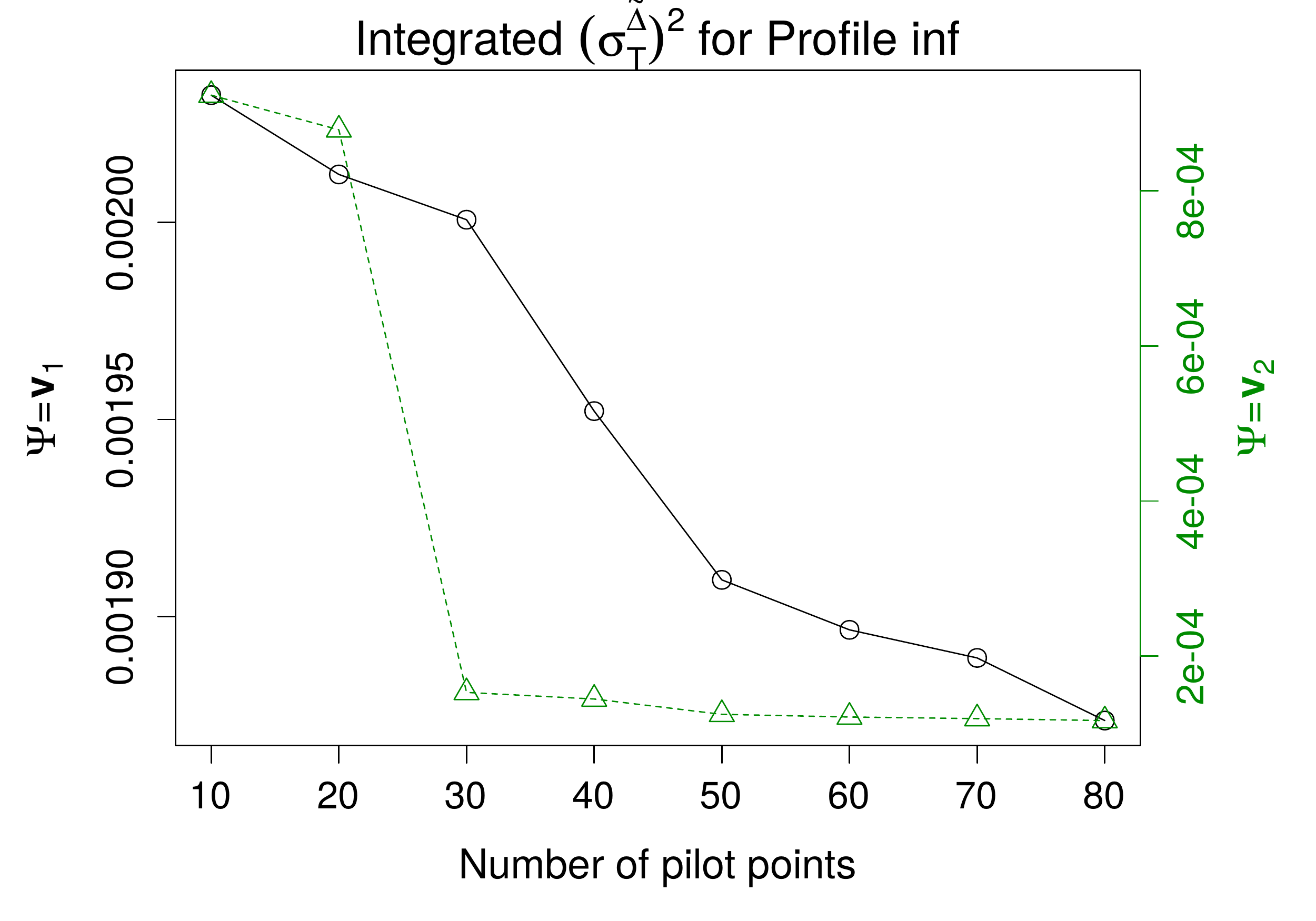}
	\end{minipage}
	\caption{Values of $I(\sigma^{\widetilde{\Delta}}_T)^2(\ell)$, Eq.~\eqref{eq:integratedVarDiff}, for $\Psi=\mathbf{v}_1$ (solid, circles) and $\Psi=\mathbf{v}_2$ (dashed, triangles) as functions of $\ell$ . Profile sup (left) and Profile inf (right).}
	\label{fig:ex2dDeltaTcomparison}
\end{figure}

The bounds for the uncertainty on $\Psup_\Psi Z$ are computed with equation~\eqref{eq:profBoundCor}. For a given $\ell$, the tightness of this approximation is controlled by $(\sigma^{\widetilde{\Delta}}_T)^2(\eta,\ell)$, i.e. the sup of the difference variance. This is a non-negative function of $\eta$, as the difference depends on where we evaluate $\Psup_\Psi Z (\eta)$. The integrated variance $I(\sigma^{\widetilde{\Delta}}_T)^2(\ell)$, Equation~\eqref{eq:integratedVarDiff}, is a summary of this quantity. 
Since  $(\sigma^{\widetilde{\Delta}}_T)^2(\eta,\ell)$ is non-negative for each $\eta \in E_\Psi$, a smaller integral implies less variability and a tighter bound.
We can control $I(\sigma^{\widetilde{\Delta}}_T)^2(\ell)$ with the number of pilot points $\ell$ chosen for the approximate realizations.  Figure~\ref{fig:ex2dDeltaTcomparison} shows the estimated integral $I(\sigma^{\widetilde{\Delta}}_T)^2(\ell)=\int_{E_{\mathbf{v}_i}}(\sigma^{\widetilde{\Delta}}_T)^2(\eta,\ell)d\eta$, $i=1,2$ for the synthetic example introduced in Section~\ref{subsec:analExUQ} with $n=90$ as a function of $\ell$. More pilot points lead to a smaller integral for each coordinate and for both profile extrema, however the rate of decrease is test case dependent.

\section{Uni and bivariate profiles, a combined approach}
\label{sec:bivEx3d}

In this section we consider the following function
\begin{equation}
\gamma(\mathbf{x}) = \sin\left( a \mathbf{v_1}^T \mathbf{x} +b\right) + \cos\left(c \mathbf{v_2}^T \mathbf{x} +d \right) +\sin \left( e \mathbf{v_3}^T \mathbf{x} + f\right) -1.5 \qquad \mathbf{x} \in [0,1]^3, \ \mathbf{v_1}, \mathbf{v_2}, \mathbf{v_3}   \in \mathbb{R}^3
\label{eq:analExample3d}
\end{equation}
where $a,b,c,d,e,f \in \R$ and 
\begin{equation*}
\mathbf{v_1} = \begin{bmatrix}
\sin(\theta)\cos(\phi) \\
\sin(\theta)\sin(\phi) \\
\cos(\theta)
\end{bmatrix}, \ \mathbf{v_2} = \begin{bmatrix}
\cos(\theta)\cos(\phi) \\
\cos(\theta)\sin(\phi) \\
-\sin(\theta)
\end{bmatrix}, \ \mathbf{v_3} =\begin{bmatrix}
-\sin(\phi) \\
\cos(\phi) \\
0
\end{bmatrix} 
\end{equation*}
with $\theta= \pi/4$, $\phi=\pi/4$. We fix $[a,b,c,d,e,f]=[1,0,10,0,1,0]$ and we study the excursion set above $\threshold=0$. Figure~\ref{fig:ex3dFunction} shows three cuts corresponding to the planes $z=0,0.49,1$. 
\begin{figure}
	\begin{minipage}{0.32\textwidth}
		\centering
		\includegraphics[width=1.1\linewidth]{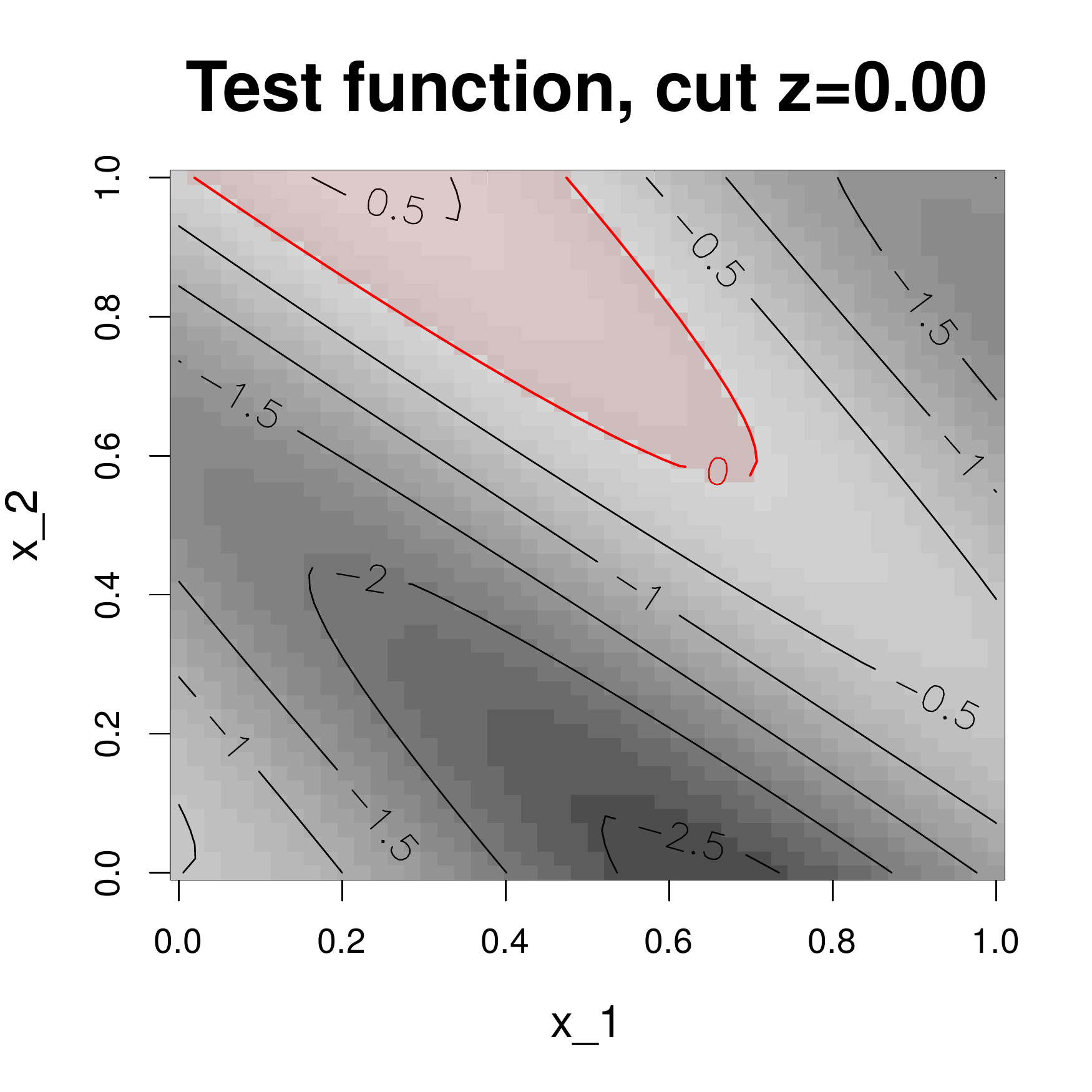}
	\end{minipage}
	\begin{minipage}{0.32\textwidth}
		\centering
		\includegraphics[width=1.1\linewidth]{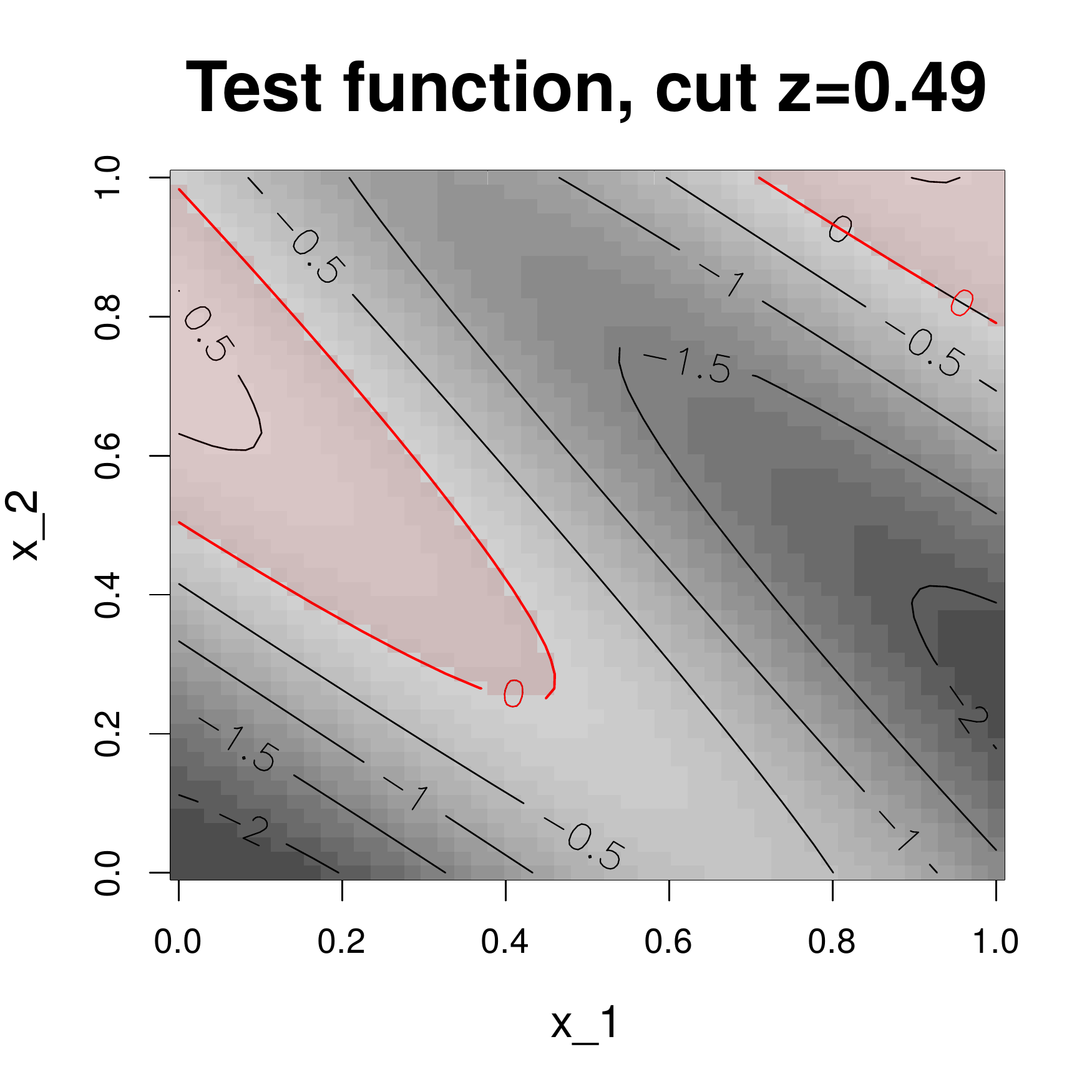}
	\end{minipage}
	\begin{minipage}{0.32\textwidth}
		\centering
		\includegraphics[width=1.1\linewidth]{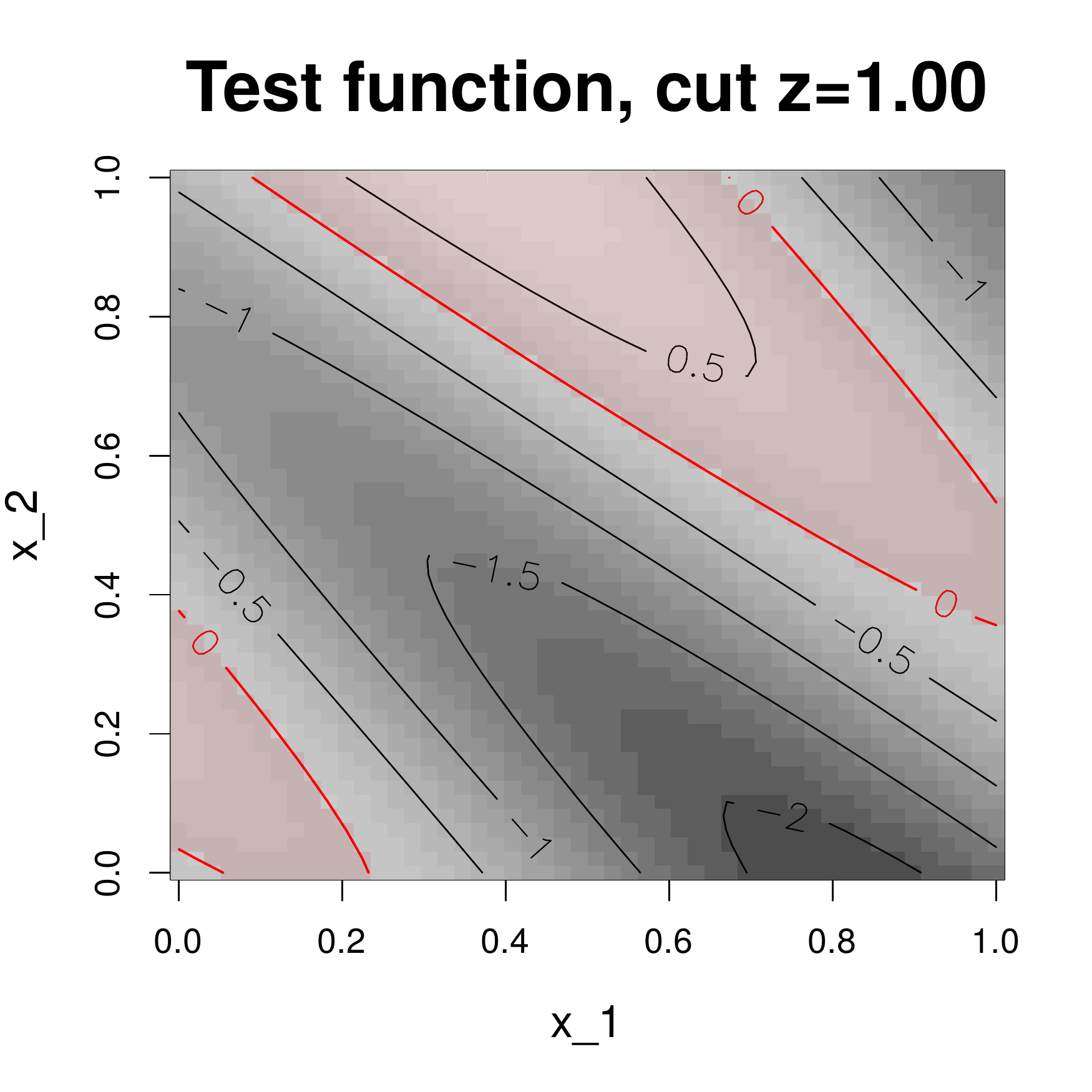}
	\end{minipage}
	\caption{Cuts of the function in Equation~\eqref{eq:analExample3d}.}
	\label{fig:ex3dFunction}
\end{figure}

Since the excursion set sits along directions which are oblique with respect to the canonical coordinates, the coordinate profile extrema do not provide information on this set. Figure~\ref{fig:ex3dCoord} shows that $\Psup_i f, i=1,2,3$ are always above the threshold and $\Pinf_i f, i=1,2,3$ are always below the threshold, thus we cannot exclude any region. In such situations we propose two approaches: compute oblique profiles according to more informative directions or compute bivariate coordinate profiles. Figure~\ref{fig:ex3dOblique} show the oblique profile extrema along the generating directions $\mathbf{v_1}, \mathbf{v_2}, \mathbf{v_3}$. This plot excludes part of the input space by considering the planes delimited by the interceptions between the profile sup and the threshold,~e.g. all points between the planes $\mathbf{v_1}\x =0$ and $\mathbf{v_1}\x =0.31$.

\begin{figure}
	\begin{minipage}{0.49\textwidth}
		\centering
		\includegraphics[width=\linewidth]{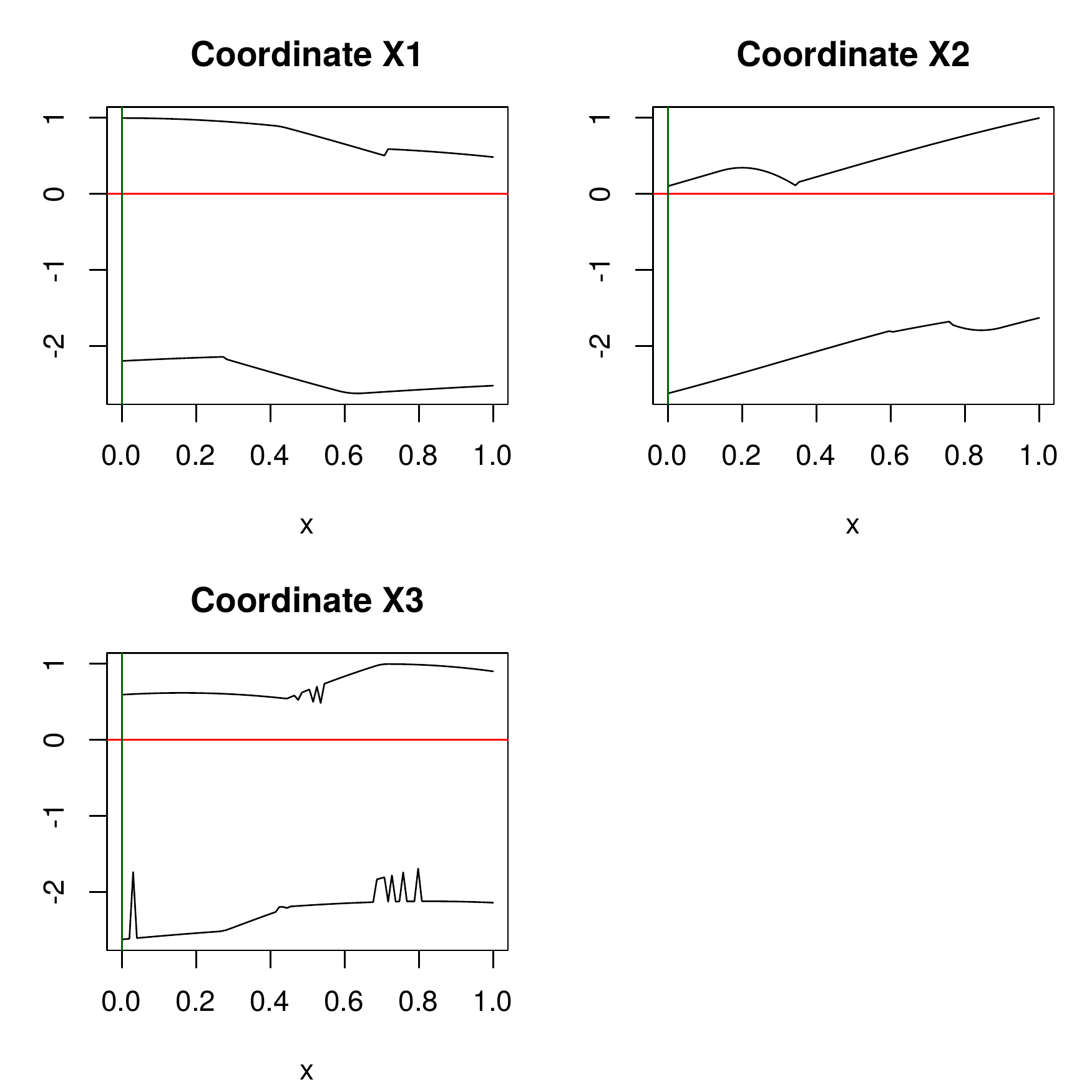}
		\caption{Coordinate profile extrema for the function in Equation~\eqref{eq:analExample3d}.}
		\label{fig:ex3dCoord}
	\end{minipage}\hspace{0.1cm}
	\begin{minipage}{0.49\textwidth}
		\centering
		\includegraphics[width=\linewidth]{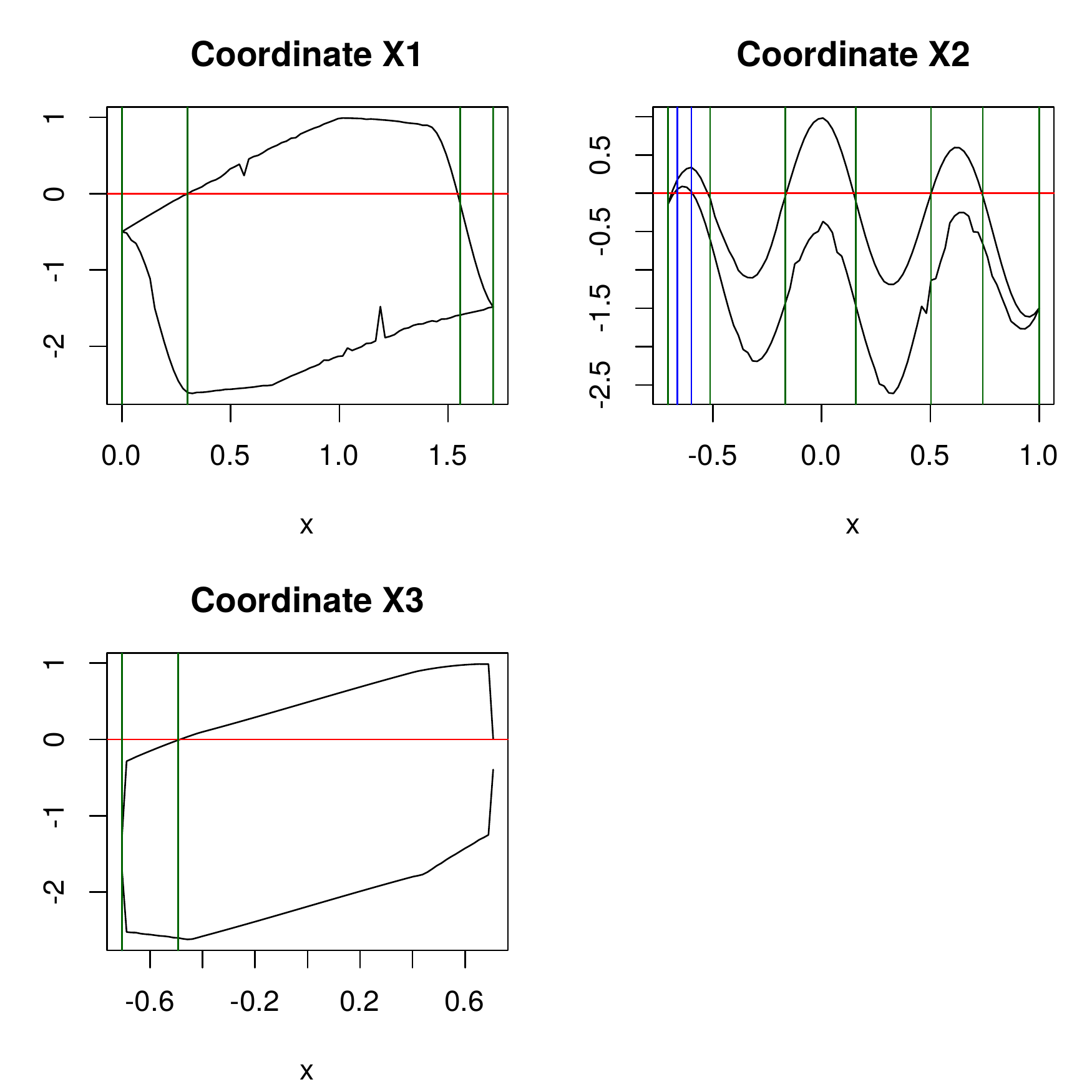}
		\caption{Profile extrema for the function in Equation~\eqref{eq:analExample3d} along the directions $\mathbf{v_1},\mathbf{v_2},\mathbf{v_3}$.}
		\label{fig:ex3dOblique}
	\end{minipage}
\end{figure}

The second option involves computing the profile extrema with projections on $2$ dimensional subspaces. Here we consider all combinations of bivariate projections on canonical axes. Figures~\ref{fig:ex3dBivCoord_sup},~\ref{fig:ex3dBivCoord_inf} show the maps obtained for the $\Psup_\Psi f, \Pinf_\Psi f$ for the three matrices build by combining the canonical directions. 
\begin{figure}
	\begin{minipage}{0.49\textwidth}
		\centering
		\includegraphics[width=\linewidth]{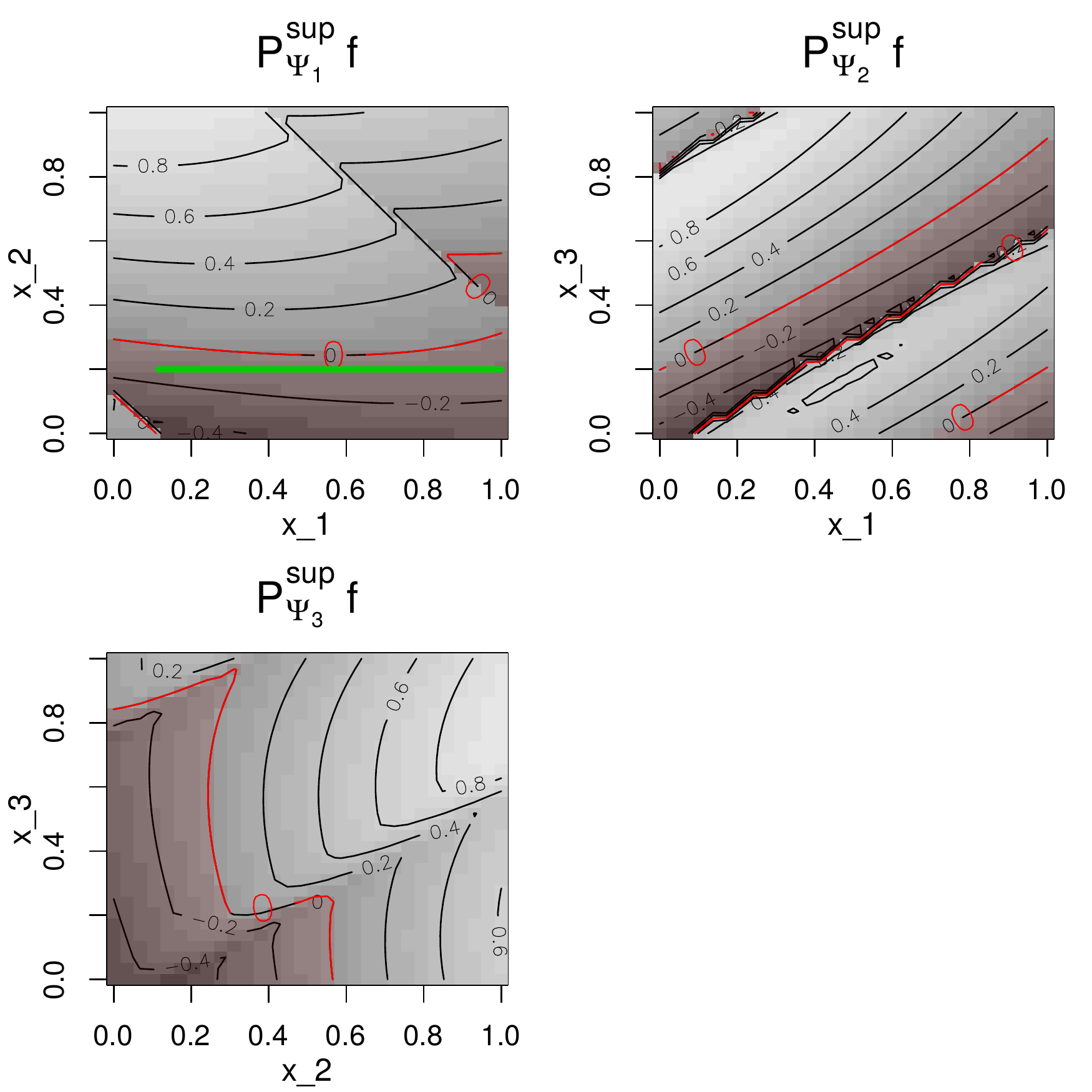}
		\caption{Bivariate profile sup for the function in Equation~\eqref{eq:analExample3d} along projections on combinations of three canonical axes.}
		\label{fig:ex3dBivCoord_sup}
	\end{minipage}\hspace{0.1cm}
	\begin{minipage}{0.49\textwidth}
		\centering
		\includegraphics[width=\linewidth]{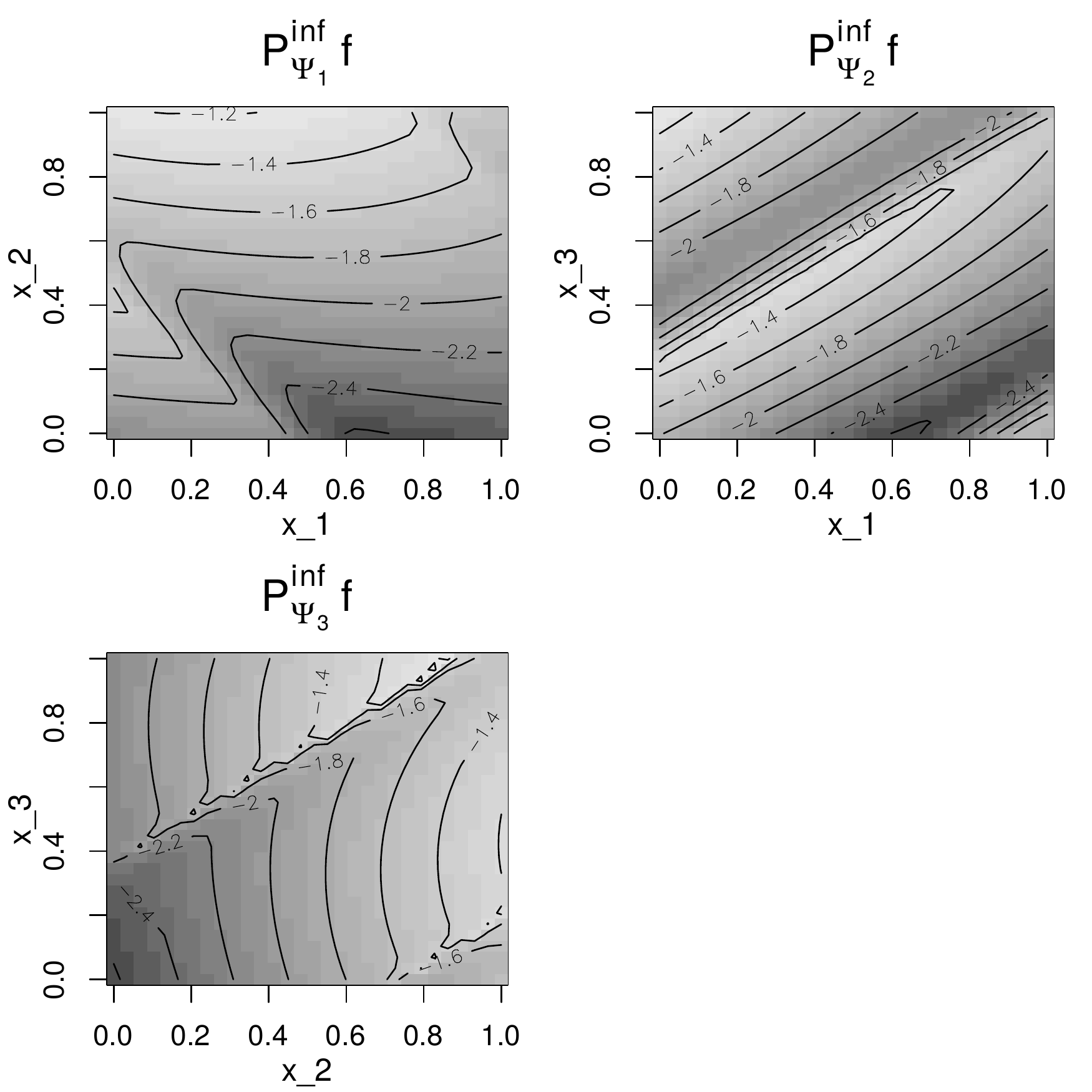}
		\caption{Bivariate profile inf for the function in Equation~\eqref{eq:analExample3d} along projections on combinations of three canonical axes.}
		\label{fig:ex3dBivCoord_inf}
	\end{minipage}
\end{figure} 
Figure~\ref{fig:ex3dBivCoord_sup} in particular highlights darker shaded (red) regions that are not part of the excursion set. Consider $\Psup_{\Psi_1}\gamma$, for all $x_1,x_2$ in the shaded region the profile tells us that there is no excursion. For example if $x_1 \in [0.115,1]$ and $x_2 = 0.2$, the segment highlighted in the top left plot of Figure~\ref{fig:ex3dBivCoord_sup}, there is no excursion. In this case, $\Pinf_\Psi f$, Figure~\ref{fig:ex3dBivCoord_inf}, does not allow us to select any excursion region.

\end{document}